\documentclass[11pt]{article}

\usepackage{latexsym}
\usepackage{graphicx,booktabs}
\usepackage{mathptmx}

\usepackage{amsmath}
\usepackage{amsfonts}
\usepackage{amssymb}
\usepackage{amsbsy}
\usepackage{amsthm}

\usepackage{epsfig,secdot,url,exscale,color,fancyhdr,epsf}
\usepackage[psamsfonts]{euscript}
\usepackage{txfonts}
\usepackage{textcomp}

\usepackage{caption}
\usepackage{subcaption}
\usepackage{color}
\usepackage{mathtools}
\usepackage{bbm}

\def\double{\par\baselineskip=24pt}

\def\standard{\oddsidemargin=0in
              \evensidemargin=0in
              \topmargin =-.5in
              \textheight=9.0in
              \textwidth=6.5in}

{\end{list}}
\newenvironment{hangref}{\begin{list}{}{\setlength{\itemsep}{0pt}%
\setlength{\parsep}{0pt}\setlength{\leftmargin}{+\parindent}%
\setlength{\itemindent}{-\parindent}}}{\end{list}}
\def\E{{\rm E}}
\def\Var{{\rm Var}}

\def\b1{{\bf 1}}

\def\blot{\quad {$\vcenter{\vbox{\hrule height.4pt
             \hbox{\vrule width.4pt height.9ex \kern.9ex \vrule width.4pt}
             \hrule height.4pt}}$}}

\standard

\newtheorem{lemma}{Lemma}
\newtheorem{corollary}{Corollary}
\newtheorem{assumption}{Assumption}

\newtheorem{remark}{Remark}

\def\d{\delta}
\def\b{k}    

\def\mui{\mu_{i}}

\def\hd1ir{\hat{\delta}_{1i}(n)}

\def\bXir{\bar{X}_i(n)}

\def\Xij{X_{ij}}

\def\Xlj{X_{\ell j}}

\newcommand{\mybold}[1]{\boldsymbol{#1}}
\newcommand{\BEAS}{\begin{eqnarray*}}
\newcommand{\EEAS}{\end{eqnarray*}}
\newcommand{\R}{\mathbb{R}}
\newcommand{\Prob}{\Pr} 

\newcommand{\id}{{\text{id}}}
\newcommand{\IID}{\mbox{iid\ }}

\newcommand{\widesim}[2][1.5]{
  \mathrel{\overset{#2}{\scalebox{#1}[1]{$\sim$}}}
}

\def\vX{\boldsymbol{X}}
\def\vW{\boldsymbol{W}}

\newcommand{\comment}[1]{}

\begin{document}

\title{Efficient Fully Sequential Indifference-Zone Procedures Using Properties of Multidimensional Brownian Motion Exiting a Sphere}
\author{A.B.~Dieker \\
Columbia University \\
\vspace{.2in}
New York, NY 10027  \\
Seong-Hee Kim\\
Georgia Institute of Technology \\
Atlanta, GA 30332-0205}
\maketitle

\begin{abstract}
We consider a ranking and selection (R\&S) problem with the goal to
select a system with the largest or smallest expected performance
measure among a number of simulated systems with a pre-specified
probability of correct selection. Fully sequential procedures take one
observation from each survived system and eliminate inferior systems
when there is clear statistical evidence that they are inferior. Most
fully sequential procedures make elimination decisions based on sample
performances of each possible pair of survived systems and exploit the
bound crossing properties of a univariate Brownian motion. In this
paper, we present new fully sequential procedures with elimination
decisions that are based on sample performances of all competing
systems.
Using properties of a multidimensional Brownian motion exiting a sphere,
we derive heuristics that aim to achieve a given target probability of correct selection.
We show that in practice the new procedures significantly outperform a
widely used fully sequential procedure. Compared to BIZ, a recent fully-sequential procedure that uses
 statistics inspired by Bayes posterior probabilities, our procedures have better
performance under difficult mean or variance configurations but similar performance under easy mean configurations.

\vspace{12pt}

\noindent {\em Subject classification:} Simulation, Ranking and
Selection, Fully Sequential, Multidimensional Brownian Motion, Sphere

\end{abstract}

\double

\section{Introduction}
\label{sec:intro}

Ranking and selection (R\&S) is one of the classical and well-studied problems in the operations research literature. It aims to find the best system among a number of systems for which noisy performance information is accessible through simulation. In this paper, we assume that the best system is one with the largest or smallest expected performance, which is known as the finding-the-best problem.
There are at least three approaches for the finding-the-best problem: the indifference-zone (IZ) approach, the Bayesian approach, and the optimal computing budget allocation (OCBA) approach. Hong et al.\ (2014) and Kim and Nelson (2011) provide a brief review of each approach. For more information on the Bayesian approach, see Chick (2006) and Chen et al.\ (2014). When there is a fixed computing budget until a decision is made, the OCBA approach provides an efficient way to find the best system, see for example Chen and Lee (2011).
In this paper we study an indifference-zone (IZ) procedure, where the decision maker specifies a difference worth detecting called the IZ parameter.

Among procedures that take the IZ approach, Rinott (1978) is one of the earliest procedures. It is a two-stage procedure and does not have any elimination step for clearly inferior systems. Nelson et al.\ (2001) also propose a two-stage procedure but their procedures can eliminate systems after the first stage if there is statistical evidence that they are inferior.  Therefore the latter procedure is more efficient than Rinott's procedure in terms of the number of observations needed until a decision is made. On the other hand, fully-sequential IZ procedures take one observation from competing systems and eliminate inferior systems as additional observations become available. They carry the risk of incorrectly eliminating the best system due to stochastic noise in the performance measurements.
Examples of fully-sequential IZ procedures are the KN procedures from Kim and Nelson (2001), which are widely used as they are available in leading commercial simulation software.
KN's parameters are chosen to control the probability of eliminating the best system.
Since this probability is intractable, the procedures instead rely on a Bonferroni-type lower bound on the worst-case probability of incorrect selection, which corresponds to the best system having a mean performance that exceeds the means of the other systems by exactly the IZ parameter;
this setup is known as the slippage configuration (SC).
Particularly when the number of systems is large, this lower bound tends to be a poor approximation for the worst-case probability of correct selection.
As discussed in Wang and Kim (2011), the result is that KN procedures tend to take many more observations than necessary to control the probability of incorrect selection, and are thus inefficient in that sense.

\begin{figure}[th!]
\begin{subfigure}{.33\textwidth}
  \centering
  \includegraphics[width=\linewidth]{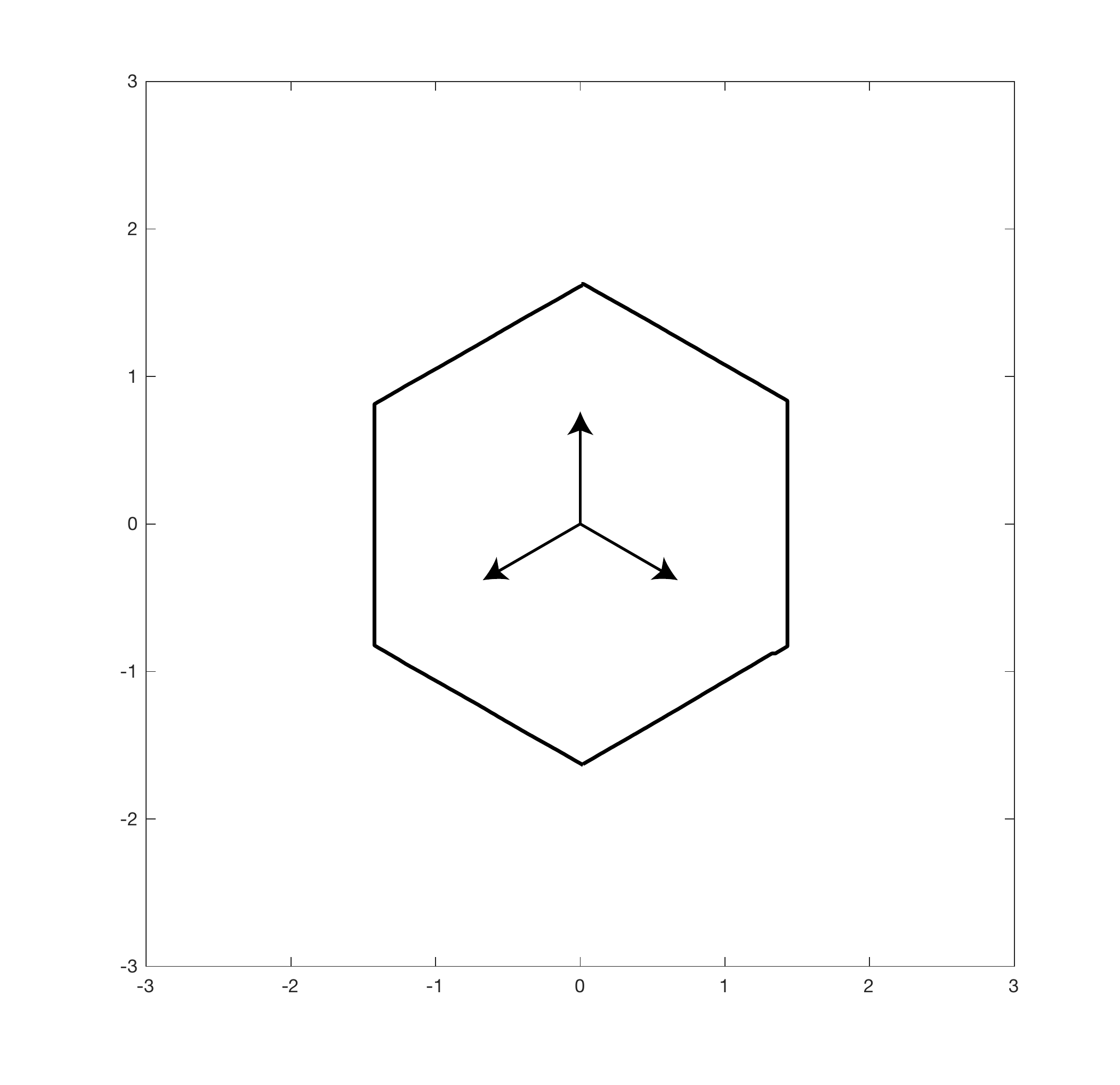}
\caption{KN; $\min_{i<j} |x_i-x_j|$}
\end{subfigure}%
\begin{subfigure}{.33\textwidth}
  \centering
  \includegraphics[width=\linewidth]{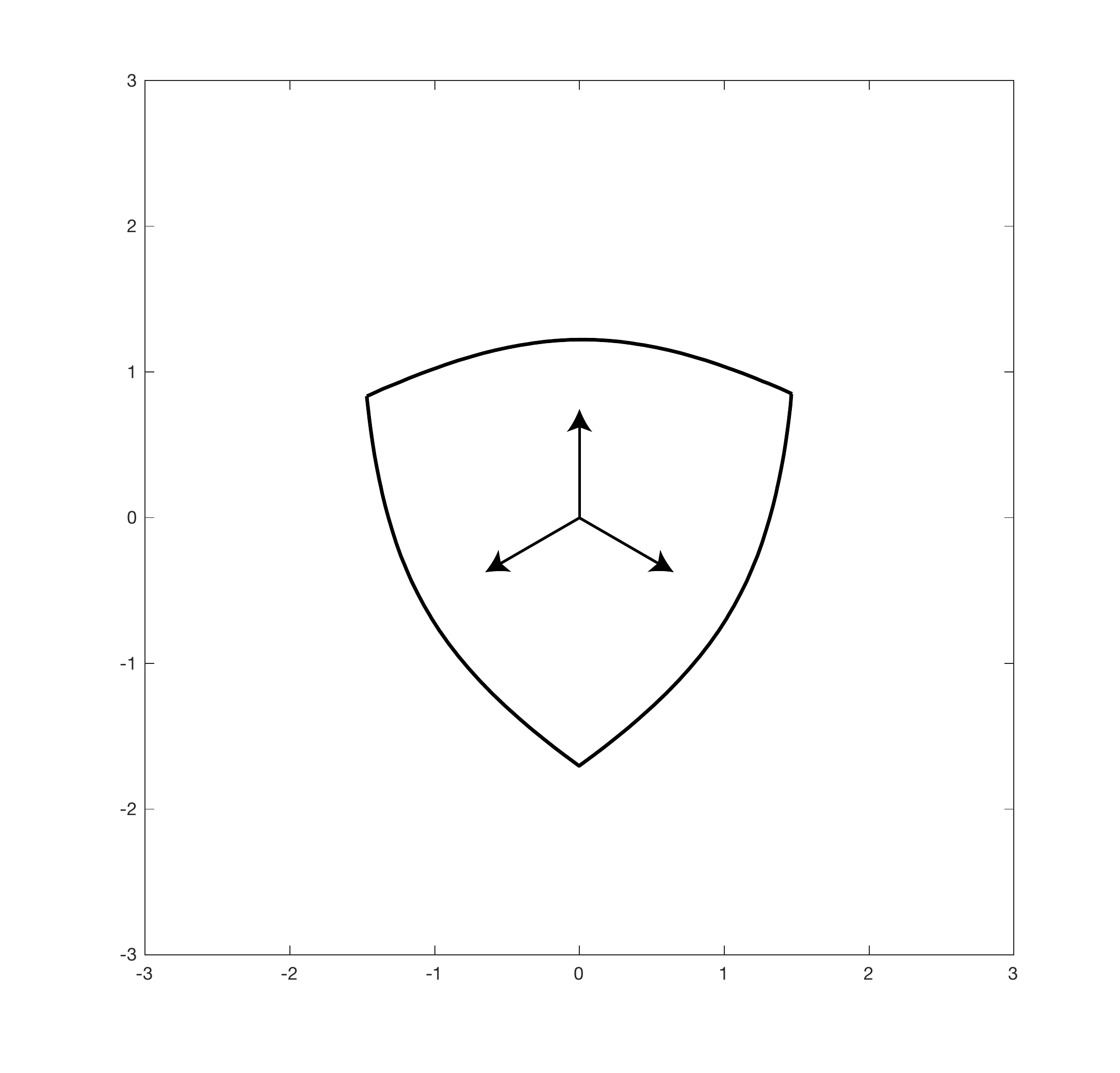}
\caption{BIZ; $\min_i e^{x_i}/(e^{x_1}+e^{x_2}+e^{x_3})$}
\end{subfigure}%
\begin{subfigure}{.33\textwidth}
  \centering
  \includegraphics[width=\linewidth]{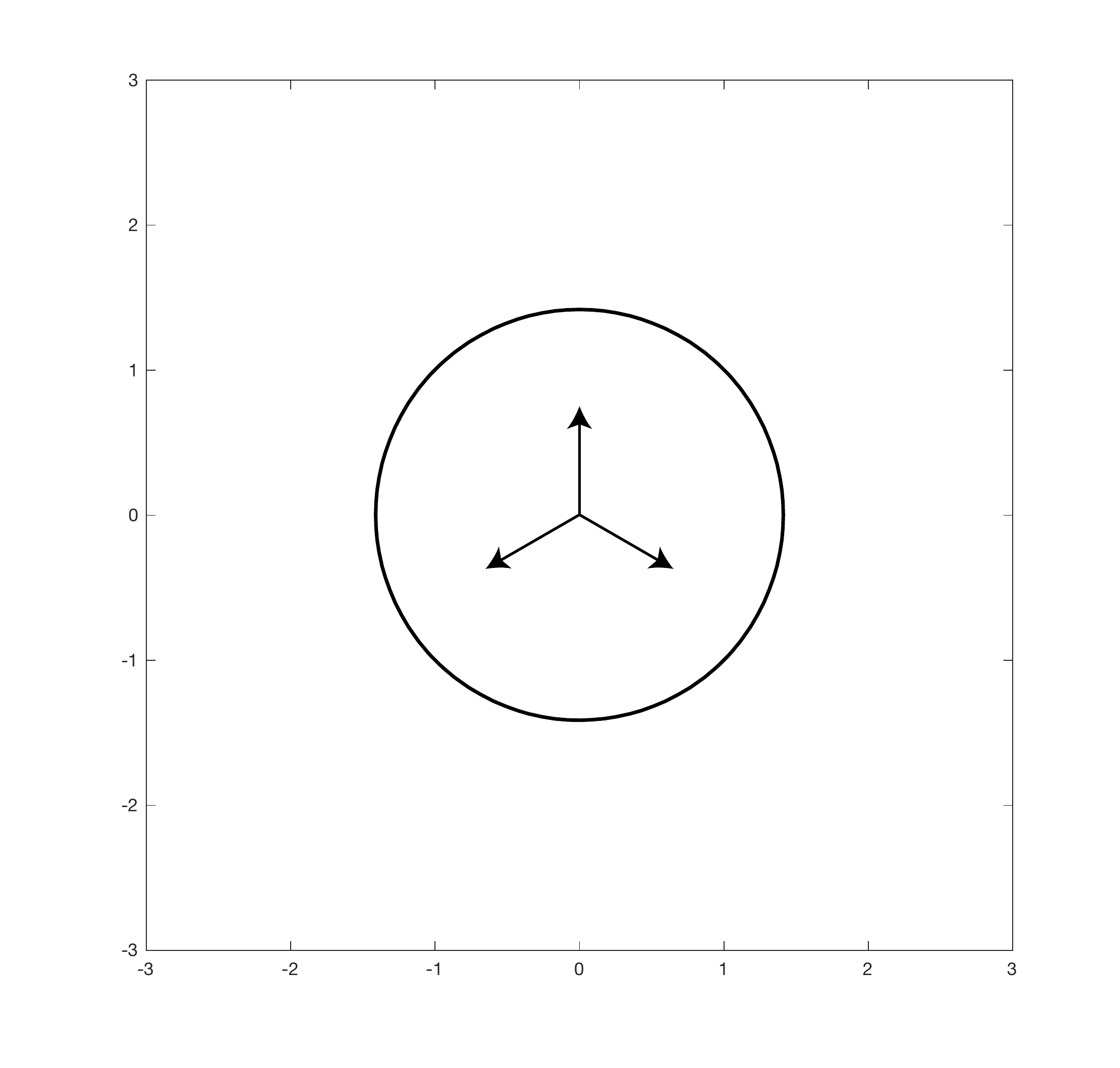}
\caption{This paper; $\sum_{i<j} (x_i-x_j)^2$}
\end{subfigure}
\caption{
Contours of the screening statistics with $k=3$ competing systems for three procedures.
The form of the screening statistic is given below each figure.
Also depicted are each of the three possible drifts (mean sample paths) under the slippage configuration (SC).
Since the screening statistics do not change when adding a constant to each coordinate,
we have plotted the plane $\{x\in \R^3:x_1+x_2+x_3=0\}$ with its intersections of the contours.}
\label{fig:1}
\end{figure}

The primary contribution of this paper is to develop a new family of IZ procedures that does not suffer from the
inefficiencies caused by the use of the Bonferroni bound.
The screening statistic used in an IZ procedure gives rise to contours, and a system is eliminated when
the vector of cumulative sums of performance measurements hits such a contour, see Figure~\ref{fig:1}.
Controlling PICS is done by analyzing hitting behavior of
Brownian motion
to set the `radius' of the contour.
It is the second ingredient where the Bonferroni bound is invoked for KN procedures,
since it is analytically intractable to study properties of a Brownian motion hitting KN contours, see Figure~\ref{fig:1}(a).

An important recently developed family of IZ ranking and selection procedures is the Bayes-inspired indifference zone (BIZ)
procedures from Frazier (2014). An example of the resulting elimination contours is given in Figure~\ref{fig:1}(b).
BIZ procedures circumvent the use of the Bonferroni bound, which results in dramatic improvements
over the KN procedure especially when the number of competing systems is large.
We see in Figure~\ref{fig:1}(b) that the three possible drifts under SC (one for each of the systems being the best)
hit the elimination contour at its closest point to the origin.
Thus, possibly sample paths that deviate significantly from their mean sample paths require a larger number of observations from the
competing systems than those that are close to their mean
sample paths.

This paper is a first investigation towards efficient procedures with spherical elimination contours, see Figure~\ref{fig:1}(c).
Using path length as a proxy for the number of observations needed, with spherical elimination contours
all points on the contour are equally close to the origin, and this could perhaps lead to faster elimination.
Setting the radius of the contours given a target probability of correct selection is facilitated by some analytic results
about a multidimensional Brownian motion hitting hyperspheres.
Like BIZ, we do not need to appeal to the Bonferroni bound to exert control over the probability of correct selection.
However, our procedure is ultimately heuristic since we found it intractable
to control the multi-stage probability of correct selection; we leave this as an open problem.

 Experimental results show that estimated probability of correct selection is all higher or close to the target confidence level for all cases tested including SC. Our procedures also significantly outperform KN. On the other hand, our procedures perform better than BIZ under difficult mean or variance configurations while they perform similarly compared to BIZ under easy mean configurations. More specifically, when variances are unknown and unequal across systems with a slippage mean configuration, our procedures show up to 30\% savings compared to the BIZ procedures in terms of the number of replications needed until a decision is made. Under easier scenarios where means spread out over systems unlike SC, our procedures perform similar to the BIZ procedures.

To extend our procedures to unknown and unequal variances, we use multiple tricks. These tricks are standard but render any statistically valid approach (including BIZ) heuristic. To handle unknown variances, we update variance estimates as the procedures advance. Kim and Nelson (2006) show that variance update enables a procedure to be treated as if variances are known in an appropriate limit.
To handle unequal variances we use a heuristic approach which essentially changes the sampling frequency of each system hoping to approximately equalize variances across systems.

Preliminary work related to this work is published in the Winter Simulation Conference proceedings which include Kim and Dieker (2011) and Dieker and Kim (2012, 2014). The first two papers consider only three systems with known variances. Dieker and Kim (2014) give a procedure for a general number of systems but require known and equal variances. Moreover, the spheres that play a crucial role in the procedure all have the same radius and the procedure performs worse than KN when the means of the systems are spread out evenly. In the procedures presented in the present paper, the radii of the spheres vary as the number of survived systems decreases, outperforming KN in all scenarios; and a version of our procedure can handle unknown and unequal variances.

When there exists a finite simulation budget or a tight deadline in time, OCBA and Bayesian procedures are shown to be highly efficient and very useful in practice. Branke et al.\ (2007), Chen and Lee (2011) and Powell and Ryzhov (2012) provide a good review of OCBA and Bayesian ranking and selection procedures and provide extensive empirical results.  As our primary goal is to investigate the impact of different shapes of continuation regions in IZ ranking and selection procedures,
we only compare our procedures with IZ procedures. Specifically, two state-of-art IZ procedures, KN and BIZ, are considered.

The paper is organized as follows. Section~\ref{sec:notation} defines our problem and introduces notation. Section~\ref{sec:procedure} proposes new fully-sequential procedures. Section~\ref{sec:stat} explains the statistics that we use for elimination decisions and the properties of our statistics.  In Section~\ref{sec:approx}, we provide justifications for our procedures and approximations in order to set the parameter values of the procedures.
Experimental results are presented in Section~\ref{sec:exp}, followed by conclusions in Section~\ref{sec:conclusion}.

\section{Problem and Notation}
\label{sec:notation}
This section introduces our notation and assumptions and defines the problem.
We assume there are $k$ systems ($k \ge 2$).
Let $X_{ij}$ represent the $j$th observation from system $i$ for $i=1,\ldots,k$ and $j=1,2,\ldots$. Then the mean and variance of the outputs from system $i$ are defined as $\mui=\E[X_{ij}]$ and $\sigma_i^2 = \Var[X_{ij}]$, respectively.
We want to find the system with the largest mean $\mu_i$.

Throughout the paper, we assume that the following assumptions hold:

\begin{assumption} \label{assump:normal}
\[
  X_{ij} \widesim[2]{\IID}
  \mybold{N}(\mu_i, \sigma_i^2), \;\;\;\; j=1,2,\ldots,
\]
where $\widesim[2]{\IID}$ represents `are
independent and identically distributed as' and $\mybold{N}(\mu_i, \sigma_i^2)$ denotes the
normal distribution with mean $\mu_i$ and
variance $\sigma_i^2$. Moreover, $X_{ij}$ and $X_{i^\prime j}$  are independent for any $i \neq i^\prime$ and $j=1,2,\ldots$.
\end{assumption}

\vspace{12pt}
\begin{assumption}\label{assump:sc}
$\mu_1 \leq \mu_2 \leq \ldots \leq \mu_{k-1} \leq \mu_k - \d$ for $\d \in \mathbb{R}^+$.
\end{assumption}
\vspace{12pt}

Assumption~\ref{assump:normal} implies that observations from each
system are marginally \IID normally distributed and systems are simulated independently (note that this rules out common random numbers).
Without loss of generality, we assume that system $k$ is the true best system. Assumption~\ref{assump:sc} assumes that the mean of the true best system $\b$ is at least $\delta$ better than any alternative system. The parameter $\d$ is a user-specified parameter known as the IZ parameter.

We aim to devise a method that observes systems sequentially and eliminates clearly inferior systems from further consideration. The method stops once only one system remains, and this system
is declared as the best system. 


Additional notation is needed for later sections:
 \begin{eqnarray*}
     n & \equiv & \mbox{the current number of observations or the current stage number};\\
      I & \equiv & \mbox{set of competing systems at the $n$th stage};\\
      \bXir  &\equiv&  \frac{1}{n}\sum_{j=1}^n{\Xij}\mbox{, the sample mean of system $i$ based on the first $n$ observations};\\
        \vX_I(n) &\equiv& \mbox{$ |I| \times 1$ vector of $\sum_{j=1}^n X_{ij}$ for $i \in I$};\\
        \hat{\sigma}^2_i(n) &\equiv& \mbox{sample variance of system $i$ from $X_{i1}, \ldots, X_{in}$} \mbox{ which is } {1 \over n-1} \sum_{j=1}^n (\Xij - \bXir)^2;\\
        A^T & \equiv& \mbox{the transpose of a matrix $A$};\\
        \d_{|I|}^2 &\equiv&  {\d^2} {|I|-1 \over |I|}.\\
\end{eqnarray*}


\section{${\cal DK}$ Procedures}
\label{sec:procedure}

In this section, we provide the descriptions of our new procedures. We present ${\cal DK}_1$ for known and equal variances and extend it to unknown but equal variances, resulting in ${\cal DK}_2$.  Then ${\cal DK}_3$ is presented for unknown and unequal variances.

\subsection{Equal and Known Variances}

We first consider a case where variances are equal across all systems and known so that $\sigma_i^2 = \sigma^2$ for any system $i$.
Suppose $x\in\R^s$ and $I \subset \{1,\ldots, k\}$ and define a function ${\cal S}_I(x)$ as follows:
\[
{\cal S}_I(x) =  {1 \over \sigma^2} \sum_{i\in I} (x_i - \bar{x})^2
\]
where $\bar{x} = {1 \over s} \sum_{i \in I} x_i$.

The ${\cal DK}_1$ procedure for equal and known variances is as follows:
\begin{center}
\fbox{
\begin{minipage}{6.5in}
\begin{center}\textbf{The ${\cal DK}_1$ Procedure}\end{center}
\begin{hangref}
\item {\bf Setup:}  Select the nominal level $1-\alpha$ and the IZ
    parameter $\d$. Set $I = \{1, 2, \ldots, k\}$  and choose $\eta_{|I|}$ (which will be discussed in Section~\ref{sec:approx}). Take one observation from each
    system. Set $n=1$ and go to {\bf Calculation}.

 \vspace{0.1in}
\item {\bf Calculation:} 
Calculate ${\cal S}_I(\vX_I(n))$.

 \vspace{0.1in}
\item {\bf Screening:} If ${\cal S}_I(\vX_I(n)) \ge \left({ \sigma \cdot \eta_{|I|} \over \delta_{|I|}}\right)^2$, then eliminate the system with the smallest $\bar{X}_{i}(n)$ among $i\in I$. Update $I$ by removing the eliminated system and go back to {\bf Calculation}.
    Otherwise, go to {\bf Stopping Rule}.

     \vspace{0.1in}
\item {\bf Stopping Rule:} If $|I| = 1$, stop and declare the surviving system as the best. Otherwise, take one more observation for all $i \in I$, set $n=n+1$, and go to {\bf Calculation}.
\end{hangref}
\end{minipage}
}

\end{center}
\vspace{0.25in}

In Section~\ref{sec:stat}, we show that ${\cal S}_I(x)$ calculates the squared distance of the point $x$ orthogonally projected onto a hyperplane $\{x:\sum_{i\in I} x_i = 0\}$ and that the screening rule in ${\cal DK}_1$ implies that we have an open (infinite) cylinder as our continuation region with a radius depending on $\eta_{|I|}$ and $\delta_{|I|}$. When $x$ is located outside the cylinder, elimination occurs and the screening rule is checked again with updated parameters (without obtaining additional observations), i.e., a lower dimensional cylinder. We only obtain new observations (i.e., move to the next stage) if no more elimination occurs for a given number of observations.


\subsection{Unknown but Equal Variances}

We present a straightforward variant of ${\cal DK}_1$ for unknown but equal variances, $\sigma^2$.
As the variance parameter $\sigma^2$ is unknown, it needs to be estimated. Let $\hat{\sigma}_i^2(n)$ represent sample variance of system $i$.
The pooled variance estimator $\hat{\sigma}_p^2(n)$ is defined as follows:
\[
\hat{\sigma}^2_p(n) = {1 \over |I|} \sum_{i \in I} \hat{\sigma}^2_i(n).
\]
Then our statistic is modified to
\[
{\cal S}_I'(x)  =    {1 \over \hat{\sigma}^2_p(n)} \sum_{i \in I} (x_i - \bar{x})^2
\]
and $\hat{\sigma}^2_i(n)$ and $\hat{\sigma}^2_p(n)$ need to be updated in the [Stopping Rule] step after additional observations are obtained. Then the ${\cal DK}_2$ procedure is defined below.

\vspace{0.2in}
\begin{center}
\fbox{
\begin{minipage}{6.5in}
\begin{center}\textbf{The ${\cal DK}_2$ Procedure}\end{center}
\begin{hangref}
\item {\bf Setup:}  Select the nominal level $1-\alpha$ and the IZ
    parameter $\d$. Set $I = \{1, 2, \ldots, k\}$  and choose $\eta_{|I|}$. Take $n_0 \ge 2$ observations from each
    system and calculate $\hat{\sigma}^2_i(n_0)$ and $\hat{\sigma}^2_p(n_0)$. Set $n=n_0$ and go to {\bf Calculation}.

 \vspace{0.1in}
\item {\bf Calculation:} 
Calculate ${\cal S}_I'(\vX_I(n))$.

 \vspace{0.1in}
\item {\bf Screening:} If ${\cal S}_I'(\vX_I(n)) \ge  \left({\hat{\sigma}_p(n)  \cdot \eta_{|I|} \over \delta_{|I|}}\right)^2$, then eliminate the system with the smallest $\bar{X}_{i}(n)$ among $i\in I$. Update $I$ by removing the eliminated system and go back to {\bf Calculation}.
    Otherwise, go to {\bf Stopping Rule}.

     \vspace{0.1in}
\item {\bf Stopping Rule:} If $|I| = 1$, stop and declare the surviving system as the best. Otherwise, take one more observation for all $i \in I$; set $n=n+1$; and update $\hat{\sigma}^2_i(n)$ for $i\in I$ and $\hat{\sigma}^2_p(n)$. Then go to {\bf Calculation}.
\end{hangref}
\end{minipage}
}

\end{center}
\vspace{0.25in}

\subsection{Unknown and Unequal Variances}

This subsection extends the ${\cal DK}_1$ procedure to handle unknown and unequal variances, resulting in ${\cal DK}_3$. The main idea is to make the sampling frequency of each system proportional to the variance parameter of the system, which eventually leads to equal variances. This approach is similar to the one in Frazier (2014).

Let $n_i$ denote the number of observations system $i$ have received so far. In ${\cal DK}_1$, $n_i = n$ for any system $i \in I$ but in ${\cal DK}_3$, $n_i \le n$.
Also let $W_i(n) = \sum_{j=1}^{n_i} \Xij/n_i$ and $\vW_I(n)$ represent a $|I| \times 1$ vector of $W_i(n)$ for $i\in I$.

Then
\[
{\cal S}_I''(x)  =  {1 \over \hat{\lambda}^2}   \sum_{i \in I} \left( x_i - \bar{x}  \right)^2
\]
where
\[
\hat{\lambda}^2 =  {\sum_{i\in I} \hat{\sigma}^2_i(n_i) \over \sum_{i\in I} n_i}.
\]

We can now describe Procedure ${\cal DK}_3$.
\vspace{0.2in}
\begin{center}
\fbox{
\begin{minipage}{6.5in}
\begin{center}\textbf{The ${\cal DK}_3$ Procedure}\end{center}
\begin{hangref}
\item {\bf Setup:}  Select the nominal level $1-\alpha$ and the IZ
    parameter $\d$. Also select a constant $B_z$. Set $I = \{1, 2, \ldots, k\}$  and choose $\eta_{|I|}$. Take $n_0$ observations from each
    system and calculate $W_i(n_0)$, $\hat{\sigma}^2_i(n_0)$ and $\hat{\lambda}^2$.  Set $n=n_0$ and $n_i=n_0$ for $i\in I$, and go to {\bf Calculation}.

 \vspace{0.1in}
\item {\bf Calculation:} 
Calculate ${\cal S}_I''(\vW_I(n))$.

 \vspace{0.1in}
\item {\bf Screening:} If ${\cal S}_I''(\vW_I(n)) \ge  \left({\hat{\lambda} \cdot \eta_{|I|} \over \delta_{|I|}}\right)^2$, then eliminate the system with the smallest $\bar{X}_{i}(n)$ among $i\in I$. Update $I$ by removing the eliminated system and go back to {\bf Calculation}.
    Otherwise, go to {\bf Stopping Rule}.

     \vspace{0.1in}
\item {\bf Stopping Rule:} If $|I| = 1$, stop and declare the surviving system as the best. Otherwise,
   let $z = \mbox{arg}\min {n_i \over \hat{\sigma}^2_i(n_i)}$ for $i\in I$.

    For each $i \in I$,
    \begin{itemize}
    \item calculate
    \[
    \Delta_i = \left\lceil { \hat{\sigma}^2_i(n_i) \cdot { n_z + B_z \over \hat{\sigma}^2_z(n_z)}} \right\rceil;
    \]
    \item if $\Delta_i > n_i$, then take $(\Delta_i - n_i)$ observations.
    \end{itemize}

    Set $n=n+1$ and $n_i =\max(n_i, \Delta_i)$; and  update $\hat{\sigma}^2_i(n_i)$ for all $i\in I$ and $\hat{\lambda}^2$. Then go to {\bf Calculation}.
\end{hangref}
\end{minipage}
}

\end{center}
\vspace{0.25in}

Frazier (2014) recommends $B_z= 1$. The parameter $\eta_{|I|}$ needs to be chosen carefully so that the actual probability of correct selection is at least $1-\alpha$. In the next section, we derive some analytical results for the ${\cal DK}_1$ procedure and then discuss how to choose $\eta_{|I|}$.

\section{Statistics for Screening}
\label{sec:stat}
The canonical choice for fully sequential procedures is to use $\sum_{j=1}^n (\Xij - \Xlj)$ for every $i\neq \ell$ as observed statistics and to eliminate a system whenever the statistics exit a so-called continuation region defined by two parallel lines such as $(-a, a)$ for a constant $a>0$ or a function $h(n)>0$ such as $(-h(n), h(n))$. Kim and Nelson (2014) use a triangular shaped continuation region defined by a decreasing linear function $h(n)$. Note that traditional continuation regions are defined in a two-dimensional space. Our procedures use different statistics based on a quadratic form and our continuation region is an open cylinder.

Consider $x\in\R^s$ and $I \subset \{1,\ldots, k\}$ with $I=\{i_1,\ldots,i_s\}$.
Furthermore let $\Gamma$ represent the covariance matrix of $(X_{i_1 j}, X_{i_2 j}, \ldots, X_{i_s j})^T$,
\[
 \Gamma =  \begin{bmatrix}
\sigma_{i_1}^2 & 0 & 0 & \cdots & 0 \\
0 & \sigma_{i_2}^2 & 0 & \cdots  & 0\\
0 & 0 & \sigma_{i_3}^2 & \cdots &0\\
\vdots & \vdots & \vdots & \ddots & \vdots\\
0 & \cdots & \cdots & \cdots & \sigma_{i_s}^2\\
\end{bmatrix}
\]
and let $V$ represent an $s-1$ by $s$ matrix given by
\[
V = \begin{bmatrix}
1 & 0 & \cdots& 0 & -1\\
0 & 1 & \cdots& 0 & -1\\
\vdots & \vdots & \ddots & \vdots & \vdots\\
0 & 0 & \cdots &  1 & -1\\
\end{bmatrix}.
\]

Then our statistic ${\cal S}_I(x)$ is defined as
\begin{equation}
\label{eq:defSI}
{\cal S}_I(x)  \equiv (Vx)^T(V \Gamma V^T)^{-1}(Vx)=\left[
\begin{array}{c}
x_{i_1} - x_{i_s} \\
\vdots\\
x_{i_{s-1}} - x_{i_s}
\end{array}\right]^T (V \Gamma V^T)^{-1} \left[
\begin{array}{c}
x_{i_1} - x_{i_s} \\
\vdots\\
x_{i_{s-1}} - x_{i_s}
\end{array}\right]
\end{equation}
and our continuation region is related to this quadratic form.
From the definition of ${\cal S}_I$ it may seem that ${\cal S}_I$ is complicated to calculate and that it depends on the order in which its elements are listed.
The following lemma is useful in deriving a simpler form of ${\cal S}_I(x)$ which allows us to argue that ${\cal S}_I(x)$ only depends on the set $I$, so not on the order of the elements in $I$. The proof is given in the appendix.

\begin{lemma}
\label{lem:S}
Suppose $x\in\R^s$ and $I \subset \{1,\ldots, k\}$ with $I=\{i_1,\ldots,i_s\}$.
If $\Pi =  \Gamma V^T (V \Gamma V^T)^{-1} V$,  then
\[
{\cal S}_I(x)={\cal S}_I(\Pi x).
\]
\end{lemma}

\begin{figure}[th!]
\begin{subfigure}{.5\textwidth}
  \centering
  \includegraphics[width=\linewidth]{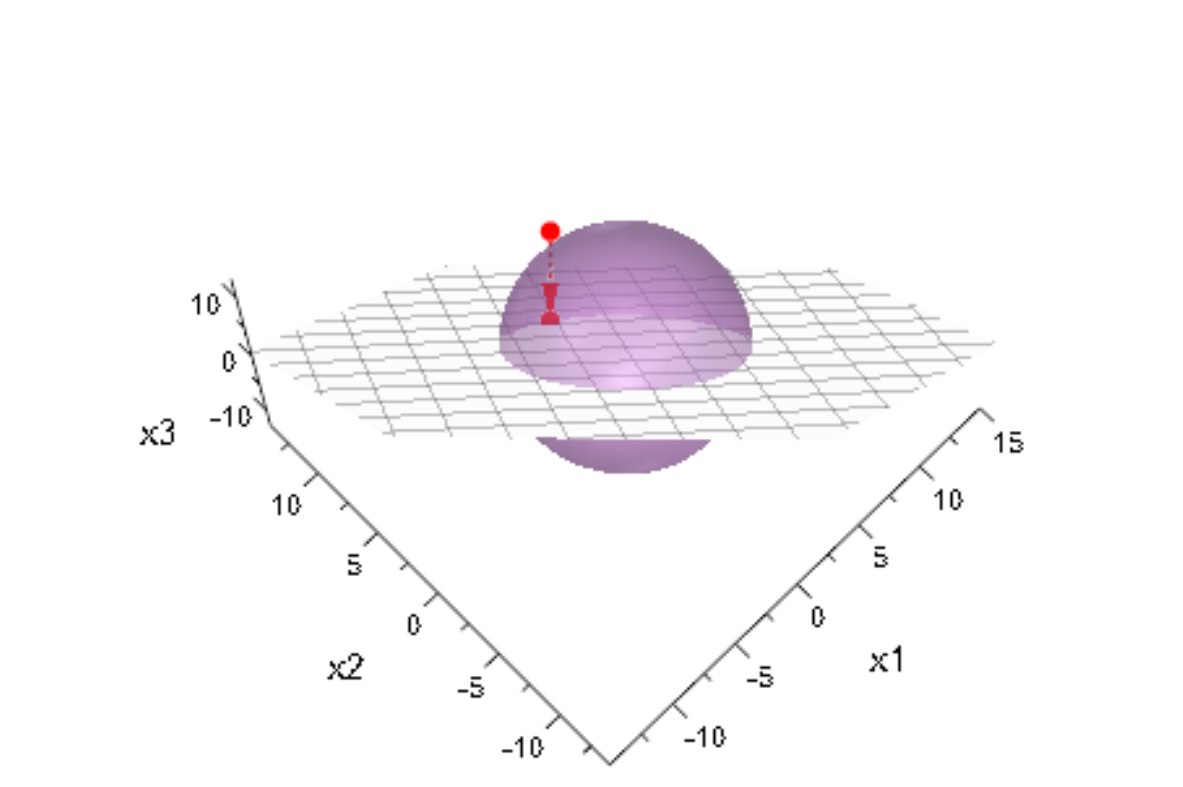}
\caption{Projection inside ball: no elimination}
\end{subfigure}%
\begin{subfigure}{.5\textwidth}
  \centering
  \includegraphics[width=\linewidth]{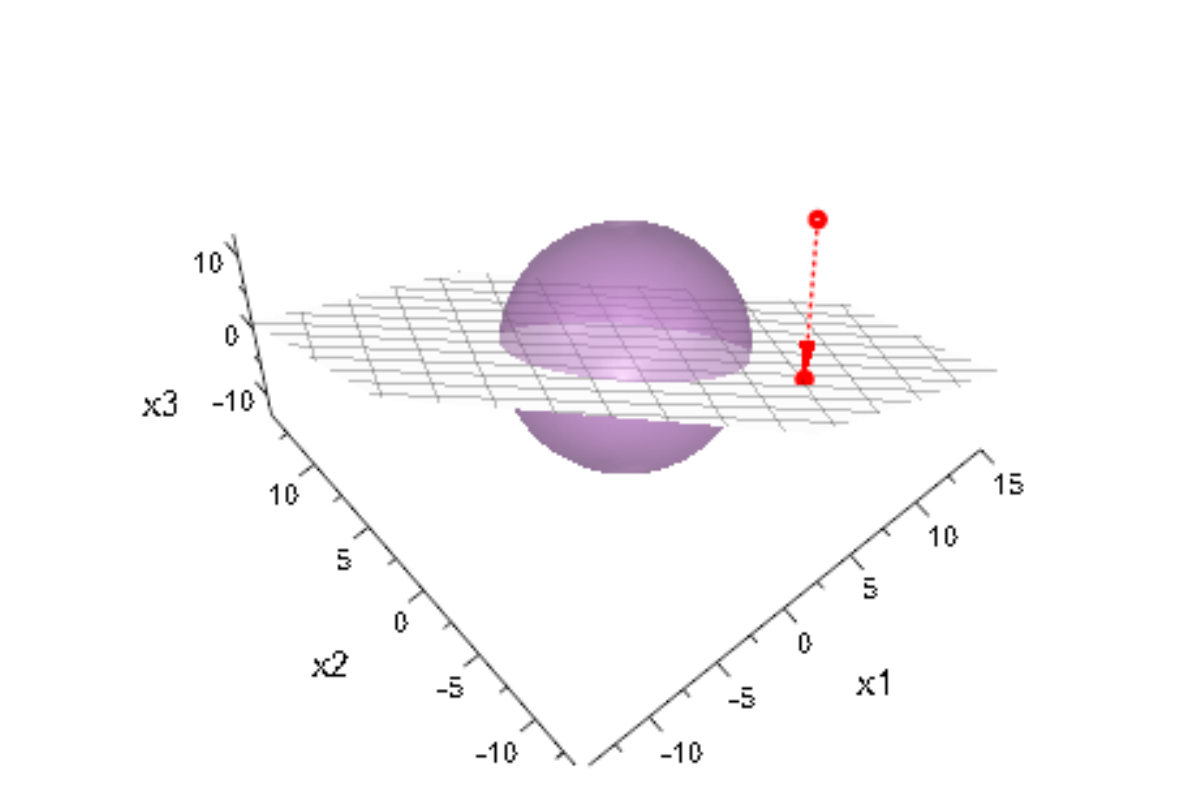}
\caption{Projection outside ball: elimination}
\end{subfigure}
\caption{Projected points on plane $y_1 + y_2 + y_3 = 0$ and elimination rules. A ball (with radius 7 here) is also visible.}
\label{fig:3d}
\end{figure}

The above lemma holds for $\Gamma$ regardless of whether it has equal diagonal elements. The matrix $\Pi$ is a (non-orthogonal) projection matrix with range $R = \{ y \in \mathbb{R}^s: \sum_{i\in I} y_i/\sigma_i^2 = 0 \}$ and null space $N = \{ \alpha (1, \ldots, 1): \alpha \in \mathbb{R} \}$, i.e., when $\Pi$ is applied to a vector then a multiple of $(1 , \ldots, 1)$ is subtracted from this vector so the result lies in $R$. It becomes an orthogonal projection matrix when $\sigma_i^2 = \sigma^2$ for all $i \in I$, since the null space $N$ is orthogonal to the range $R$ in that case. This lemma implies that the value of our statistic at any $x$ equals the value of our statistic at the projected point on the plane determined by $\sum_{i\in I} y_i/\sigma_i^2=0$ (or $\sum_{i\in I} y_i=0$ for equal variances). Since the null space and range do not change if the order of the elements in $I$ changes, Lemma~\ref{lem:S} shows that the quadratic form $\cal S_I$ remains the same if the indices are ordered differently, i.e., both in $\Gamma$ and in $x$. In Section~\ref{sec:approx} this lemma is used to make the elimination decision depend on the IZ parameter $\delta$ only and not on the unknown mean parameter. Using the above lemma, we next derive a simpler form of ${\cal S}_I(x)$ when the variances are equal.
As an aside, it may be tempting to think that ${\cal S}_I(x) = x^T V^T(V^T)^{-1}\Gamma^{-1} V^{-1}Vx=x^T\Gamma^{-1}x$  in view of (\ref{eq:defSI}),
but this is incorrect since $V$ is not invertible.

\begin{corollary}
\label{cor:S}
Suppose $x\in\R^s$ and $I \subset \{1,\ldots, k\}$ with $I=\{i_1,\ldots,i_s\}$.
If $\sigma_i^2 = \sigma^2$, then
\[
{\cal S}_I(x) = \frac1{\sigma^2}\frac 1{|I|} \sum_{i < \ell \atop i, \ell \in I} (x_{i} - x_{\ell})^2
= {1 \over \sigma^2} \sum_{i\in I} (x_i - \bar{x})^2
\]
where $\bar{x} = {1 \over s} \sum_{i \in I} x_i$.
\end{corollary}

Our elimination decision rule takes the form ${\cal S}_I(x) \ge r^2$ for $r\in \mathbb{R}^+$. Let $x'$ denote the orthogonal projection of $x$ on the plane with $\sum_{i=1}^s y_i=0$. From Lemma~\ref{lem:S}, we know that ${\cal S}_I(x)$ is equal to ${\cal S}_I(x')$. As $x'$ lies on the hyperplane $\{y:\sum_{i=1}^s y_i=0\}$, we know that $\bar{x'} = 0$. From the second equality of ${\cal S}_I(x)$ in Corollary~\ref{cor:S}, it is easy to see that ${\cal S}_I(x')$ becomes simply the squared distance between $x'$ and the origin.  Therefore our elimination decision rule implies that no elimination occurs and sampling continues when the projected point $x'$ is inside a sphere as in Figure~\ref{fig:3d}(a); but one system with the smallest value is eliminated when the projected point $x'$ is outside the sphere as in Figure~\ref{fig:3d}(b).

One may wonder why our elimination rule only considers the largest set $I$ but not any subset $J$, thinking that the screening statistics ${\cal S}_J(x_J)$ for $J \subseteq I$ may be larger than ${\cal S}_I(x_I)$. This would mean that even if an elimination does not occur with set $I$, an elimination might be possible for a subset $J$. However, the following lemma implies that our elimination rule for the largest set $I$  actually verifies elimination for all $2^{|I|}-1$ nonempty subsets $J\subseteq I$ by showing that we always get the largest screening statistics with set $I$.

\begin{lemma}
\label{lem:Smon}
Suppose $J \subseteq I\subseteq \{1,\ldots,k\}$.
Then ${\cal S}_J(x_J)\le {\cal S}_I(x_I)$ for all $x\in\R^k$.
\end{lemma}

\section{Proofs and Approximations}
\label{sec:approx}
This section presents an approximation for the probability of incorrect selection
under ${\cal DK}_1$, which assumes known and equal variances $\sigma^2$.
We use these approximations in lieu of possibly conservative bounds in order to choose
the parameters $\eta_2,\ldots,\eta_k$ of ${\cal DK}_1$, thus bypassing a main source of inefficiencies.
In the course of the presentation, we explain how we choose the parameters $\eta_2,\ldots,\eta_{k}$ of our procedure.

The event of incorrect selection can be partitioned according to when the best system is eliminated.
If the best system is eliminated first, then we say that the level of elimination is $1$.
Similarly, if the second system to be eliminated is the best system, then we say that the level of elimination is $2$.
Thus, the possible levels of incorrect elimination are $1,\ldots,k-1$. The key building block for our approximation scheme is an approximation for the probability of incorrect selection at the first elimination level, which we discuss in Section~\ref{subsec:immediateelim}.
Other levels of incorrect elimination are studied in Section~\ref{subsec:otherelim}.
With this, we devise a procedure for choosing the parameter $\eta_{|I|}$ for ${\cal DK}_1$. We then explain how $\eta_2,\ldots,\eta_k$ for ${\cal DK}_1$ are related to parameters for ${\cal DK}_2$ and ${\cal DK}_3$ in Section~\ref{subsec:unknown}.

In the continuous analog of our problem, the discrete observation window is replaced with a continuous one.
The analog of the random walk $\vX_{\{1,\ldots,k\}}(n)$ is $\sigma B(t)$, where $B(t)$ is a standard
Brownian motion in $\R^k$ with drift $( \mu,\ldots,\mu,\mu+\delta) \times 1/\sigma$.
Throughout this section, we study this continuous problem as a proxy for the discrete problem, and the results we state for the
${\cal DK}_1$ algorithm are to be understood as for its continuous analog.

\begin{lemma}
\label{lem:deltasigma}
For fixed $\eta_k,\ldots,\eta_2$ and $\ell\in\{2,\ldots,k\}$, the probability of eliminination at level $\ell$ in ${\cal DK}_1$ is constant as a
function of $\delta$ and $\sigma$.
In particular, the probability of incorrect selection in ${\cal DK}_1$ does not depend on $\delta$ or $\sigma$.
\end{lemma}
\begin{proof}
Consider $\sigma B(t) + \mu \boldsymbol{1} t + \delta w t$ instead of $\vX_{\{1,\ldots,k\}}(n)$, where $B(\cdot)$ is a standard Brownian motion in $\R^N$ with $B(0)=0$, $w=(0,\ldots,0,1)$, and $\boldsymbol{1}=(1,\ldots,1)$.
Since the screening statistic uses the projection of this process on the hyperplane $\{x:\sum_i x_i=0\}$, i.e., the `sample mean' is subtracted,
we may assume without loss of generality that $\mu=0$ and replace $w$ by its projected version $v=(1/k,\ldots,1/k,-(k-1)/k)$.
Suppose that we are given some $x\in \R^N$ and an $N$-dimensional set $S$. We set
\[
\tau_S = \inf\left\{t\ge 0: \frac {\sigma^2}\delta  x + \sigma B(t)+ \delta vt \in \frac{\sigma^2}\delta S\right\},
\]
so that
\begin{eqnarray*}
\Prob(\tau_S<\infty)
&=&\Prob\left(\exists t' \ge 0: \frac {\sigma^2}\delta  x + \sigma B\left(\frac{\sigma^2}{\delta^2}t' \right)+\frac {\sigma^2} \delta  vt' \in \frac{\sigma^2}\delta S\right)\\
&=&\Prob\left(\exists t' \ge 0: \frac {\sigma^2}\delta  x +\frac{\sigma^2}{\delta}B(t')+\frac {\sigma^2} \delta  vt' \in \frac{\sigma^2}\delta S\right)\\
&=&\Prob\left(\exists t' \ge 0: x + B(t')+ vt' \in S\right),
\end{eqnarray*}
where the first equality follows from rescaling time and the second from the Brownian scaling property.
This argument extends to the hitting location, i.e.,
$\Prob(\tau_S<\infty, \frac {\sigma^2}\delta  x + \sigma B(\tau_S)+ \delta v \tau_S \in \frac {\sigma^2}\delta dy)$.
In particular, the hitting location scales with $\sigma^2/\delta$.

Elimination at level $\ell$ amounts to successively hitting appropriate regions of sets of the form
\begin{eqnarray*}
\left\{x: {\cal S}_I(x_I)\ge \frac{k \sigma^2\eta^2 }{(k-1) \delta^2}\right\}
&=& \left\{x: {1 \over \sigma^2} \sum_{i\in I} (x_i - \bar{x})^2\ge \frac{k \sigma^2\eta^2 }{(k-1) \delta^2}\right\}\\
&=& \left\{x: \sum_{i\in I} (x_i - \bar{x})^2\ge \frac{k \sigma^4\eta^2 }{(k-1) \delta^2}\right\} \\
&=& \frac{\sigma^2}{\delta} \left\{x: \sum_{i\in I} (x_i - \bar{x})^2\ge \frac{k \eta^2 }{k-1}\right\},
\end{eqnarray*}
where we used Corollary~\ref{cor:S}. Such sets are of the form $(\sigma^2/\delta) S$, and
the successive hitting locations scale with $\sigma^2/\delta$.
By the strong Markov property and the calculation in the first part of this proof,
this means that the elimination probability does not depend on $\delta$ or $\sigma$.
\end{proof}

\subsection{Immediate (Level 1) Elimination of the Best System}
\label{subsec:immediateelim}
Our approximation for the probability of eliminating system $k$ first is
based on an asymptotic analysis as the number of systems $k$ goes to infinity.
Our results use the commonly employed idea of (i) considering the slippage configuration (SC) where $\mu_1 = \cdots = \mu_{k-1} = \mu_k-\delta = \mu$ and (ii) replacing the (discrete) Gaussian observation sequence with
a (continuous) Brownian motion.

Throughout this section we use the following notation. For a given a vector $x\in\R^k$,
we define
\[
\E_k(x) = \frac 1k \sum_{i=1}^k x_i,\quad
\Var_k(x) = \frac 1k \sum_{i=1}^k x_i^2 - \E_k(x)^2.
\]

The $\cal {DK}$ algorithms require evaluating ${\cal S}_{\{1,\ldots,k\}}$ at $\vX_{\{1,\ldots,k\}}(n)$, and by
Lemma~\ref{lem:S} this equals (up to $1/\sigma^2$) the squared norm of $\vX_{\{1,\ldots,k\}}(n) - \overline{\vX_{\{1,\ldots,k\}}(n)}$,
which corresponds to $\sigma B(t)-\sigma\E_k(B(t))$. (We abuse notation and interpret subtraction of a constant as elementwise subtraction.)
The following lemma specifies the probabilistic behavior of this process,
and it is important to note that it is free of the unknown mean parameter $\mu$.
\begin{lemma}
$B(t)-\E_k(B(t))$ has drift $(-1/k,\ldots,-1/k,(1-1/k))\times \delta/\sigma$ and is a standard Brownian motion in the $(k-1)$-dimensional
hyperplane
\[
H=\left\{x\in\R^k:\sum_{i=1}^k x_i=0\right\}.
\]
\end{lemma}
\begin{proof}
The claim that $B(t)-\E_k(B(t))$ takes values in $H$ is evident.
We can write $B(1)-\E_k(B(1)) = (\id_k - \boldsymbol{1}_k \boldsymbol{1}_k^T /k) B(1)$,
where $\id_{k}$ is the $k \times k$ identity matrix and $\boldsymbol{1}_k$ is the $k \times 1$ vector of ones.
Therefore, the covariance matrix of $B(1)-\E_k(B(1))$ is
\[
(\id_k - \boldsymbol{1}_k \boldsymbol{1}_k^T /k) \times (\id_k - \boldsymbol{1}_k \boldsymbol{1}_k^T /k) = (\id_k - \boldsymbol{1}_k \boldsymbol{1}_k^T /k).
\]
This matrix has one eigenvalues 0 (with corresponding eigenvector $\boldsymbol{1}_k$) and 1 (with corresponding eigenspace $H$). Therefore it acts as the identity on $H$ and it is degenerate on the complement. \end{proof}

Setting $r ={\sigma \eta_{k} \over \delta_k}$, we define a $(k-1)$-dimensional sphere in $H$ by
\[
C=\left\{x\in \R^k: \sum_{i=1}^k x_i = 0, \|x\| = r\right\}.
\]
Elimination of the best system can be formulated as $B(t)-\E_k(B(t))$ hitting $C$ in the region
\[
E_k = \{x\in C: x_k = \min(x_1,\ldots,x_k) \}.
\]

\begin{figure}[t!]
\begin{subfigure}{.5\textwidth}
  \centering
  \includegraphics[width=\linewidth]{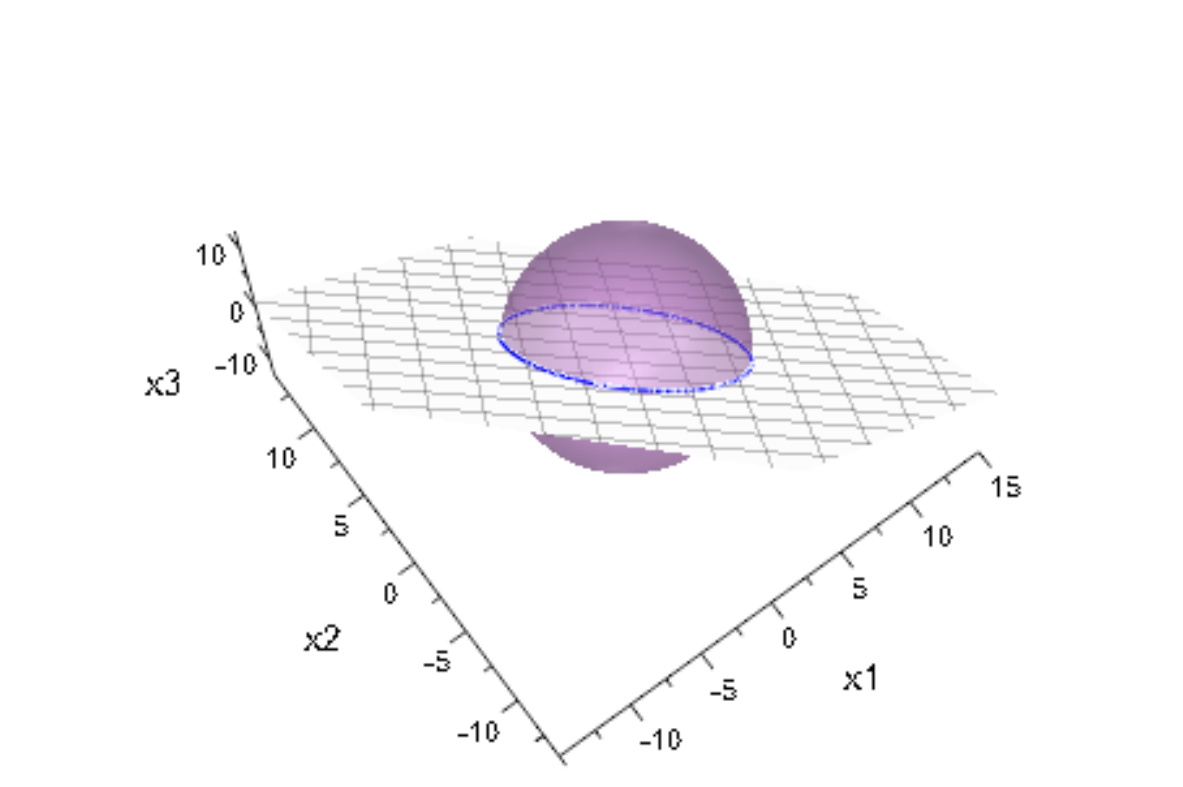}
\caption{Sphere $C$ (circle here) on hyperplane $H$}
\end{subfigure}%
\begin{subfigure}{.5\textwidth}
\centering
$C=\left\{x\in \R^k: \sum_{i=1}^k x_i = 0, \|x\| = r\right\}$
$E_k = \{x\in C: x_k = \min(x_1,\ldots,x_k) \}$
\end{subfigure}
\begin{subfigure}{.5\textwidth}
  \centering
  \includegraphics[width=\linewidth]{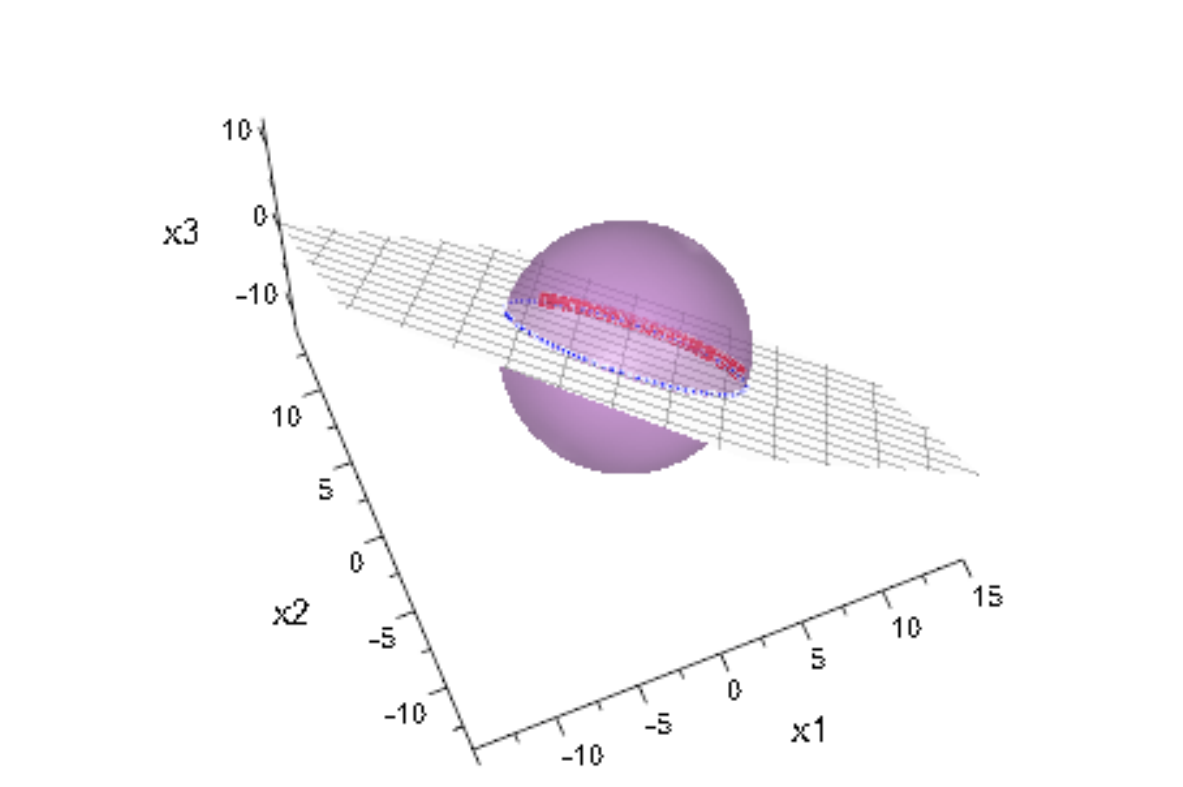}
\caption{Region $E_k$ (red) on hyperplane $H$ }
\end{subfigure}
\begin{subfigure}{.5\textwidth}
  \centering
  \includegraphics[width=\linewidth]{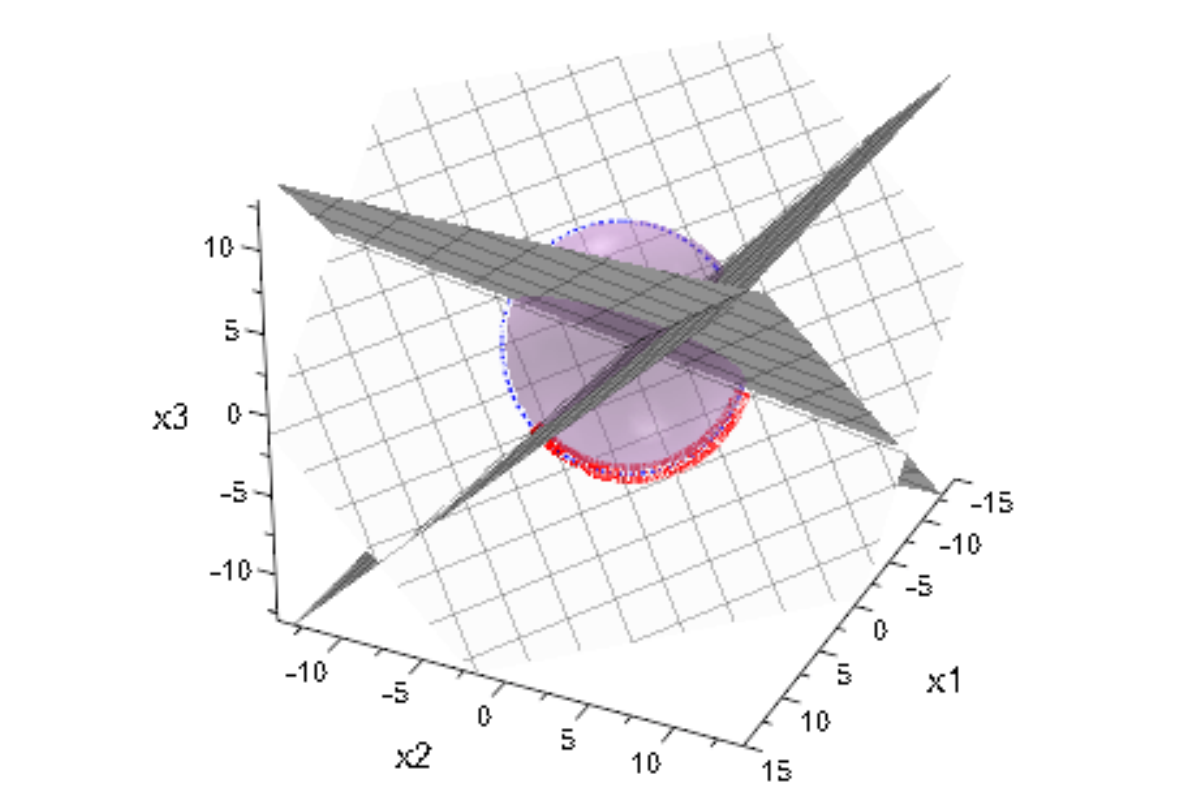}
\caption{Region $E_k$ (red) with planes $x_1=x_3$ and $x_2=x_3$ on hyperplane $H$}
\end{subfigure}
\caption{Graphical depiction of $C$ and $E_k$ for $k=3$.}
\label{fig:region}
\end{figure}

Plane $H$ is shown in Figure~\ref{fig:region} when $k=3$. The blue curve in Figure~\ref{fig:region}(a) shows $C$ when $k=3$ and the red curve in Figure~\ref{fig:region}(b) shows $E_k$, which is a part of $C$ divided by planes $x_1=x_3$ and $x_2=x_3$ as shown in Figure~\ref{fig:region}(c).

We now state the main result of this section.

\begin{lemma}
\label{lem:level1error}
Let $k\ge 3$.
Suppose that $Z_1,\ldots, Z_k$ are \IID standard normal.
The probability that the process $B(t)-\E_k(B(t))$
first hits $C$ in the part $E_k$ where the best system $k$ gets eliminated equals
\begin{equation}
\label{eq:probfirstelim}
\frac{\int_{-r}^r e^{{\frac{\delta_k}{\sigma}} y } d_y\Prob(Z_k=\min(Z_1,\ldots,Z_k), r(Z_k-\E_k(Z))\le  y \sqrt{(k-1)\Var_k(Z)} )}
{\left( {\eta_k \over 2} \right)^{-\nu} \Gamma(\nu+1)I_\nu\left({ \eta_k }\right)},
\end{equation}
where $\nu=(k-3)/2$,
$\Gamma$ stands for the Gamma function, and $I_\nu$ for the modified Bessel function of the first kind.
\end{lemma}
\begin{proof}
Writing $\zeta$ for the drift of $B(t)-\E_k(B(t))$, then
the hitting place of $B(t)-\E_k(B(t))$ on $C$ has density $f$ with respect to the uniform distribution $u_C$ on $C$.
Here $u_C$ should be interpreted as a volume element on $C$ in the terminology of differential geometry,
and by rotational invariance it has a `simulation interpretation' as the distribution of
\[
X=\frac{r(Z_1-\E_k(Z),\ldots,Z_k-\E_k(Z))}{\sqrt{k\Var_k(Z)}},
\]
where $Z$ is a standard normal vector in $\R^k$.
The density $f$ with respect to $u_C$ is given by
(e.g., Rogers and Pitman (1981))
\[
f(x) = \frac{e^{\langle\zeta, x\rangle}}{\int_C e^{ {\langle\zeta, w \rangle}} u_C(dw)}, \quad x\in C.
\]
This distribution is known as the von Mises distribution.

According to Rogers and Pitman (1981), for any $\mu\in\R^k$ with $\sum_i \mu_i=0$,
\[
\int_C e^{\langle \mu, w \rangle} u_C(dw) = (\|\mu\| r/2)^{-\nu} \Gamma(\nu+1) I_\nu(\|\mu\| r),
\]
where $\nu = (k-3)/2$.
Therefore, the denominator can be written as
\[
\int_C e^{ {\langle\zeta, w \rangle}} u_C(dw) =\left( { {\delta_k \over \sigma} r \over 2} \right)^{-\nu} \Gamma(\nu+1)I_\nu\left({ {\delta_k \over \sigma} r}\right)=\left( { \eta_k \over 2} \right)^{-\nu} \Gamma(\nu+1)I_\nu\left({ \eta_k}\right)
\]
because $\|\zeta\|=\delta\sqrt{(k-1)/k}/\sigma=\delta_k/\sigma$ and $(\delta_k / \sigma) r = \eta_k$.
Note that larger values of $B_k(t) -\E_k(B(t))$ are more likely than smaller values when the process hits $C$, which
should be expected because system $k$ is the best one.

The probability of eliminating the best system in level 1 equals
\begin{eqnarray*}
\int_{E_k}f(x)u_C(dx) &=& \E[\mathbbm{1}(X\in E_k) f(X)] \\
&=& \E[\mathbbm{1}(X_k=\min(X_1,\ldots,X_k)) f(X)] \\
&=& \frac{\E[\mathbbm{1}(X_k=\min(X_1,\ldots,X_k)) e^{\langle \zeta, X\rangle}]}{\int_C e^{\langle\zeta, w\rangle} u_C(dw)} \\
&=& \frac{\int_{-r}^r e^{\frac{\delta_k}\sigma y} d_y\Prob(X_k=\min(X_1,\ldots,X_k), \langle \zeta,X\rangle\le \frac{\delta_k}\sigma y)}{\int_C e^{\langle\zeta, w\rangle} u_C(dw)},
\end{eqnarray*}
where $X$ has a uniform distribution on $C$ (see the beginning of this proof) and $\mathbbm{1}$ stands for the indicator function.
Since $\langle \zeta,x\rangle = \frac\delta\sigma x_k$ for $x\in H$, the sought probability equals
\[
\frac{\int_{-r}^r e^{\frac{\delta_k}\sigma y} d_y\Prob(Z_k=\min(Z_1,\ldots,Z_k), \frac \delta \sigma r (Z_k-\E_k(Z))/\sqrt{k\Var_k(Z)}\le \frac{\delta_k}\sigma y)}{\int_C e^{\langle\zeta, w\rangle} u_C(dw)},
\]
as claimed.
\end{proof}

The preceding lemma yields a Monte Carlo method for calculating the probability of immediate elimination of the best system.
Indeed, it states that this probability equals
\begin{equation}
\label{eq:simnew}
\frac{\E\left[\exp\left(\eta_k \frac{Z_k-\E_k(Z)}{\sqrt{(k-1)\Var_k(Z)}}\right); Z_k=\min(Z_1,\ldots,Z_k)\right]}
{\left( {\eta_k \over 2} \right)^{-\nu} \Gamma(\nu+1)I_\nu\left({ \eta_k }\right)},
\end{equation}
for \IID standard normal $Z_1,\ldots,Z_k$.
However, for large $k$, such a Monte Carlo method is not efficient and we instead approximate the level 1 probability (\ref{eq:probfirstelim})
by replacing several of its components by asymptotic approximations.
For instance, as $k\to\infty$, the random variables $\E_k(Z)$ and $\Var_k(Z)$ converge in distribution to 0 and 1, respectively, by the
strong law of large numbers. The rate of convergence is relatively fast (order $1/\sqrt{k}$ by the central limit theorem).
We, therefore, approximate those variables by
their deterministic asymptotic approximations. The term with the minimum is slightly more complicated.
Writing
\[
c_k = \sqrt{2\log k} - \frac{\log\log k + \log(4\pi)}{2\sqrt{2\log k}},
\]
$\min(Z_1,\ldots,Z_{k-1})+c_{k-1}$ converges in distribution to 0. For example, see Example 3.3.29 in Embrechts, Kluppelberg and Mikosch (1997). The rate of convergence is relatively slow (order $1/\sqrt{2\log k}$), so we use an approximation based on the fact that
\[
\sqrt{2\log k}(\min(Z_1,\ldots,Z_{k-1})+c_{k-1})
\]
converges in distribution to $-G$ where $G$ is a Gumbel distributed random variable which is equal in distribution to
$-\log(-\log(U))$ where $U$ is standard uniformly distributed.
Even when the central limit theorem is used for the sum instead of the law of large numbers,
the minimum and sum are asymptotically independent (e.g., Chow and Teugels 1978). This motivates the approximation, for $ y\in(-r,r)$,

\begin{eqnarray*}
\lefteqn{\Prob(Z_k=\min(Z_1,\ldots,Z_k), r(Z_{k}-\E_k(Z))\le  y\sqrt{(k-1)\Var_k(Z)}  ) }\\&\approx&
\Prob(Z_k\le \min(Z_1,\ldots,Z_{k-1}), rZ_{k}\le  y\sqrt{(k-1)}  ) \\&\approx&
\Prob(Z_{k} \le {-G}/\sqrt{2\log k} -c_{k-1}, rZ_{k} \le y \sqrt{k-1}),
\end{eqnarray*}
where $Z_{k}$ and $G$ are independent.

We are now ready to formulate our approximation for (\ref{eq:probfirstelim}), and we first assume $G$ is a given constant.

\begin{lemma}\label{lem:level1approx}
For fixed $a \in \R$, we have
\begin{eqnarray*}
\lefteqn{
\int_{-r}^r e^{\frac{\delta_k}\sigma y} d_y\Prob(Z_{k}\le -a /\sqrt{2\log k} -c_{k-1}, rZ_{k}/\sqrt{k-1} \le y)}\\ &=&
\exp\left(\frac{\eta_k^2}{2(k-1)}\right) \left[ \Phi\left(\min\left( \max\left(-\sqrt{k-1},\frac{ -a}{\sqrt{2\log k}} -c_{k-1}\right),\sqrt{k-1}\right) - \frac{ \eta_k}{{\sqrt{k-1}}} \right) - \Phi\left( -\sqrt{k-1} -  \frac{ \eta_k}{{\sqrt{k-1}}} \right) \right].
\end{eqnarray*}
where $\Phi(\cdot)$ is the cumulative distribution function (cdf) of the standard normal random variable.
\end{lemma}
\begin{proof}

Letting $Y$ be a centered Gaussian variable with variance $r^2/(k-1)$. For any $\kappa\in\R$, we then have
\begin{eqnarray*}
\lefteqn{
\int_{-r}^r e^{(\delta_k/\sigma) y} d_y\Prob(Z_{k}\le \kappa, rZ_{k}/\sqrt{k-1} \le y)}\\
&=& \int_{-r}^{r\min(\max(-1,{\kappa/\sqrt{k-1}}),1)} e^{(\delta_k/\sigma) y} d\Prob(Y\le y) \\
&=& \int_{-r}^{r\min(\max(-1,{\kappa/\sqrt{k-1}}),1)} e^{(\delta_k/\sigma) y} \frac{\sqrt{k-1}}{r\sqrt{2\pi}} \exp\left({-\frac {(k-1)y^2}{2r^2}}\right) dy \\
&=& e^{\frac{ (\delta_k/\sigma)^2 r^2}{2(k-1)}} \frac{\sqrt{k-1}}{\sqrt{2\pi} r} \int_{-r}^{r\min(\max(-1,{ \kappa/\sqrt{k-1}}),1)} \exp\left({-\frac {\left(y - \frac{(\delta_k/\sigma) r^2}{{(k-1)}}\right)^2}{2r^2/(k-1)}}\right)dy\\
&=& e^{\frac{ \eta_k^2}{2(k-1)}} \frac{\sqrt{k-1}}{\sqrt{2\pi} r} \int_{-r}^{r\min(\max(-1,{\kappa/\sqrt{k-1}}),1)} \exp\left({-\frac {\left(y - \frac{(\delta_k/\sigma) r^2}{{(k-1)}}\right)^2}{2r^2/(k-1)}}\right)dy \\
&=& e^{\frac{ \eta_k^2}{2(k-1)}} \left[ \Phi\left( \min(\max(-\sqrt{k-1},\kappa),\sqrt{k-1}) - \frac{ (\delta_k / \sigma) r}{{\sqrt{k-1}}} \right) - \Phi\left( -\sqrt{k-1} -  \frac{ (\delta_k / \sigma) r}{{\sqrt{k-1}}} \right) \right]\\
&=&e^{\frac{ \eta_k^2}{2(k-1)}} \left[ \Phi\left(\min(\max(-\sqrt{k-1},\kappa),\sqrt{k-1}) - \frac{ \eta_k}{{\sqrt{k-1}}} \right) - \Phi\left( -\sqrt{k-1} -  \frac{ \eta_k}{{\sqrt{k-1}}} \right) \right],
\end{eqnarray*}
as claimed.
\end{proof}

In summary, we approximate the probability of first eliminating the best system by
\begin{equation}
\label{eq:approximationlevelI}
\frac{\exp\left(\frac {\eta_k^2}{2(k-1)}\right) \left[ \E\Phi\left( \min\left(\max\left(-\sqrt{k-1},\frac{ -G}{\sqrt{2\log k}} -c_{k-1}\right),\sqrt{k-1}\right) - \frac{ \eta_k}{{\sqrt{k-1}}} \right)- \Phi\left(-\sqrt{k-1} - \frac{\eta_k }{\sqrt{k-1}}\right) \right]}
{\left(\eta_k /2\right)^{-\nu} \Gamma(\nu+1)I_\nu(\eta_k)}.
\end{equation}

The expectation in (\ref{eq:approximationlevelI}) can be estimated through either Monte Carlo by generating Gumbel random variates or numerical integration from 0 to 1 which is the range of a random number $U$. Both can be done fast but we use the latter method because it is faster and free of sampling error. In \label{subsec:otherelim} we explain in more detail how the numerical integration is performed.


\begin{remark}
\label{rem:shane}
Shane Henderson communicated to us that the numerator in (\ref{eq:simnew}) can be written as
\[
\E\left[\exp\left(\eta_k \frac{\min(Z_1,\ldots,Z_k) - \E_k(Z)}{\sqrt{(k-1)\Var_k(Z)}}\right); Z_k = \min(Z_1,\ldots,Z_k)\right],
\]
and since $\min(Z_1,\ldots,Z_k)$, $\E_k(Z)$, and $\Var_k(Z)$ do not change when the elements of $Z$ are permuted, this equals
\[
\frac 1k \sum_{i=1}^k
\E\left[\exp\left(\eta_k \frac{\min(Z_1,\ldots,Z_k) - \E_k(Z)}{\sqrt{(k-1)\Var_k(Z)}}\right); Z_i = \min(Z_1,\ldots,Z_k)\right] =
\frac 1k \E\left[\exp\left(\eta_k \frac{\min(Z_1,\ldots,Z_k) - \E_k(Z)}{\sqrt{(k-1)\Var_k(Z)}}\right)\right].
\]
Using the approximations $\E_k(Z)\approx 0$, $\Var_k(Z)\approx 1$, $\min(Z_1,\ldots,Z_k) \approx -G/\sqrt{2\log(k)} - c_k$ as before,
the numerator in (\ref{eq:simnew}) can be approximated by
\[
\frac 1k \E\left[\exp\left(-\eta_k \frac{G}{\sqrt{2(k-1)\log(k)}} - \frac{\eta_kc_k}{\sqrt{k-1}}\right)\right] =
\frac 1k  e^{-\eta_k c_k/\sqrt{k-1}}  \Gamma\left( 1+ \frac{\eta_k}{\sqrt{2(k-1) log(k)}}\right),
\]
since $\E[e^{-\gamma G}] = \E[Y^\gamma] = \int_0^\infty x^\gamma e^{-x} dx = \Gamma(\gamma+1)$ for a standard exponentially distributed random variable $Y$.

We do not use this approximation in the remainder of this paper, since experiments have shown that it leads to higher PCS than (\ref{eq:approximationlevelI}).
\end{remark}

\subsection{Other Level Errors}
\label{subsec:otherelim}
For level $\ell$ errors for $\ell=2,3,\ldots, k-1$, the number of survived systems $|I|$ is $|I|=k-\ell+1$ and it is natural to replace $k$ with $|I|$  in (\ref{eq:approximationlevelI}) as follows: \begin{equation}
\label{eq:raweta}
\frac{\exp\left(\frac {\eta_{|I|}^2}{2(|I|-1)}\right) \left[ \E\Phi\left( \min\left(\max\left(-\sqrt{|I|-1},\frac{ -G}{\sqrt{2\log |I|}} -c_{|I|-1}\right),\sqrt{|I|-1}\right) - \frac{ \eta_{|I|}}{{\sqrt{|I|-1}}} \right)- \Phi\left(-\sqrt{|I|-1} - \frac{\eta_{|I|} }{\sqrt{|I|-1}}\right) \right]}
{\left(\eta_{|I|} /2\right)^{-\nu} \Gamma(\nu+1)I_\nu(\eta_{|I|})}
\end{equation}
where $\nu = (|I|-3)/2$.
In our procedure,
$\eta_{|I|}$ is calculated as the solution to $(\ref{eq:raweta}) = \beta_\ell$ for $0< \beta_\ell < \alpha$. We let $P_k(\ell/k,\beta_\ell)$ represent level $\ell$ error, the probability of incorrectly eliminating the best system at level $\ell$ when $\eta_{|I|}$ is calculated with target $\beta_\ell$. Note that it does not depend on $\delta$ or $\sigma$ by Lemma~\ref{lem:deltasigma}.
The probability of incorrect selection (PICS) of ${\cal DK}_1$ is
\[
{\rm PICS} = \sum_{\ell=1}^{k-1} P_k(\ell/k,\beta_\ell).
\]

Let $\beta_0 = \alpha/(k-1)$. If $P_k(\ell/k, \beta_0)$ for $\ell=1,\ldots,k-1$ are all approximately equal to $\beta_0$, then the overall PICS would be approximately equal to $\alpha$. For large $k$, the analysis in Section~\ref{subsec:immediateelim} ensures that $\eta_k$, the solution to $(\ref{eq:approximationlevelI}) = \beta_0$, would result in the level 1 error approximately equal to $\beta_0$. For other level errors, we do not have control over the error probability but we propose an approximation.

For the derivation of (\ref{eq:approximationlevelI}), it is critical that the starting point of the corresponding Brownian motion is the origin. For levels $\ell > 1$, we start at a random point from the previous level and thus we do not necessarily have $P_k(\ell/k, \beta_0) \approx \beta_0$ if we let $\eta_{|I|}$ be the solution to $(\ref{eq:raweta}) = \beta_0$, unless we discard all observations from previous levels. This is not desirable because too many observations would be wasted. Instead, we seek for a heuristic way to determine $\eta_{|I|}$ under the following assumption:

\begin{assumption}\label{assump:strong}
For $0< \beta_\ell<\alpha$, $\ell = 1, 2, \ldots, k-1$ and $\beta_0 = \alpha/(k-1)$,
\begin{enumerate}
\item  $P_k( \ell/k,  \beta_\ell) \approx \beta_\ell \cdot q_k(\ell/k)$; and

\item If $\beta_\ell = \beta_0$ for all $\ell$, then the probability that an incorrect selection (ICS) event occurs at standardized level $\ell/k$ is approximately $\int_{\ell -1 \over k-1}^{\ell \over k-1}g(w)dw$ for a density function $g(\cdot)$ in [0,1].
\end{enumerate}
\end{assumption}

Assumption~\ref{assump:strong}.1 is effectively a first-order Taylor approximation under appropriate differentiability assumptions because $\lim_{\beta\downarrow 0}P_k(\ell/k,\beta)=0$. This assumption implies that for small $\beta_\ell$, the level $\ell$ error is approximately linear in $\beta_\ell$. For example, if $\beta_\ell$ decreases in half for level $\ell$, then the level $\ell$ error is expected to be cut in half.

We have empirical evidence for Assumption~\ref{assump:strong}.2.  To test Assumption~\ref{assump:strong}.2, we made one million replications and recorded standardized levels where ICS occurred for each experimental setting. Then a kernel density estimator is fitted to the data using Matlab with a normal kernel and support $[0,1]$. A bandwidth was chosen by Matlab, which is known to be optimal for the normal kernel. Figure~\ref{fig:ratio} shows kernel density estimates of standardized levels $\ell/k$ of having ICS for $k=75, 150, 500$ and $1000$ and $\alpha = 0.05$ and $0.10$ when $\delta = 0.3$ and $\sigma^2 = 1$. Note that the specific choice of $\delta$ and $\sigma$ does not matter in view of Lemma~\ref{lem:deltasigma}. From the figure, one can see that the shapes of kernel estimates for various $k$ are similar.

\begin{figure}[h!]
\begin{subfigure}{.5\textwidth}
  \centering
  \includegraphics[width=\linewidth]{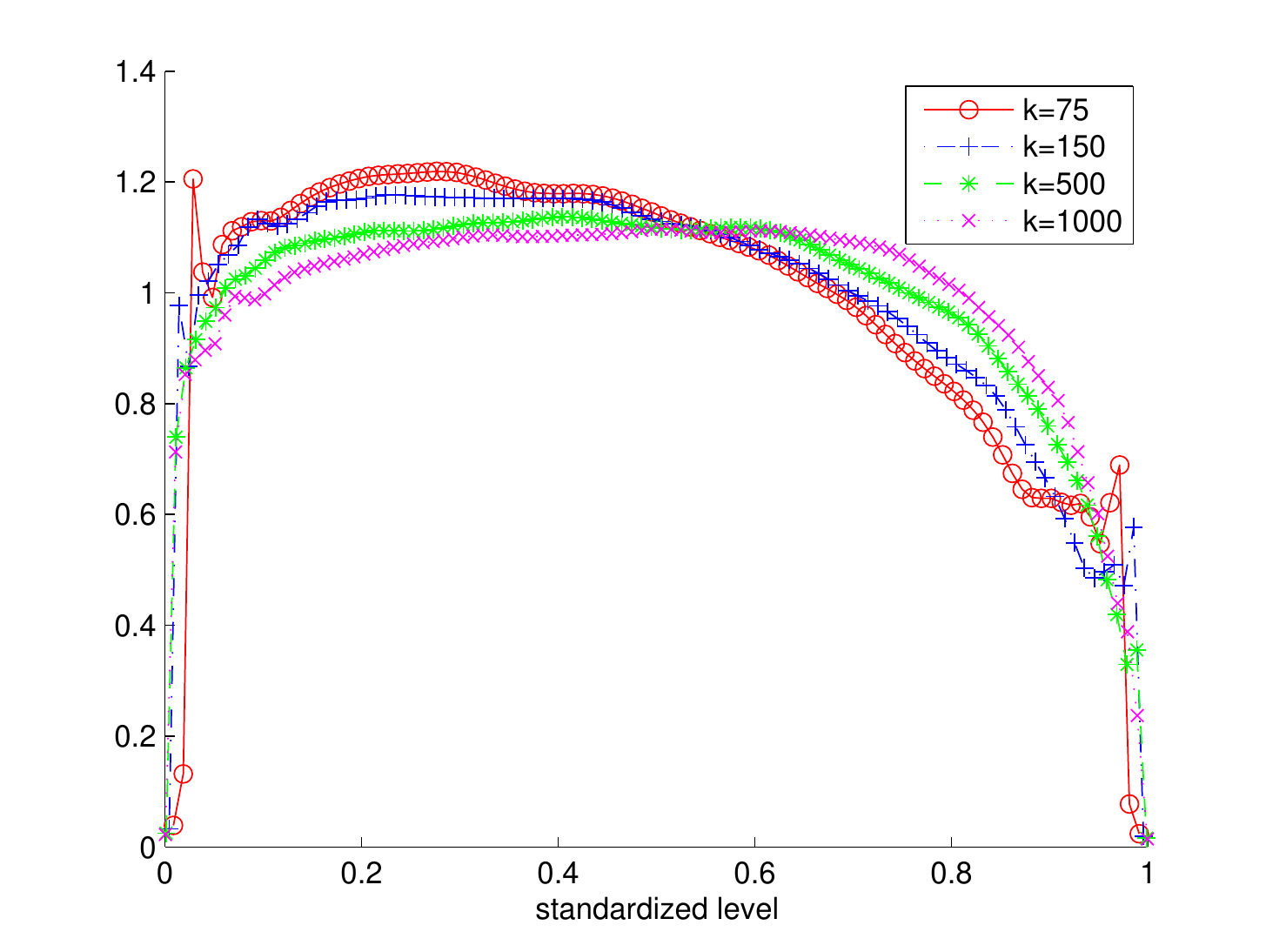}
  \caption{ $\alpha = 0.10$}
  \end{subfigure}
\begin{subfigure}{.5\textwidth}
\centering
\includegraphics[width=\linewidth]{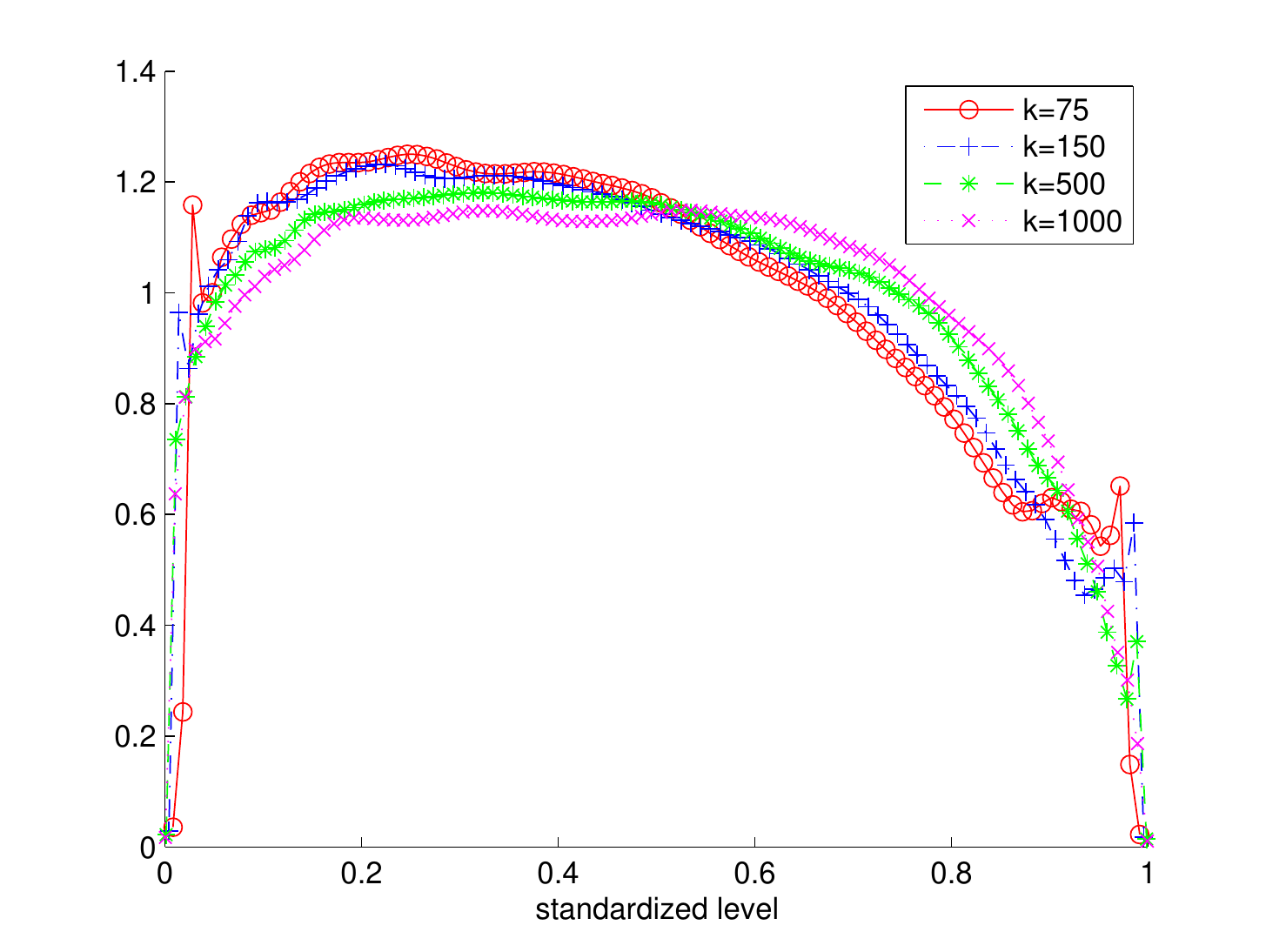}
\caption{$\alpha = 0.05$}
\end{subfigure}
\caption{Kernel density estimates on standardized levels where an incorrect selection occurs for various $k$ when $\delta =0.3$, and $\sigma^2 =1$.}\label{fig:ratio}
 \end{figure}


\begin{figure}[t!]
\begin{subfigure}{.5\textwidth}
  \centering
  \includegraphics[width=\linewidth]{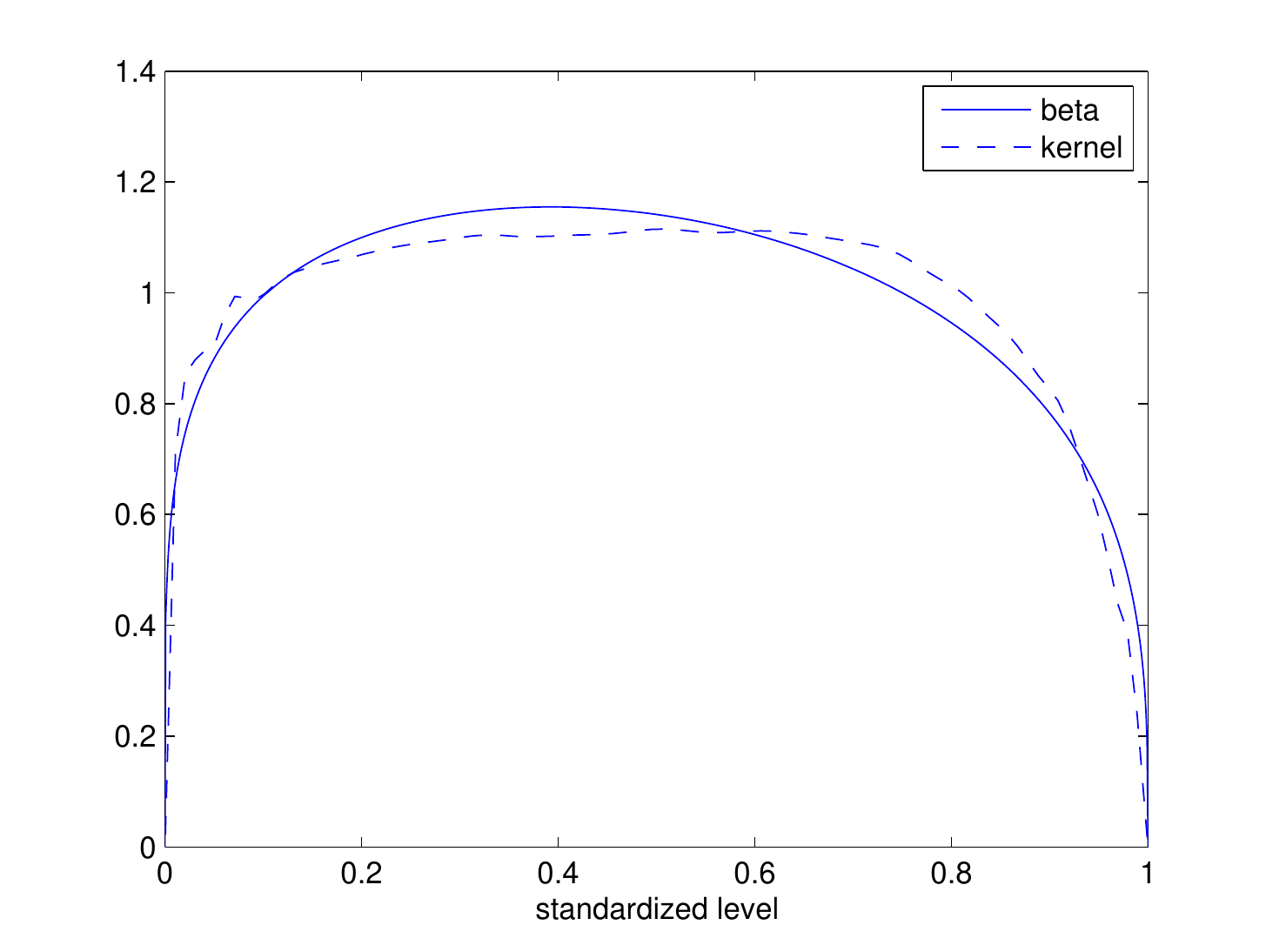}
\caption{ $\alpha = 0.10$ and $g(w) \propto w^{0.19805}(1-w)^{0.30662}$}
\end{subfigure}%
\begin{subfigure}{.5\textwidth}
  \centering
  \includegraphics[width=\linewidth]{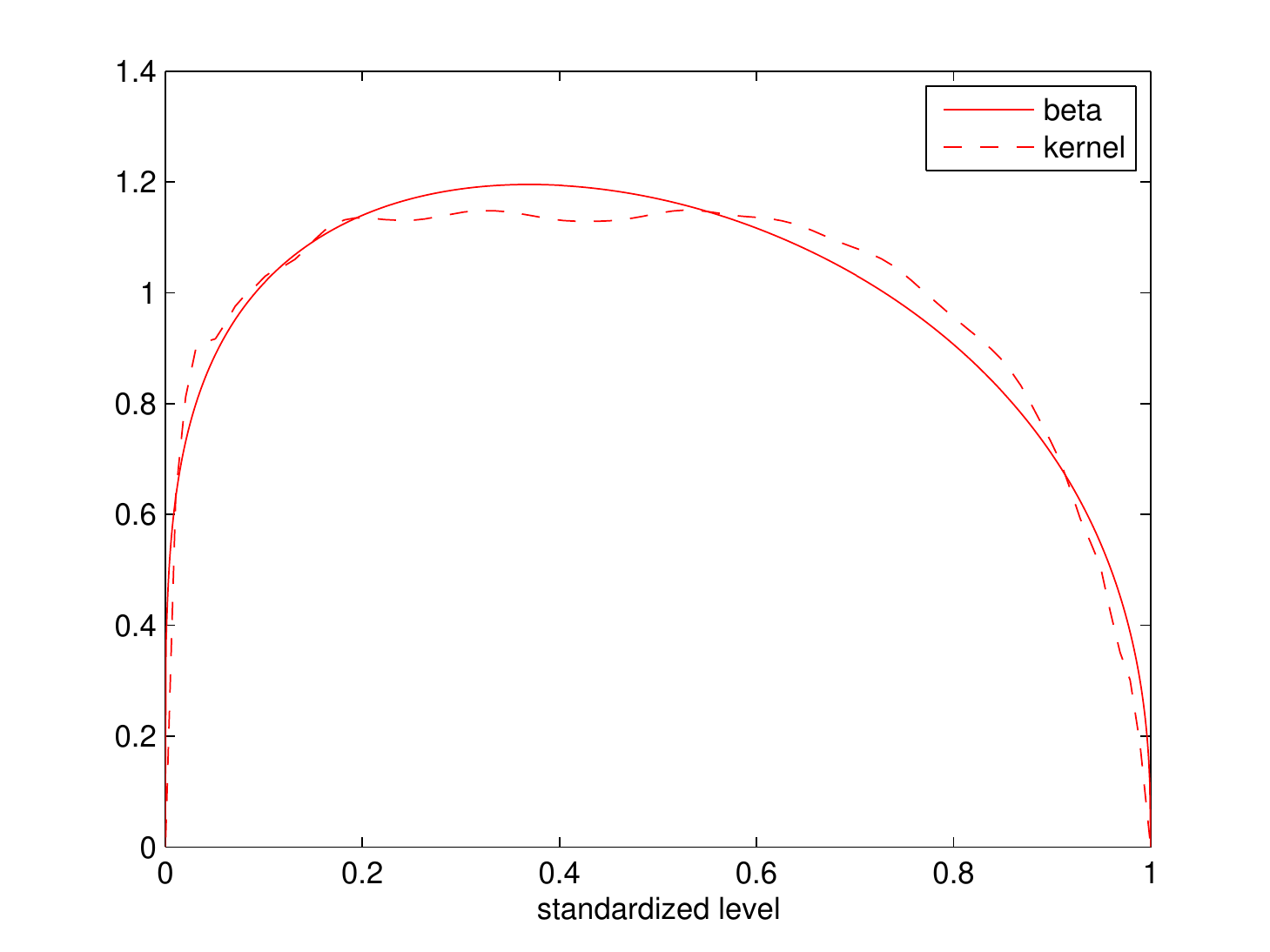}
\caption{ $\alpha = 0.05$ and $g(w) \propto w^{0.2317}(1-w)^{0.39658}$}
\end{subfigure}
\caption{ Kernel density estimates and beta density estimates for $g(w)$ when $k=1000$ and $\delta  = 0.3$.}
\label{fig:actualfitted}
\end{figure}

To approximate $g(w)$ for $0<w<1$, we use $k=1000$ rather than $k=75$ because $k=75$ gives sparse points in $[0,1]$. In addition, as kernel estimates are not stable on the boundary 0 and 1, we further fit it using a beta distribution to get a smooth function especially close to the boundary points, assuming
\[
g(w)  \approx {1 \over {\rm Beta}(A, B)} \; w^{A-1} \; (1 - w)^{B-1} \quad \mbox{for} \quad 0<w<1 \mbox{ and } A,B \in \mathbb{R}
\]
where Beta$(A,B) = \int_0^1 t^{A-1}(1-t)^{B-1} dt$. Figure~\ref{fig:actualfitted} shows the fitted beta densities for $\alpha = 0.05$ and $\alpha = 0.10$, respectively. Both beta densities are very similar.

Once $g(w)$ is approximated, $\eta_{|I|}$ which ensures the probability of correct selection (PCS) of ${\cal DK}_1$ can be calculated as follows:
\begin{description}
\item [Step 1:] Calculate the constant $m_\ell$ as follows:
\begin{eqnarray*}
 m_\ell  &= & \frac{ G\left( \ell \over k-1\right) - G\left( \ell - 1  \over k-1\right) } {G\left(1 \over k-1\right)}
\end{eqnarray*}
where $G(w) = \int_0^w g(t) dt$, the cdf of $g(w)$.
\item [Step 2:] Set $\beta_\ell = {\beta_0 \over m_\ell}$ and calculate $\eta_{|I|}$ from $(\ref{eq:raweta}) = \beta_\ell$.
\end{description}

The relative magnitude between $G\left(1 \over k-1\right)$ and $G\left( \ell \over k-1\right) - G\left( \ell - 1  \over k-1\right)$ can be interpreted as the relative magnitude between level 1 error and level $\ell$ error when $\eta_{|I|}$ is calculated from $(\ref{eq:raweta}) = \beta_0$. As we know that level 1 error is $\beta_0$,
\[
G\left(1 \over k-1\right): G\left( \ell \over k-1\right) - G\left( \ell - 1  \over k-1\right) \approx \beta_0: P_k(\ell/k, \beta_0).
\]
Therefore the level $\ell$ error is expected to be inflated by $m_\ell = \frac{ G\left( \ell \over k-1\right) - G\left( \ell - 1  \over k-1\right) } {G\left(1 \over k-1\right)}$ compared to target $\beta_0$. If we adjust $\beta_\ell = \beta_0/m_\ell$, then $P_k(\ell/k, \beta_\ell) \approx \beta_0$ by Assumption~\ref{assump:strong}.1, which in turns implies that the overall PICS is approximately equal to $\alpha$.


Figure~\ref{fig:adjustedlevel} shows estimated level errors $\hat{P}_k (\ell/k, \beta_0/m_\ell)$ for $\alpha =$ 5\% and 10\% when $k=512$ with $\delta = 0.3$ and $\sigma^2=1$ and one million replications. One can see that the level errors do not show a beta shape as in Figure~\ref{fig:actualfitted}. Instead the ratios between level errors and $\beta_0$ fluctuate around one for the two values of $\alpha$, which empirically supports Assumption~\ref{assump:strong}.1. Also, it shows that the function we found $g(w)$ for $\alpha=5\%$ and $k=1000$ seems to work well for other popular choices of $\alpha$, including $\alpha= 10\%$.

\begin{figure}[t!]
\begin{subfigure}{.5\textwidth}
  \centering
  \includegraphics[width=\linewidth]{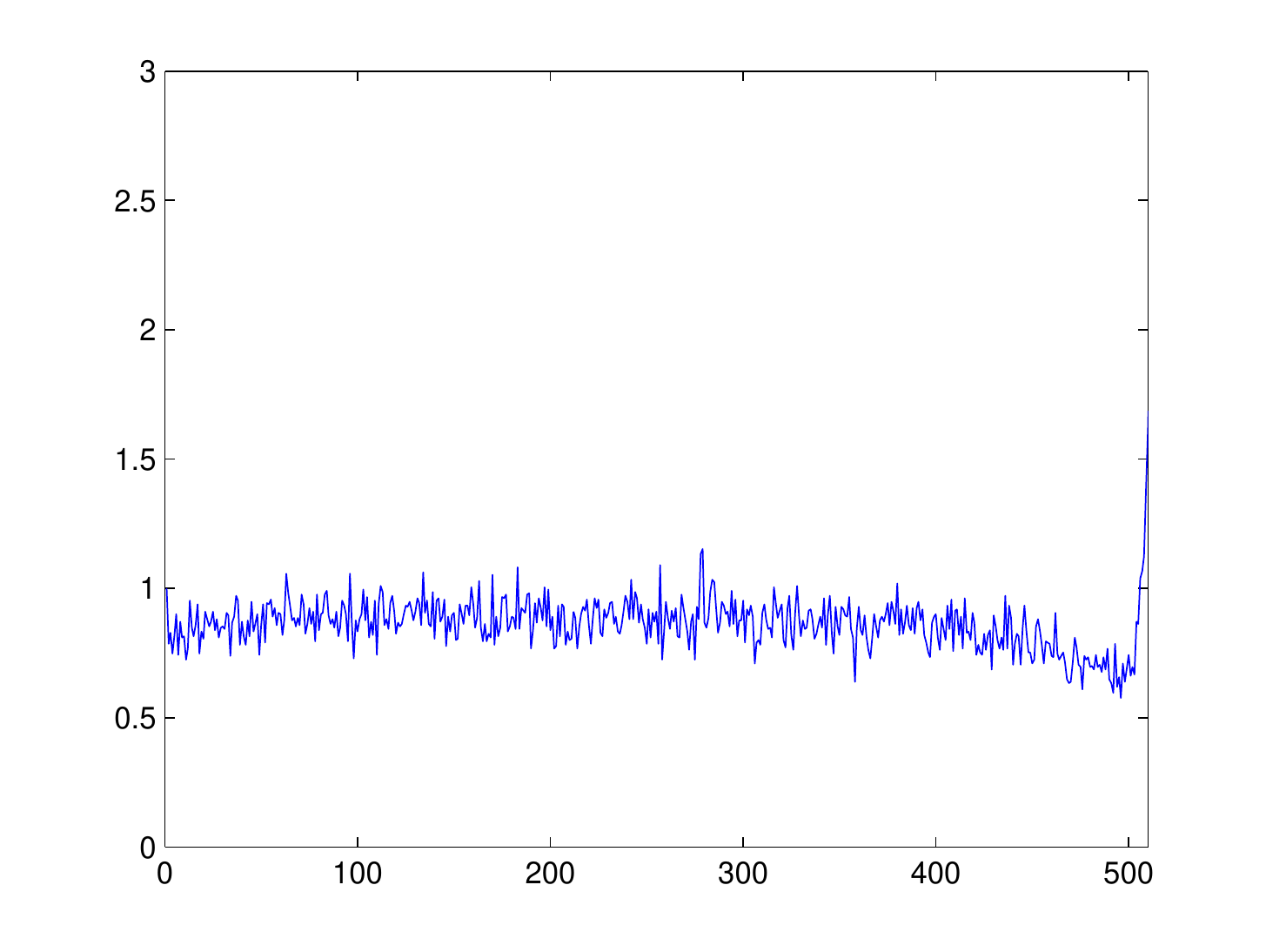}
\caption{ $1-\alpha = 0.90, {\rm PCS}=0.907$}
\end{subfigure}%
\begin{subfigure}{.5\textwidth}
  \centering
  \includegraphics[width=\linewidth]{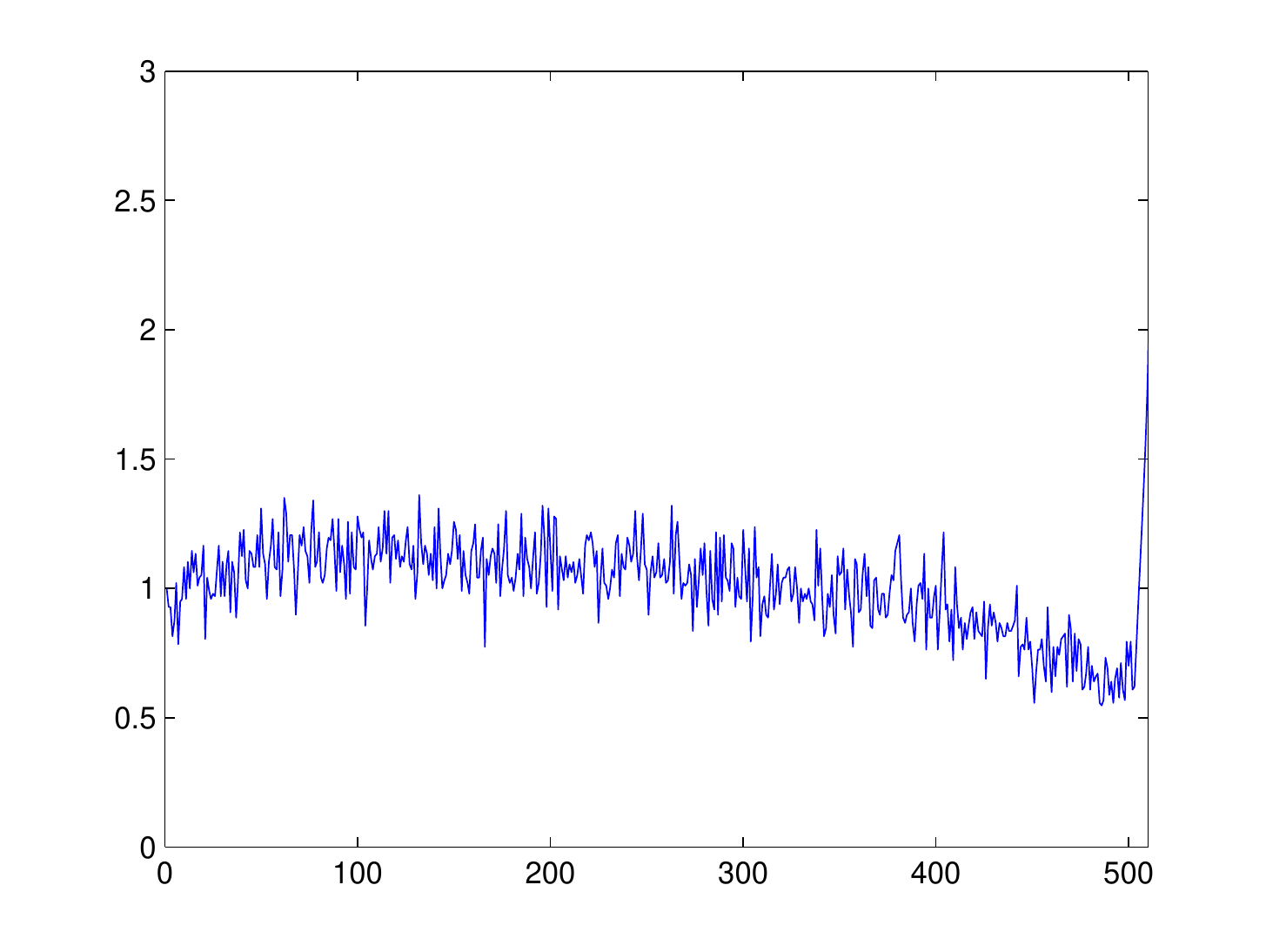}
\caption{ $1-\alpha = 0.95, {\rm PCS}=0.950$}
\end{subfigure}
\caption{Ratios between $\hat{P}_{k} (\ell/k, \beta_0/m_\ell)$ and $\beta_0$ for $\alpha=10\%$ and $5\%$ when  $k=512$, $\delta = 0.3$ and $\sigma^2 =1$.}
\label{fig:adjustedlevel}
\end{figure}

To search for $\eta_{|I|}$ for given target $\beta_\ell$, a bisection search is used and this requires estimating the expectation in (\ref{eq:raweta}). Instead of using Monte Carlo by generating Gumbel random variates $G$, we use numerical integration in the range of a standard uniform random variable $U$, $[0,1]$ using one million intervals on function $f(u)$ defined as follows:
\[
f(u) \equiv \Phi\left(\min\left(\max\left(-\sqrt{|I|-1},\frac{ \log(-\log u)}{\sqrt{2\log |I|}} -c_{|I|-1}\right),\sqrt{|I|-1}\right) - \frac{ \eta_{|I|}}{{\sqrt{|I|-1}}}\right) \quad \mbox{for} \quad 0 < u < 1,
\]
 $f(0) \equiv \lim_{u\rightarrow 0} f(u)$ and $f(1) \equiv \lim_{u\rightarrow 1}f(u)$.  When $u\rightarrow 0$ or $u \rightarrow 1$, $\log(-\log U)$ converges to either $\infty$ or $-\infty$ but the minimum and maximum functions inside $\Phi(\cdot)$ in $f(u)$ ensure a finite number is returned.
Then
\begin{eqnarray*}
& \E\Phi\left( \min\left(\max\left(-\sqrt{|I|-1},\frac{ \log(-\log u)}{\sqrt{2\log |I|}} -c_{|I|-1}\right),\sqrt{|I|-1}\right) - \frac{ \eta_{|I|}}{{\sqrt{|I|-1}}} \right)\\[6pt]
&\approx {1\over 1000000}\left[ {1 \over 2} f(0) + \sum_{j=1}^{999999} f( j/1000000) + {1\over 2} f(1) \right].
\end{eqnarray*}

 The parameter $\eta_{|I|}$ is searched using the deterministic bisection method when the numerical integration is used.
 Note that the above approximation is based on the assumption that $|I|$ is large. When $|I|$ is small, say $|I|<10$, (\ref{eq:raweta}) does not work well. Instead we use (\ref{eq:simnew}) which requires a Monte Carlo simulation with $|I|$ number of iid standard normal random variables. When a Monte Carlo simulation is used, there is a chance that a deterministic bisection method may fail due to simulation error. Therefore when $|I|<10$, we employee a probabilistic bisection algorithm  (Section 1.5 of Waeber (2013)) is used. The stochastic bisection algorithm stops when the returned median of a posterior distribution in the current search iteration is within 0.001 of the median from the previous iteration. A sequential test of power one which determines the sign of the objective function is implemented with parameters $r_0=50000$ and $\gamma = 0.01$. The sequential test stops either when $m$ reaches 1000 or when its test statistics exit $(-k_m, k_m)$ where $k_m$ is from equation (B.6) of Waeber (2013).

  Since $\eta_{|I|}$ only depends on $\alpha$ and $k$, a table can be made for popular choices of $\alpha$ such as 5\% and 10\% and $k=2,3,\ldots, 10000$. Then the values of $\eta_{|I|}$ can be read from the table while running our procedures. Table~\ref{tab:eta} in the appendix shows the values of $\eta_{|I|}$ for a few selected values of $k$ when $\alpha = 10\%$.


\subsection{Justification of Procedures for Unknown Variances}
\label{subsec:unknown}

In this subsection, we discuss why ${\cal DK}_2$ and ${\cal DK}_3$ should be expected to work for unknown variances as well.
For unknown variances, it is natural to replace variance parameters in ${\cal DK}_1$ to their estimated values. In general, it is not sufficient to replace the variance parameter with its estimated value to keep the statistical validity. It is critical to account for the variability in the estimated parameter especially when variances are estimated only once based on an initial $n_0$ observations. Kim and Nelson (2006) and Wang and Kim (2011) show that if variance estimators are updated on the fly in a procedure as more observations are obtained, then the procedure converges to the known variance case under some appropriate asymptotic regime. In the light of these results, we employ a variance updating scheme in ${\cal DK}_2$ and ${\cal DK}_3$ to avoid the difficulty of accounting for the variability in the estimated variance parameters but without any claim for the asymptotic validity in this paper.

When the decision maker believes that the variances across systems are equal (but unknown), then the natural estimator for $\sigma^2$ is the pooled variance estimator
\[
\hat{\sigma}^2_p(n) = {1 \over |I|} \sum_{i \in I} \hat{\sigma}^2_i(n).
\]
As we update $\hat{\sigma}^2_p(n)$ as more observations become available, the estimator converges to $\sigma^2$ and thus it is expected that ${\cal DK}_2$ works similarly to ${\cal DK}_1$.

When variances are unknown and unequal, we use similar arguments as in Frazier (2014). Let $n_i = \gamma \sigma_i^2 n$ for some $\gamma >0$ and thus the number of samples obtained by stage $n$ for system $i$ is proportional to its variance $\sigma_i^2$.
Then
\[
{\sum_{j=1}^{n_i} \Xij \over \gamma \sigma_i^2}  \sim  N\left( {n_i \over \gamma \sigma_i^2} \mu_i, {n_i \over \gamma^2 \sigma_i^2} \right) = N\left( n \mu_i, {n \over \gamma} \right) \approx B_{(\mu_i, 1/\gamma)}(t)
\]
where $B_{(\mu_i, 1/\gamma)}(t)$ is a Brownian motion with drift $\mu_i$ and variance $1/\gamma$. The $ {\sum_{j=1}^{n_i} \Xij \over \gamma \sigma_i^2}$ have equal variance as long as $n_i = \gamma \sigma_i^2 n$ and thus we can apply ${\cal DK}_1$ to ${\sum_{j=1}^{n_i} \Xij \over \gamma \sigma_i^2}$. Note that when $n_i = \gamma \sigma_i^2 n$,
\[
n_i \lambda^2 = n_i {\sum_{i\in I} \sigma_i^2 \over \sum_{i \in I} n_i} = \sigma_i^2
\]
where
\[
\lambda^2 = {\sum_{i\in I} \sigma^2_i(n_i) \over \sum_{i\in I} n_i}.
\]
Then
\[
{\sum_{j=1}^{n_i} \Xij \over \gamma \sigma_i^2} = {\sum_{j=1}^{n_i} \Xij \over \gamma n_i \lambda^2} = {W_i(n) \over \gamma \lambda^2}.
\]
Finally, the screening rule in ${\cal DK}_1$ is
\[
\frac{ \sum_{i\in I} \left( {W_i(n) \over \gamma \lambda^2} - {1 \over |I|} \sum_{i\in I} {W_i(n) \over \gamma \lambda^2}\right)^2} { 1/\gamma }
 \ge {1 \over \gamma} \left({\eta_{|I|} \over \delta_{|I|}}\right)^2
\]
which is equivalent to
\[
 {1 \over \lambda^4} \sum_{i\in I} \left( W_i(n)  - {1 \over |I|} \sum_{i\in I} W_i(n) \right)^2 \ge  \left({\eta_{|I|} \over \delta_{|I|}}\right)^2
\]
or
\begin{equation} \label{eq:screening3}
 {1 \over \lambda^2} \sum_{i\in I} \left( W_i(n)  - {1 \over |I|} \sum_{i\in I} W_i(n) \right)^2 \ge  \left({\lambda \cdot \eta_{|I|} \over \delta_{|I|}}\right)^2
\end{equation}

When $\lambda^2$ is replaced with its estimator $\hat{\lambda}^2$ in (\ref{eq:screening3}), we get the same elimination rule in the ${\cal DK}_3$ procedure, which is
\[
{\cal S}_I''(\vW_I(n)) \ge  \left({ \hat{\lambda} \cdot \eta_{|I|} \over  \delta_{|I|}}\right)^2.
\]

\section{Experiments}
\label{sec:exp}

In this section, we compare the performance of ${\cal DK}$ procedures with KN and BIZ.
For unknown variances, we use the KN procedure as originally described in Kim and Nelson (2001) with $c=1$ and $n_0=30$ and Algorithm 2 of Frazier (2014) with $B_z=1$ and $n_0=30$. For known variances, we use KN with $h^2 = 2 \eta$ where $\eta = - \ln{ \left(2 {\alpha\over k-1}\right)}$ and $n_0=1$, which is same as the ${\cal P}$ procedure in Wang and Kim (2011), and Algorithm 1 of Frazier (2014).
Throughout this section, KN and BIZ refer procedures for known variances while KN-UNK and BIZ-UNK refer procedures for unknown variances.

The number of systems $k$ varies over
\[
k\in \{ 2,3,4,5,6,7,8,16,32, 64, 128, 256, 512, 1024, 2048, 4096, 8192\}.
\]
For the mean, we consider two mean configurations, namely slippage configuration (SC) and monotonic decreasing mean configuration (MDM); and for variances, we consider three variance configurations called Equal, INC, and DEC. Thus we have total six configurations: SC-Equal, MDM-Equal, SC-INC, SC-DEC, MDM-INC and MDM-DEC.
We use same parameter settings for mean,  variances, $\delta$ and $\alpha$ as in Frazier (2014). Table~\ref{tab:conf} gives all six mean-variance configurations and other parameter settings.

\begin{table}
\begin{center}
\caption{Mean and variance configurations}\label{tab:conf}
\begin{tabular}{ccccc} \hline
Configuration & Means& Variances& $\delta$ & $\alpha$\\ \hline
SC-Equal & $\mu = [\delta, 0, \ldots, 0]$ & $\sigma^2 = 100$ & 1 & 0.1\\
MDM-Equal & $\mu_i =  -\delta i$ & $\sigma^2 = 100$ & 1 & 0.1\\
SC-INC &  $\mu = [\delta, 0, \ldots, 0]$ & $\sigma_i^2 = 25 \left( 1 + 3 {i-1 \over k-1}\right)^2$ & 1 & 0.1\\
SC-DEC & $\mu = [\delta, 0, \ldots, 0]$ &  $\sigma_i^2 = 25 \left( 1 + 3 {k-i \over k-1}\right)^2$ & 1 & 0.1\\
MDM-INC & $\mu_i =  -\delta i$ &  $\sigma_i^2 = 25 \left( 1 + 3 {i-1 \over k-1}\right)^2$ & 1 & 0.1\\
MDM-DEC  & $\mu_i =  -\delta i$ & $\sigma_i^2 = 25 \left( 1 + 3 {k-i \over k-1}\right)^2$ & 1 & 0.1\\ \hline
\end{tabular}
\end{center}
\end{table}

When calculating $\eta_{|I|}$ for ${\cal DK}$ procedures,  we take logs to avoid numerical overflows and underflows in the denominator, since the Gamma term can be very large and the Bessel term can be very small.  When $\ell=k-1$ or only two systems are survived, we use $\eta_2 = -\ln(2\beta_\ell)$.

The nominal confidence level is set to $1-\alpha = 0.9$. Estimated probability of correct selection (PCS) and an average number of observations per system until a decision is made (REP/$k$) are reported based on 10,000 macro replications. Standard errors for estimated PCS are approximately 0.003.

\subsection{${\cal DK}_1$ with Known and Equal Variances}

When variances are known and equal, we compare ${\cal DK}_1$ with KN and BIZ. Figure~\ref{fig:knownequal} shows REP/$k$ and PCS under SC and MDM configurations.
Procedure ${\cal DK}_1$ significantly outperforms KN under both SC and MDM. When $k$ is large, ${\cal DK}_1$ is more than three times better than KN in terms of REP/$k$.
On the other hand, the performances of BIZ and ${\cal DK}_1$ are very similar under the slippage configuration in terms of both REP/$k$ and PCS. When $k$ is small, ${\cal DK}_1$ spends a slightly more number of observations than BIZ but its probability of correct selection is slightly higher than BIZ. For large $k$, their performances are very close in both measures. Under the monotonic decreasing mean configuration, ${\cal DK}_1$ achieves PCS greater than the nominal value 90\% and clearly outperforms KN. However, BIZ achieves PCS close to the nominal value 90\% than ${\cal DK}_1$ and spends slightly fewer but very similar number of observations than ${\cal DK}_1$.

\begin{figure}[t!]
\begin{subfigure}{.5\textwidth}
  \centering
  \includegraphics[width=\linewidth]{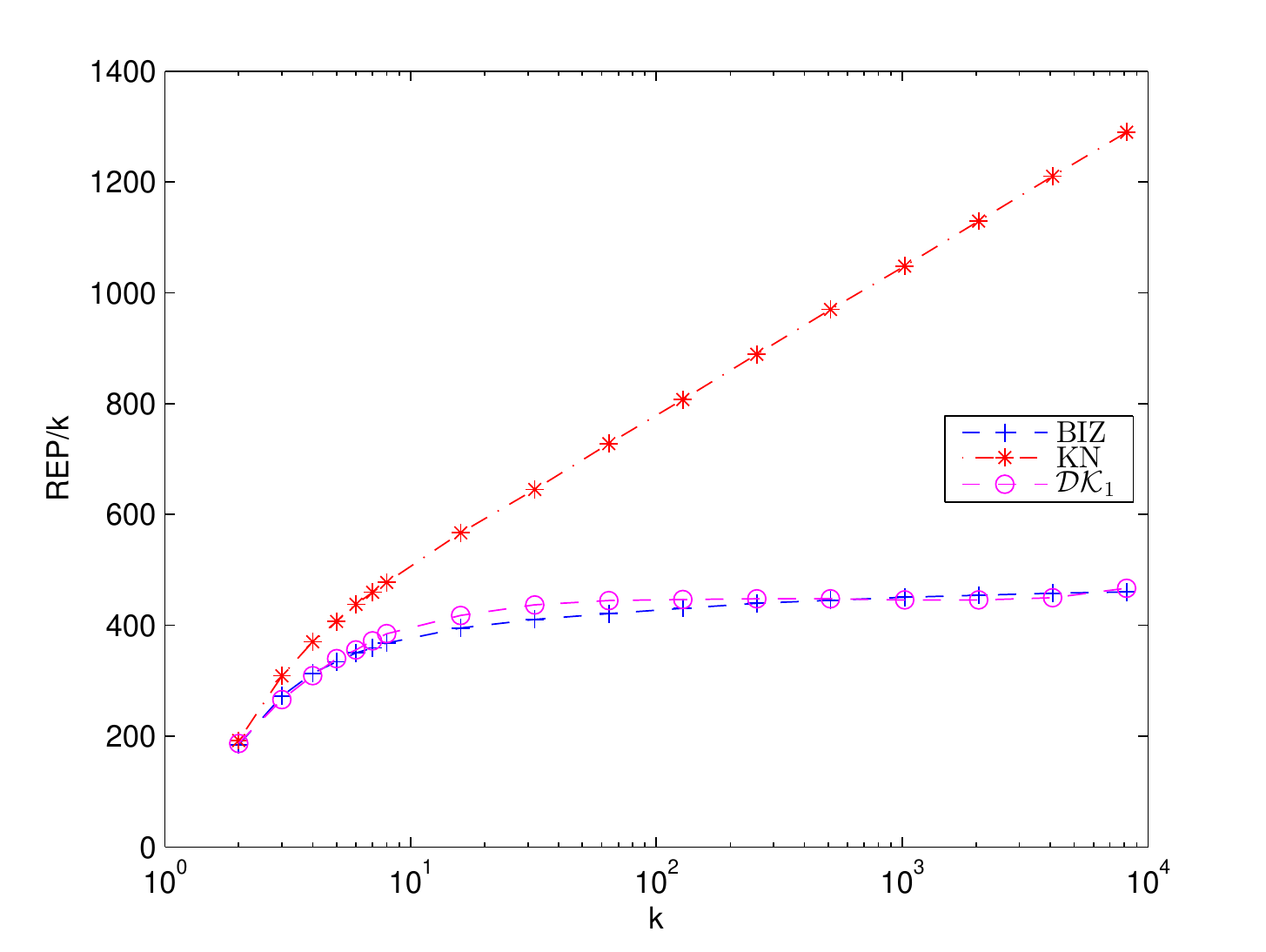}
\caption{ SC-REP}
\end{subfigure}%
\begin{subfigure}{.5\textwidth}
  \centering
  \includegraphics[width=\linewidth]{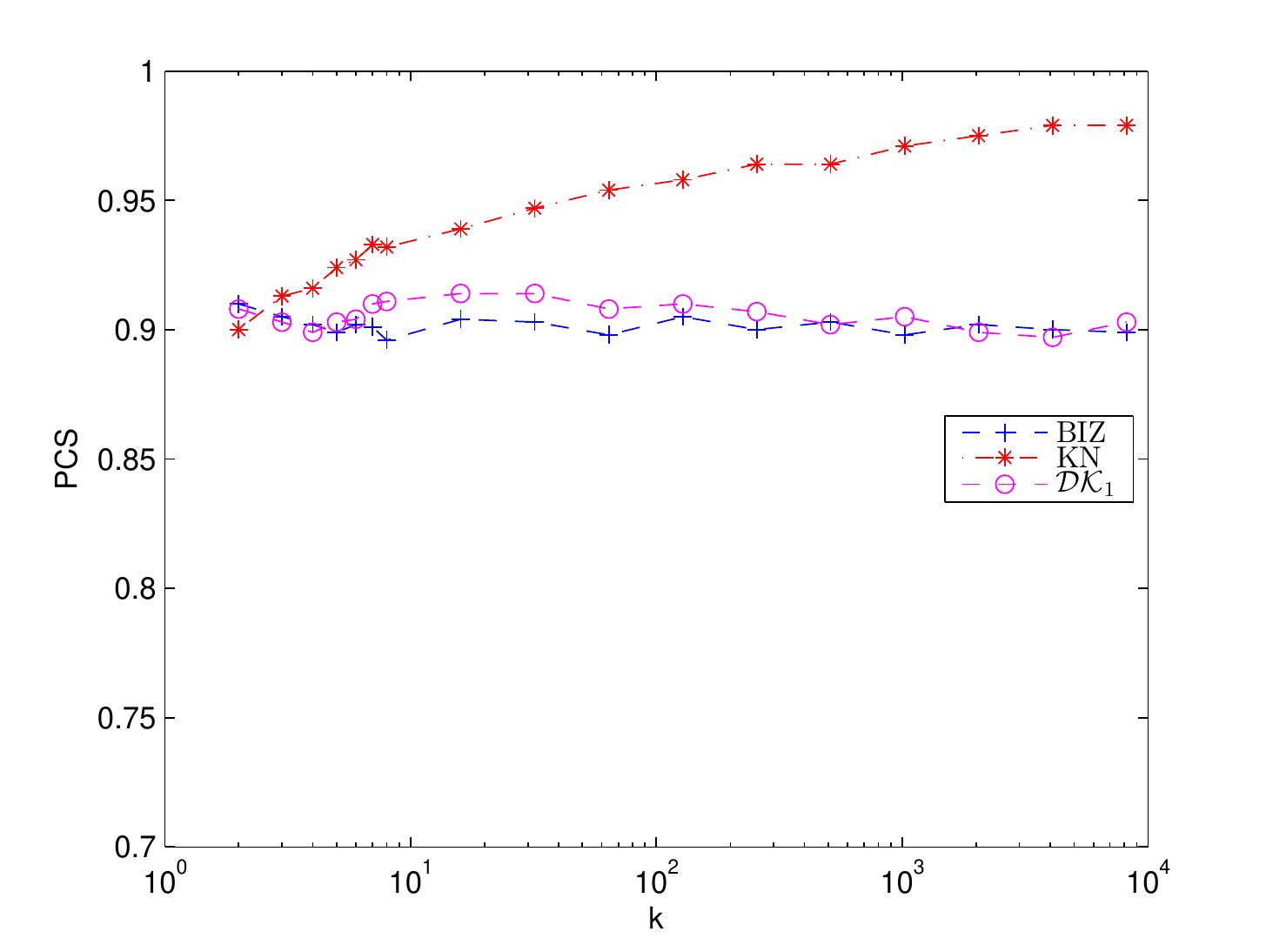}
\caption{ SC-PCS}
\end{subfigure}
\begin{subfigure}{.5\textwidth}
  \centering
  \includegraphics[width=\linewidth]{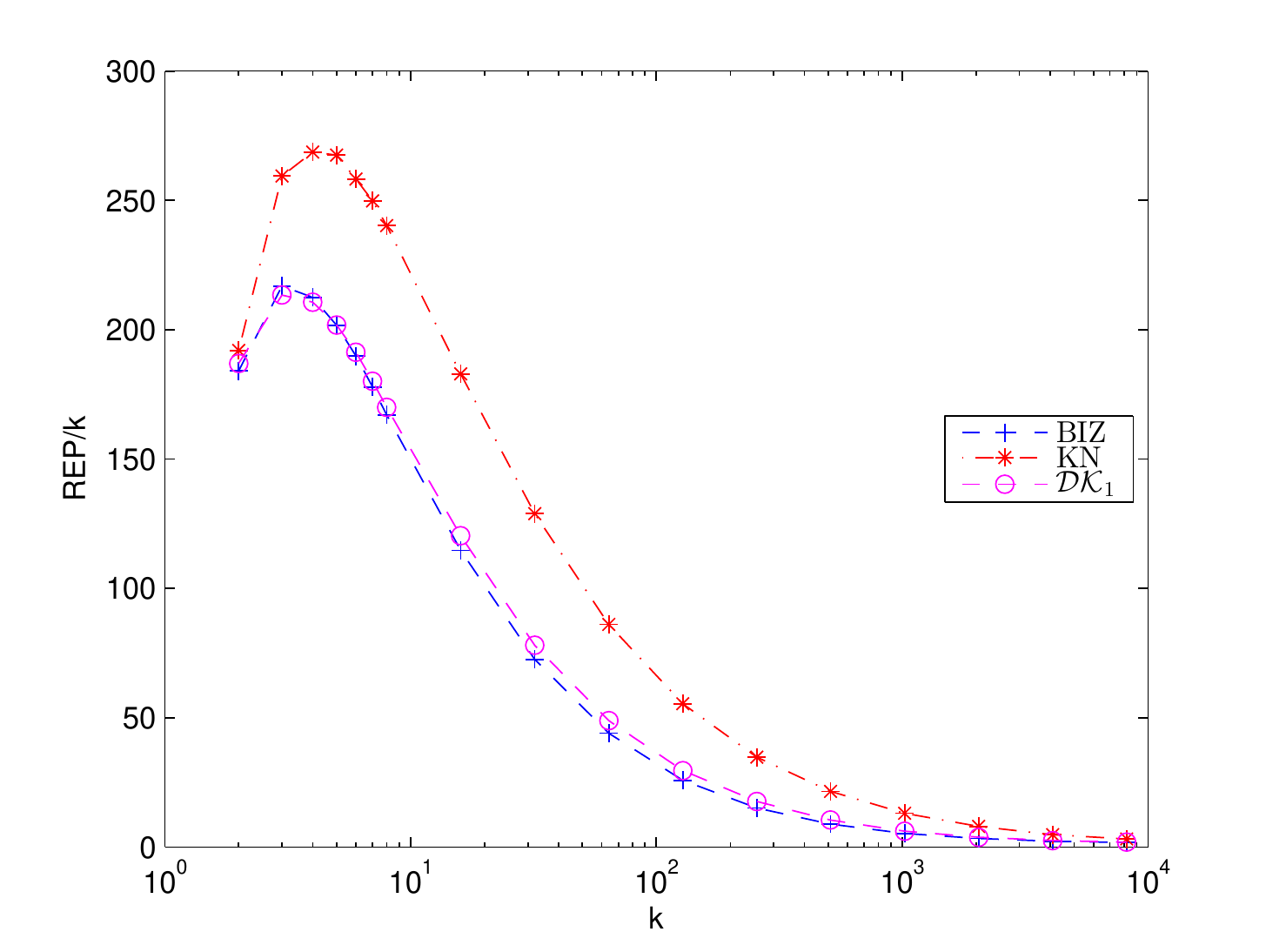}
\caption{ MDM-REP}
\end{subfigure}%
\begin{subfigure}{.5\textwidth}
  \centering
  \includegraphics[width=\linewidth]{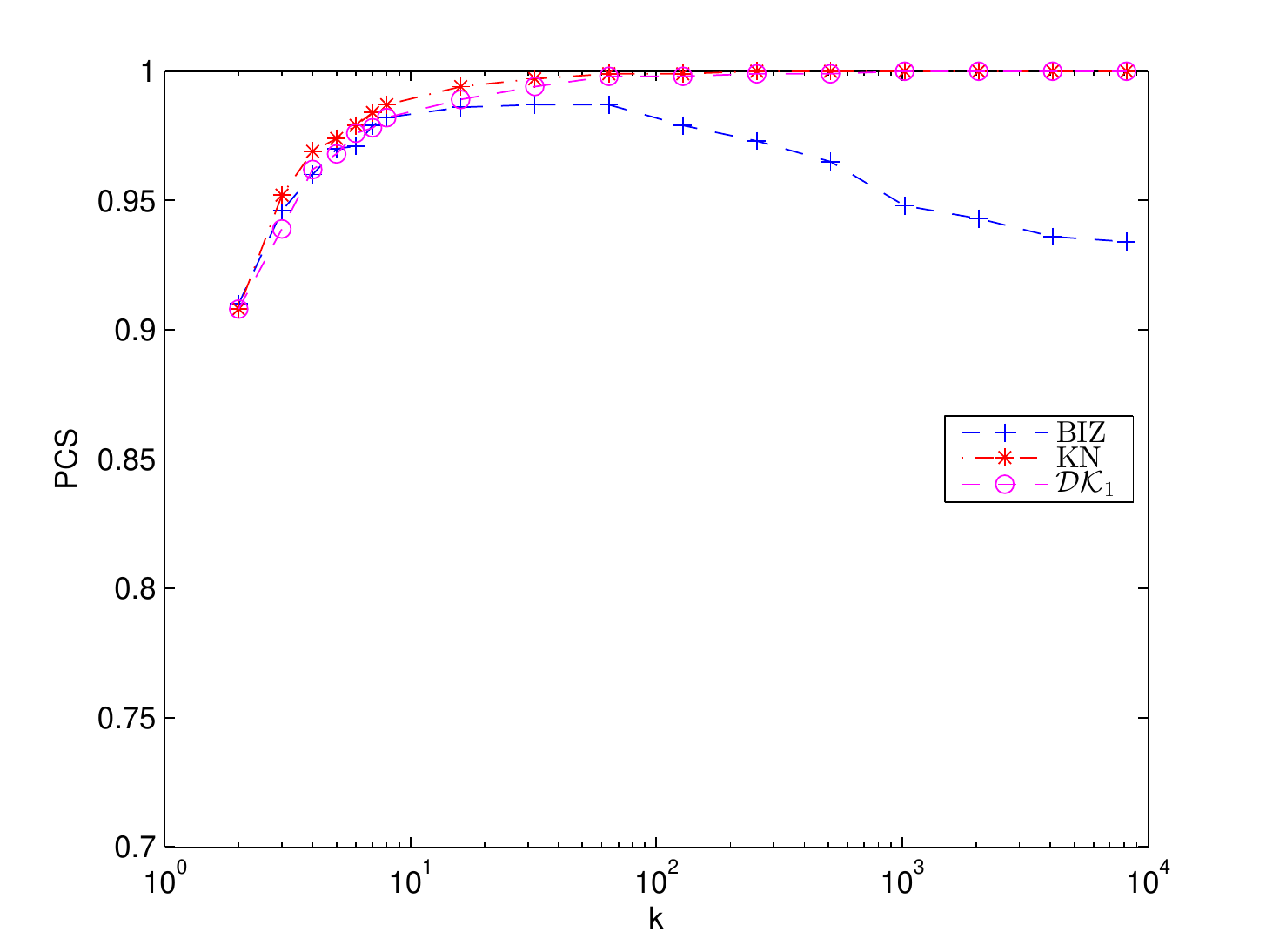}
\caption{ MDM-PCS}
\end{subfigure}
\caption{REP/$k$ and PCS for ${\cal DK}_1$ when variances are known and equal with $1-\alpha = 0.9$
}
\label{fig:knownequal}
\end{figure}

\subsection{${\cal DK}_2$ and ${\cal DK}_3$ with Unknown but Equal Variances}

When variances are unknown but a decision maker knows that variances across systems are equal, ${\cal DK}_2$ or ${\cal DK}_3$ can be used.

Figure~\ref{fig:unknownequal} compares performances of ${\cal DK}_2$ with those of KN-UNK and BIZ-UNK. As in the case of known and equal variances, ${\cal DK}_2$ significantly outperforms KN-UNK. Compared to BIZ-UNK, ${\cal DK}_2$ achieves slightly higher PCS and spends fewer number of observations for large $k$ under the slippage configurations. Then under the monotonic decreasing configuration, PCS is higher in ${\cal DK}_2$ and uses slightly more observations for small $k$ and then similar number of observations for large $k$.

\begin{figure}[tb!]
\begin{subfigure}{.5\textwidth}
  \centering
  \includegraphics[width=\linewidth]{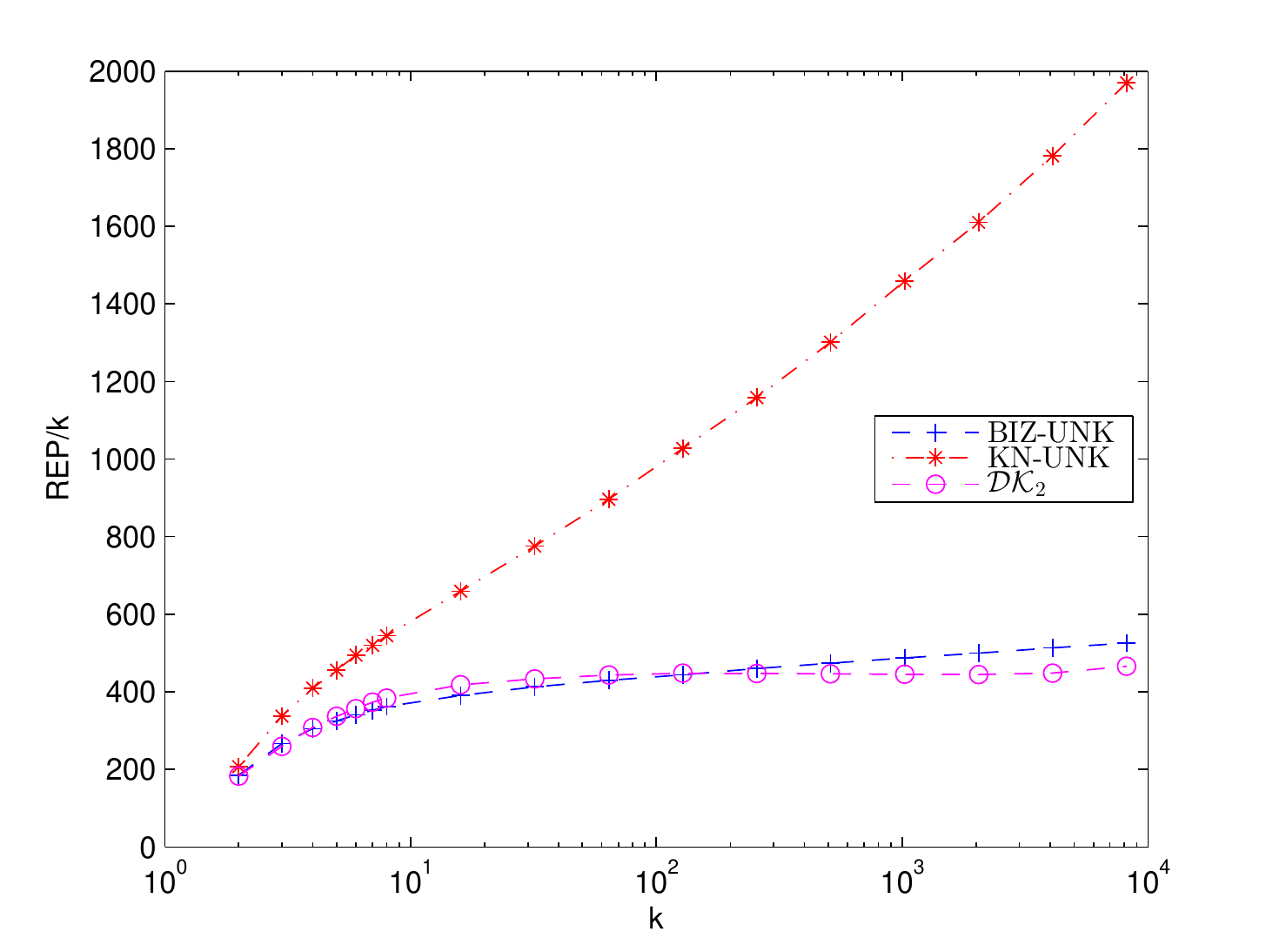}
\caption{ SC-REP}
\end{subfigure}%
\begin{subfigure}{.5\textwidth}
  \centering
  \includegraphics[width=\linewidth]{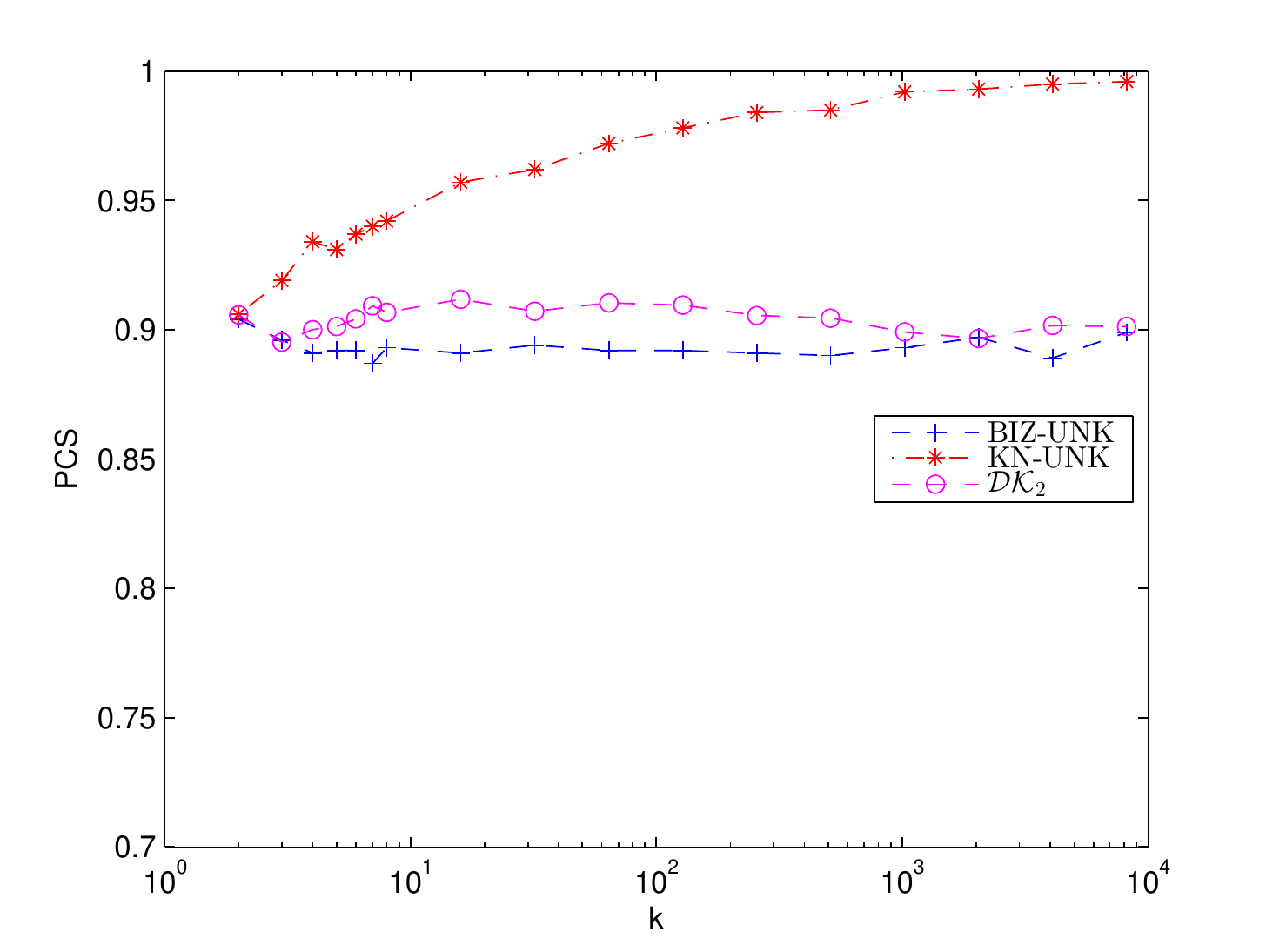}
\caption{ SC-PCS}
\end{subfigure}
\begin{subfigure}{.5\textwidth}
  \centering
  \includegraphics[width=\linewidth]{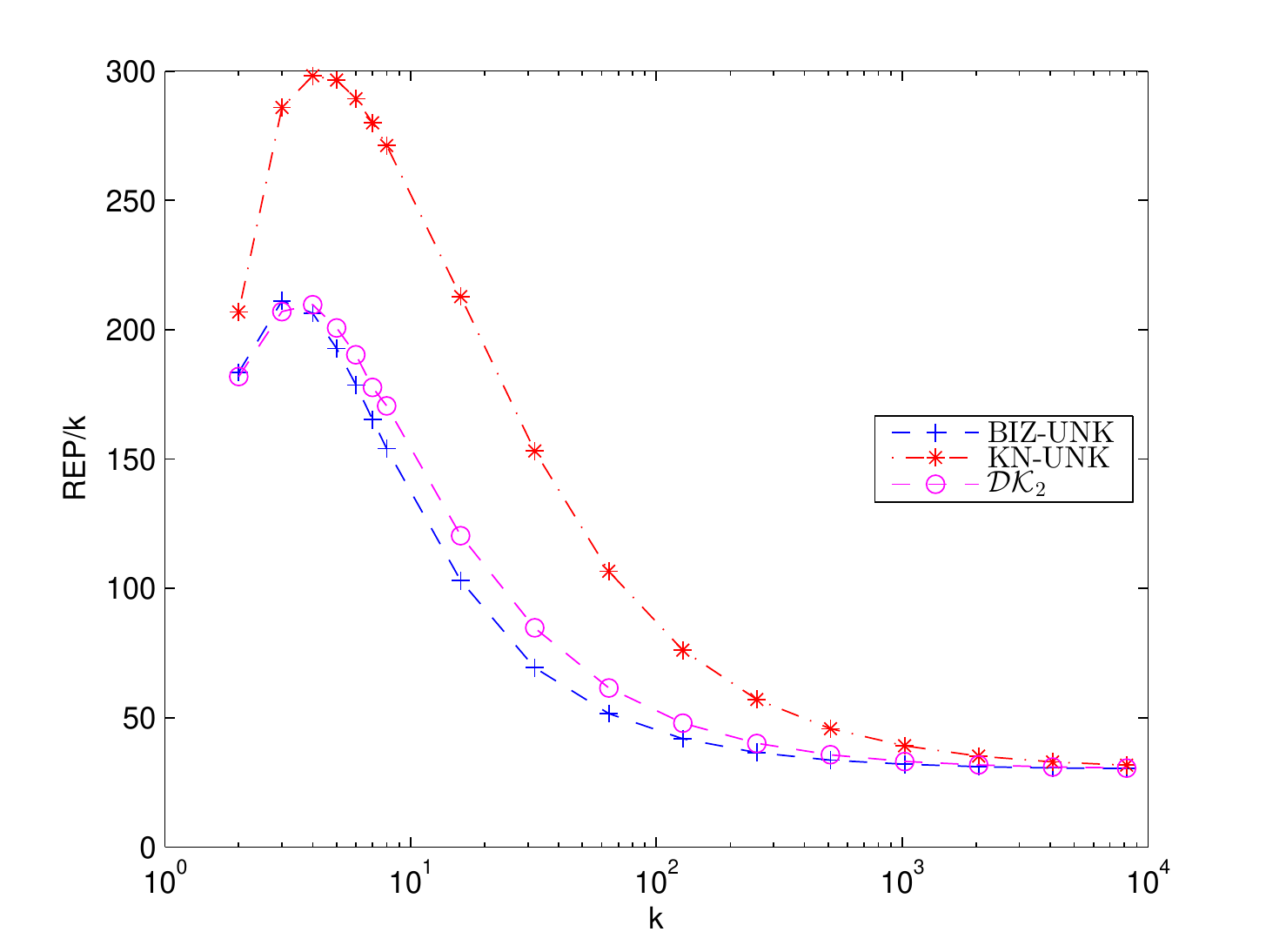}
\caption{ MDM-REP}
\end{subfigure}%
\begin{subfigure}{.5\textwidth}
  \centering
  \includegraphics[width=\linewidth]{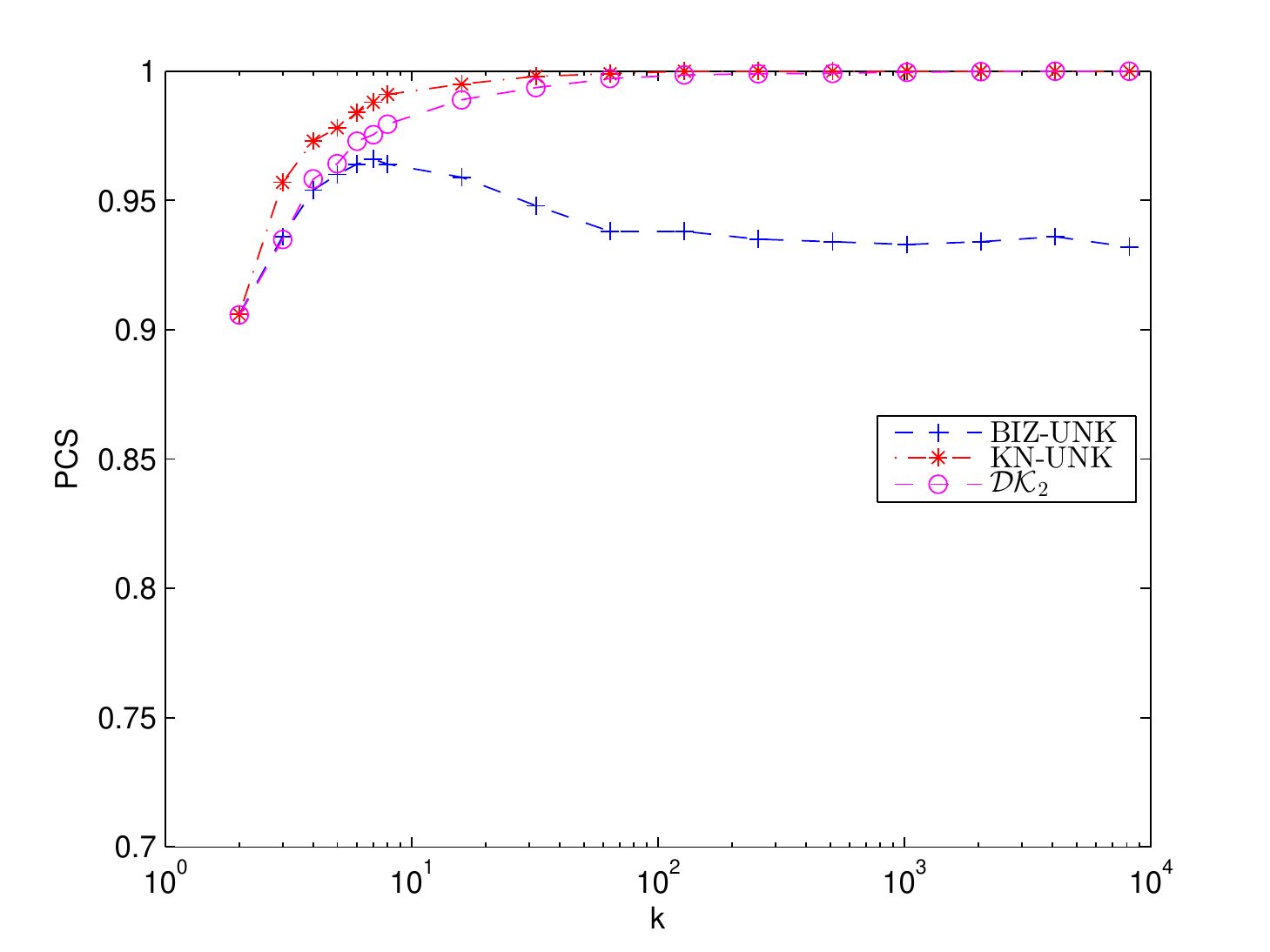}
\caption{ MDM-PCS}
\end{subfigure}
\caption{REP/$k$ and PCS for ${\cal DK}_2$ when variances are unknown but equal with $1-\alpha = 0.9$
}
\label{fig:unknownequal}
\end{figure}

In reality, it is impossible to know in advance whether variances across systems are equal. In fact, equal variances across systems rarely hold. Thus we also consider ${\cal DK}_3$. Our experiments show that ${\cal DK}_3$ actually spends slightly fewer observations than ${\cal DK}_2$ while achieving similar PCS. Figure~\ref{fig:unknownequalvaryingfreq} compares ${\cal DK}_3$ with KN-UNK and BIZ-UNK. Graphs in Figure~\ref{fig:unknownequalvaryingfreq} show similar tendency as those in Figure~\ref{fig:unknownequal}.

\begin{figure}[tb!]
\begin{subfigure}{.5\textwidth}
  \centering
  \includegraphics[width=\linewidth]{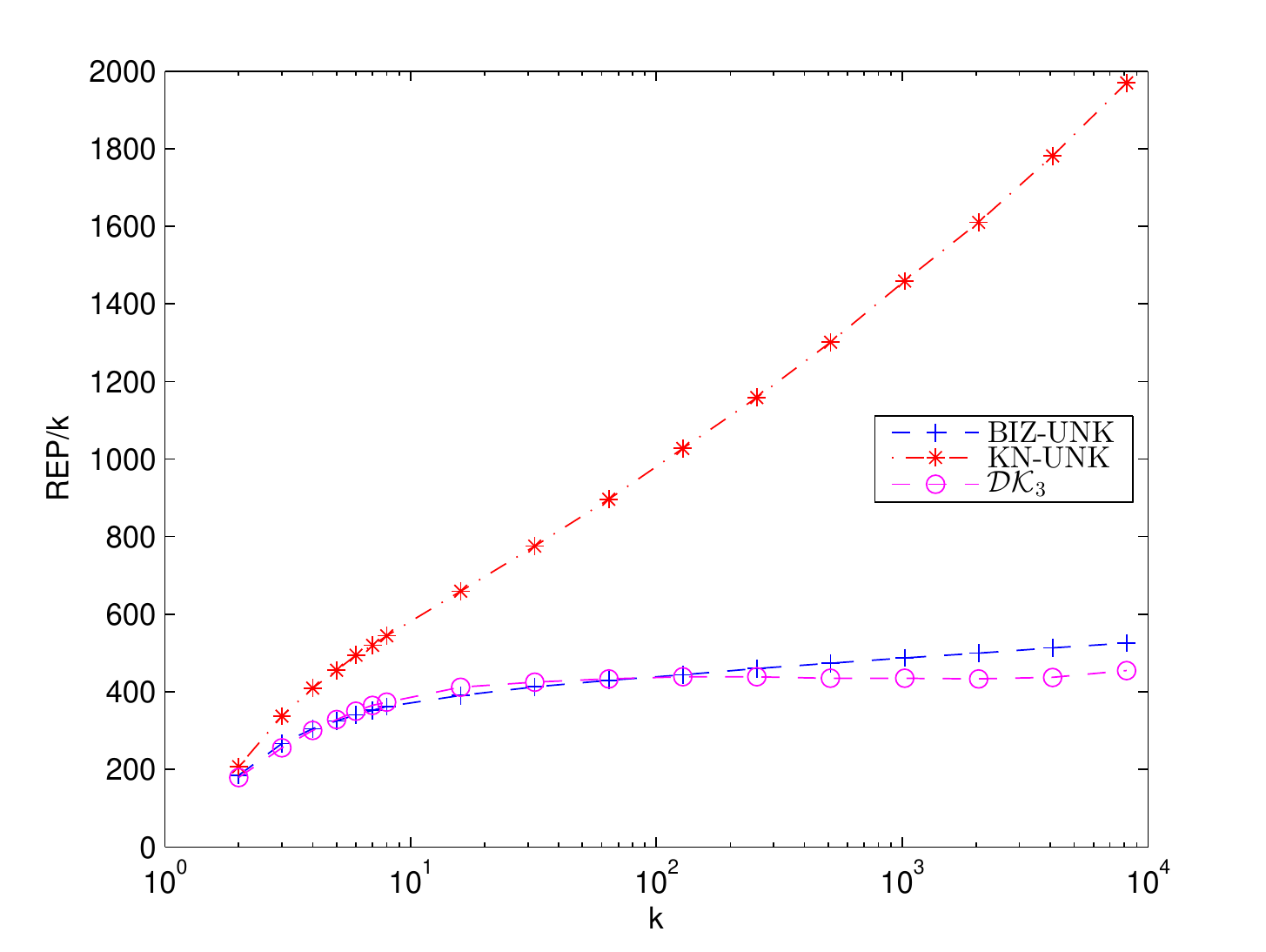}
\caption{SC-REP}
\end{subfigure}%
\begin{subfigure}{.5\textwidth}
  \centering
  \includegraphics[width=\linewidth]{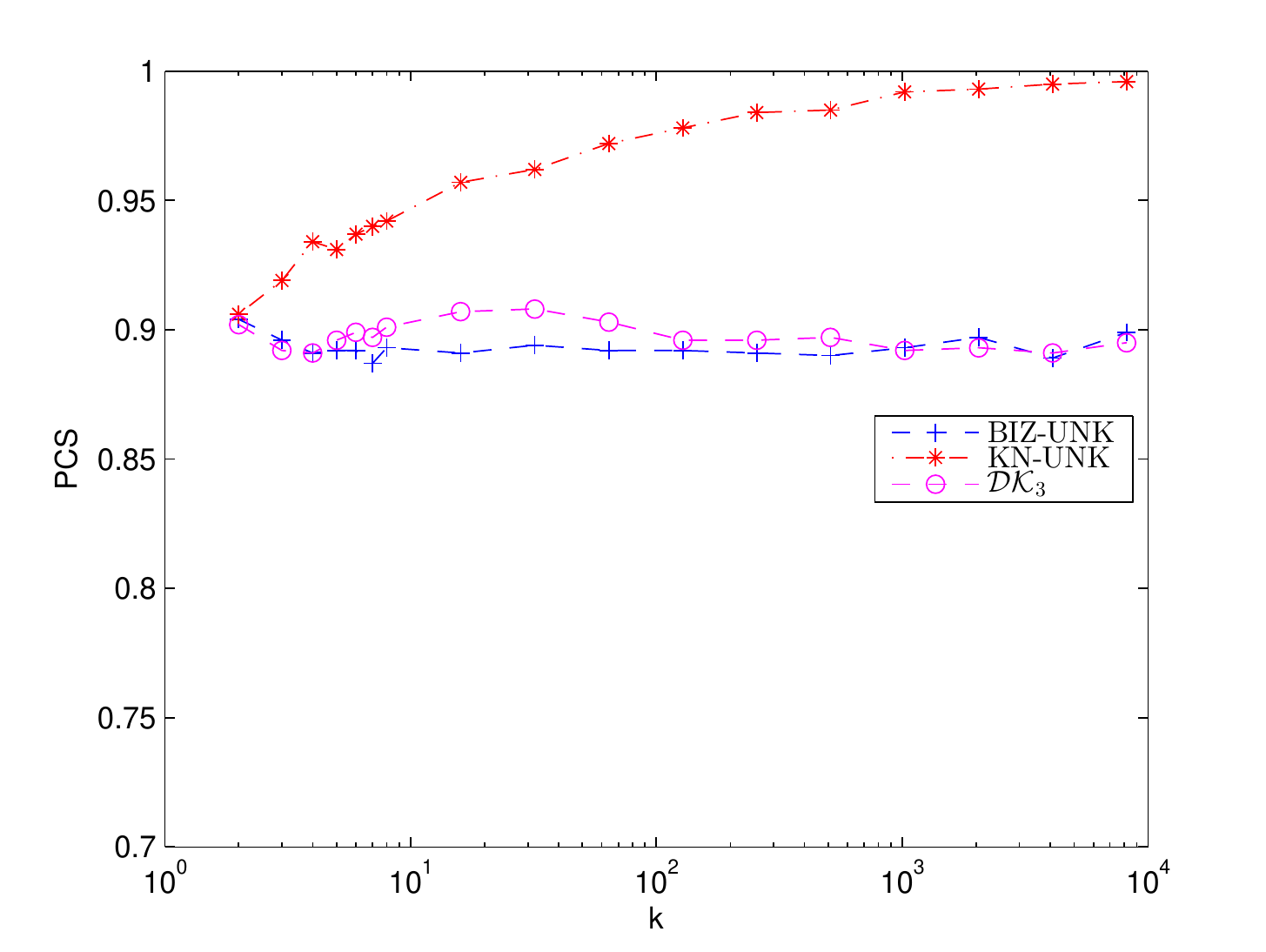}
\caption{SC-PCS}
\end{subfigure}
\begin{subfigure}{.5\textwidth}
  \centering
  \includegraphics[width=\linewidth]{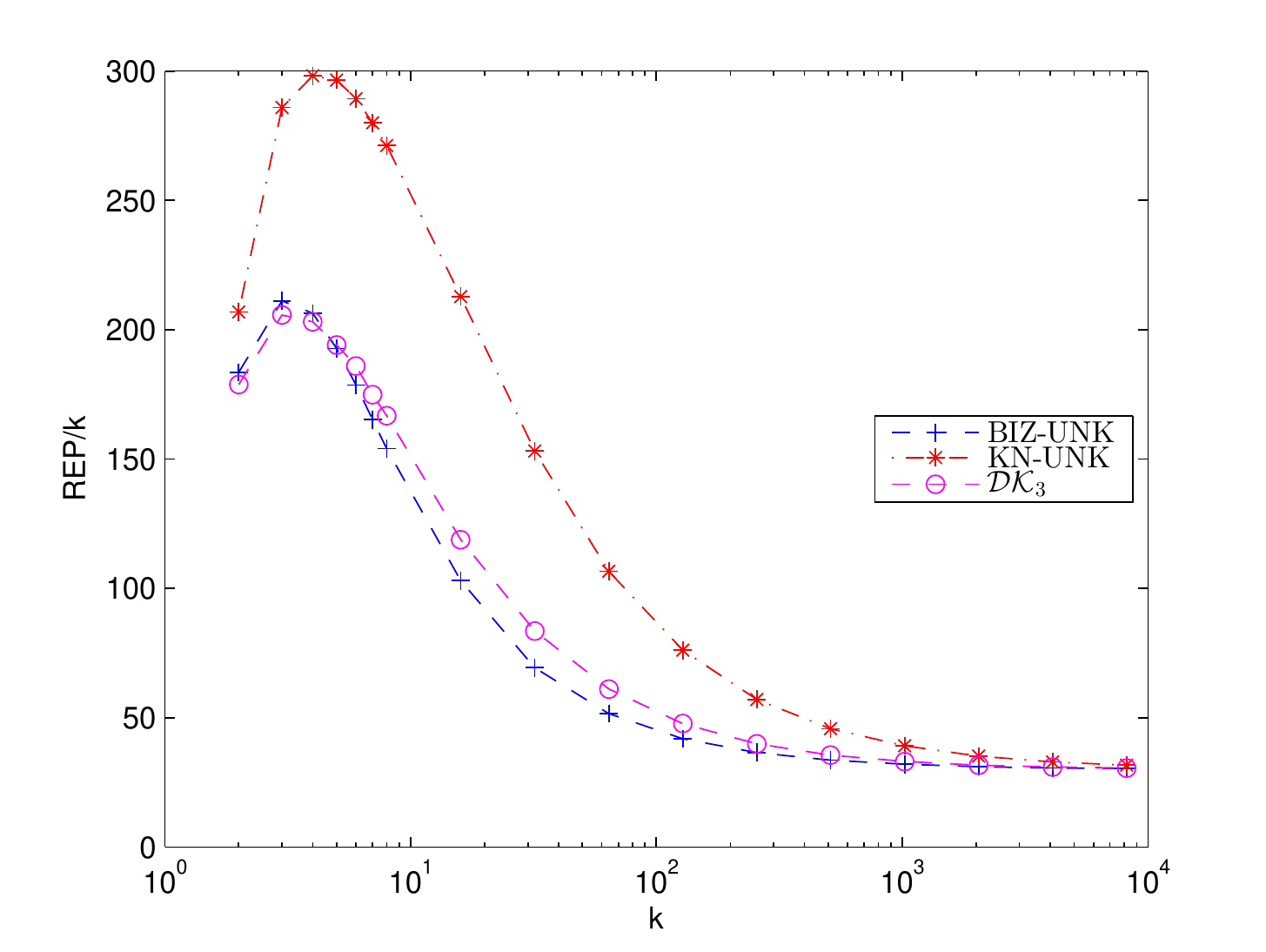}
\caption{ MDM-REP}
\end{subfigure}%
\begin{subfigure}{.5\textwidth}
  \centering
  \includegraphics[width=\linewidth]{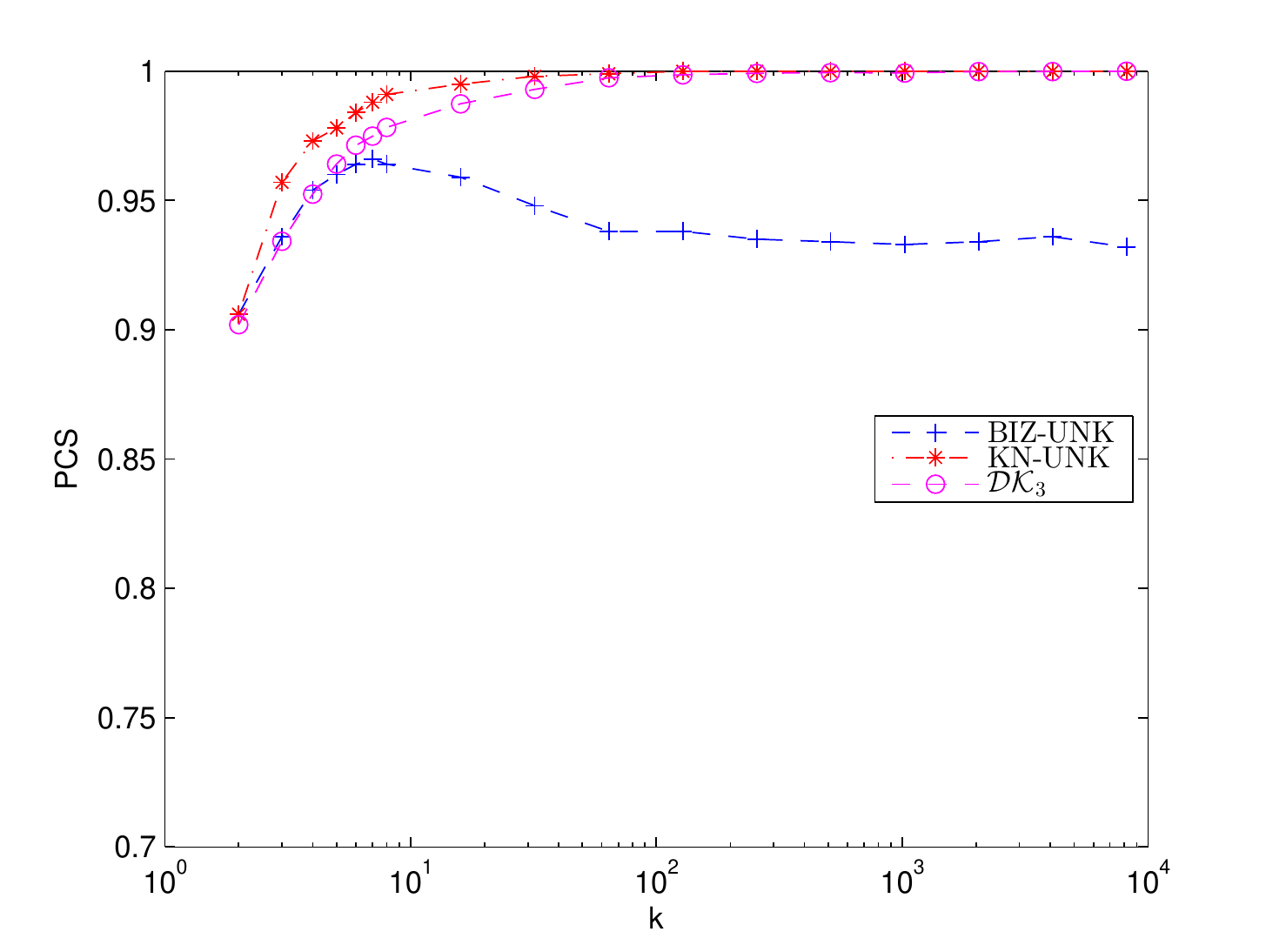}
\caption{ MDM-PCS}
\end{subfigure}
\caption{REP/$k$ and PCS for ${\cal DK}_3$ when variances are unknown but equal with $1-\alpha = 0.9$
}
\label{fig:unknownequalvaryingfreq}
\end{figure}

\subsection{${\cal DK}_3$ with Unknown and Unequal Variances}

Finally, we consider unknown and unequal variances. Figure~\ref{fig:incdec_SC} compares the three procedures under the slippage configuration with increasing and decreasing variances while  Figure~\ref{fig:incdec_MDM} compares them under the MDM configuration with increasing and decreasing variances.

The efficiency of ${\cal DK}_3$ compared to KN-UNK is more obvious. When $k=8192$, ${\cal DK}_3$ is more than four times better than KN-UNK under SC-INC and six times better under SC-DEC in terms of REP/$k$ while achieving PCS close to 90\%. ${\cal DK}_3$ spends up to 30\% fewer observations than BIZ-UNK under the slippage configuration.

Interestingly, under the MDM configuration with increasing variances, ${\cal DK}_3$ significantly outperforms both KN-UNK and BIZ-UNK, showing up to 63\% savings in the number of observations compared to BIZ-UNK. But ${\cal DK}_3$ uses slightly more observations than BIZ-UNK under decreasing variances.

\begin{figure}[tb!]
\begin{subfigure}{.5\textwidth}
  \centering
  \includegraphics[width=\linewidth]{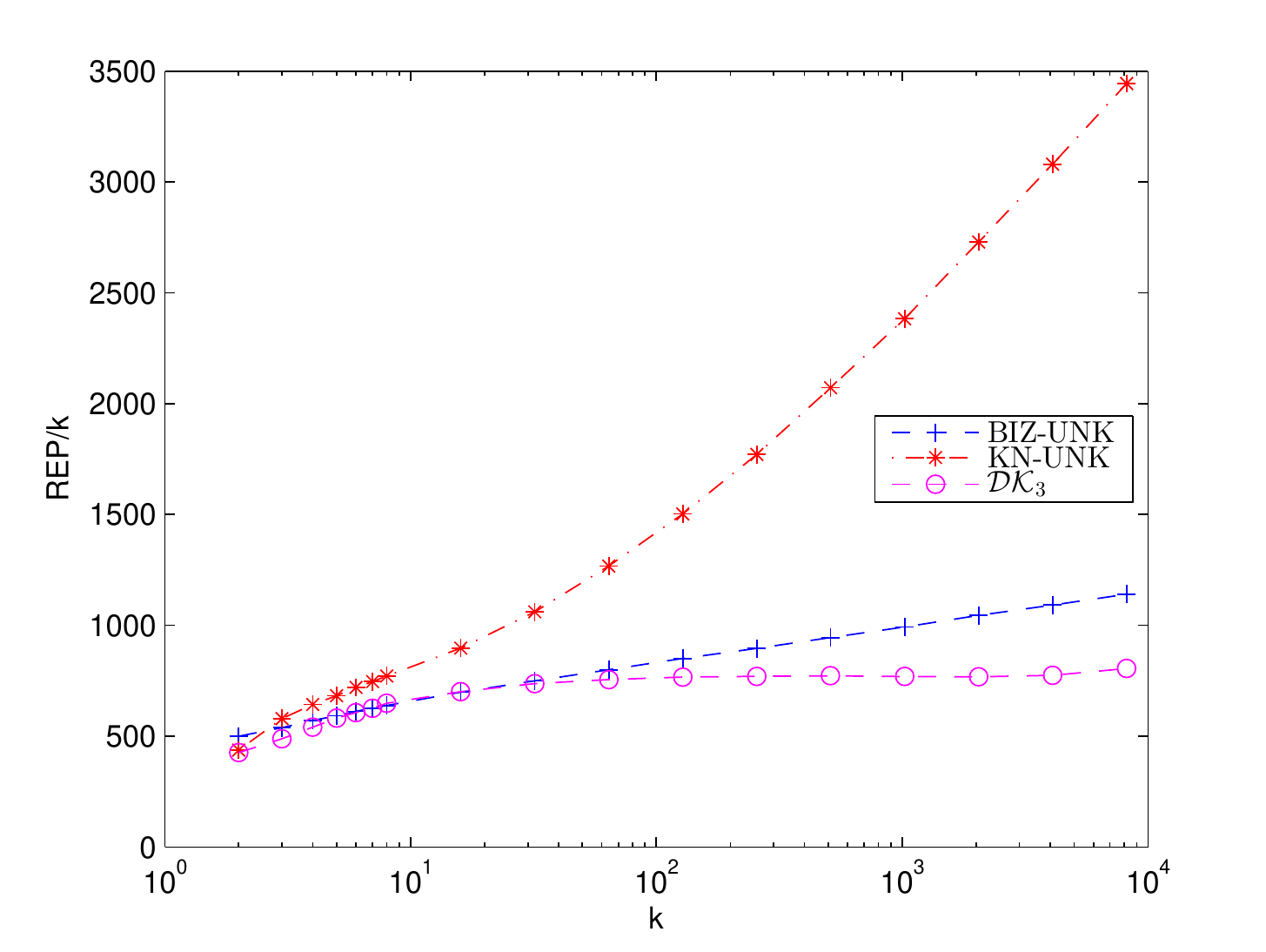}
\caption{SC-INC-REP}
\end{subfigure}%
\begin{subfigure}{.5\textwidth}
  \centering
  \includegraphics[width=\linewidth]{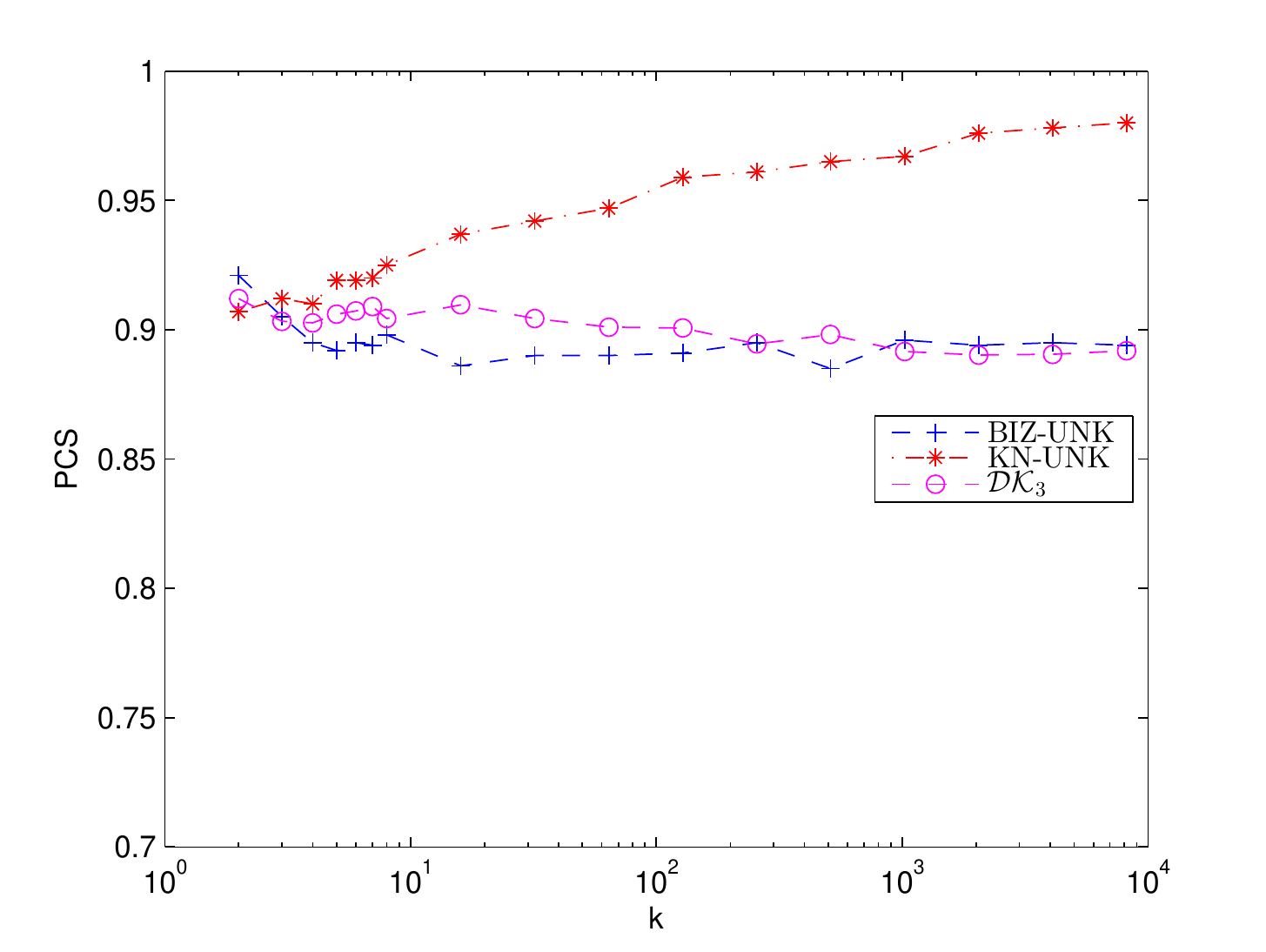}
\caption{SC-INC-PCS}
\end{subfigure}

\begin{subfigure}{.5\textwidth}
  \centering
  \includegraphics[width=\linewidth]{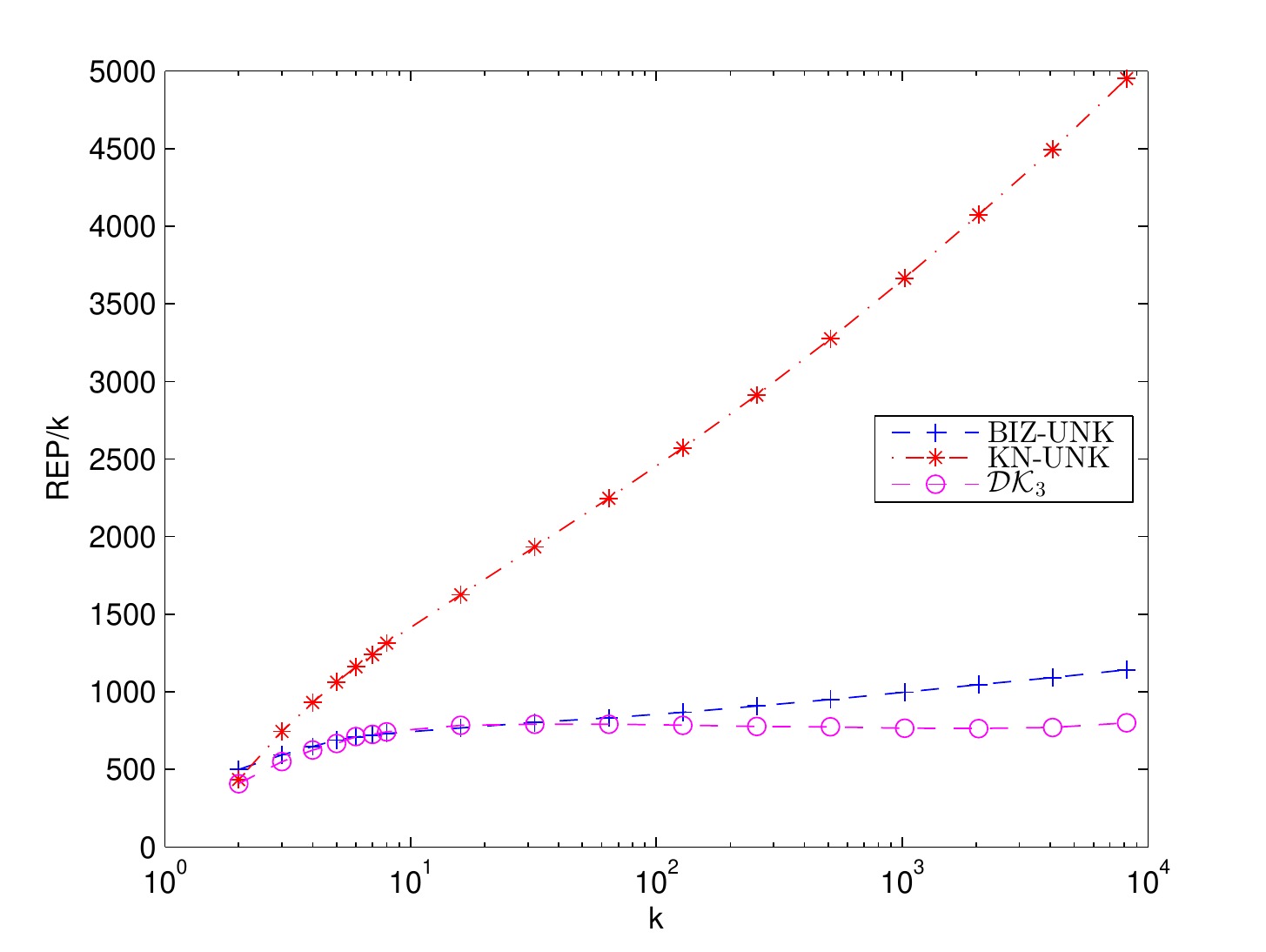}
\caption{SC-DEC-REP}
\end{subfigure}%
\begin{subfigure}{.5\textwidth}
  \centering
  \includegraphics[width=\linewidth]{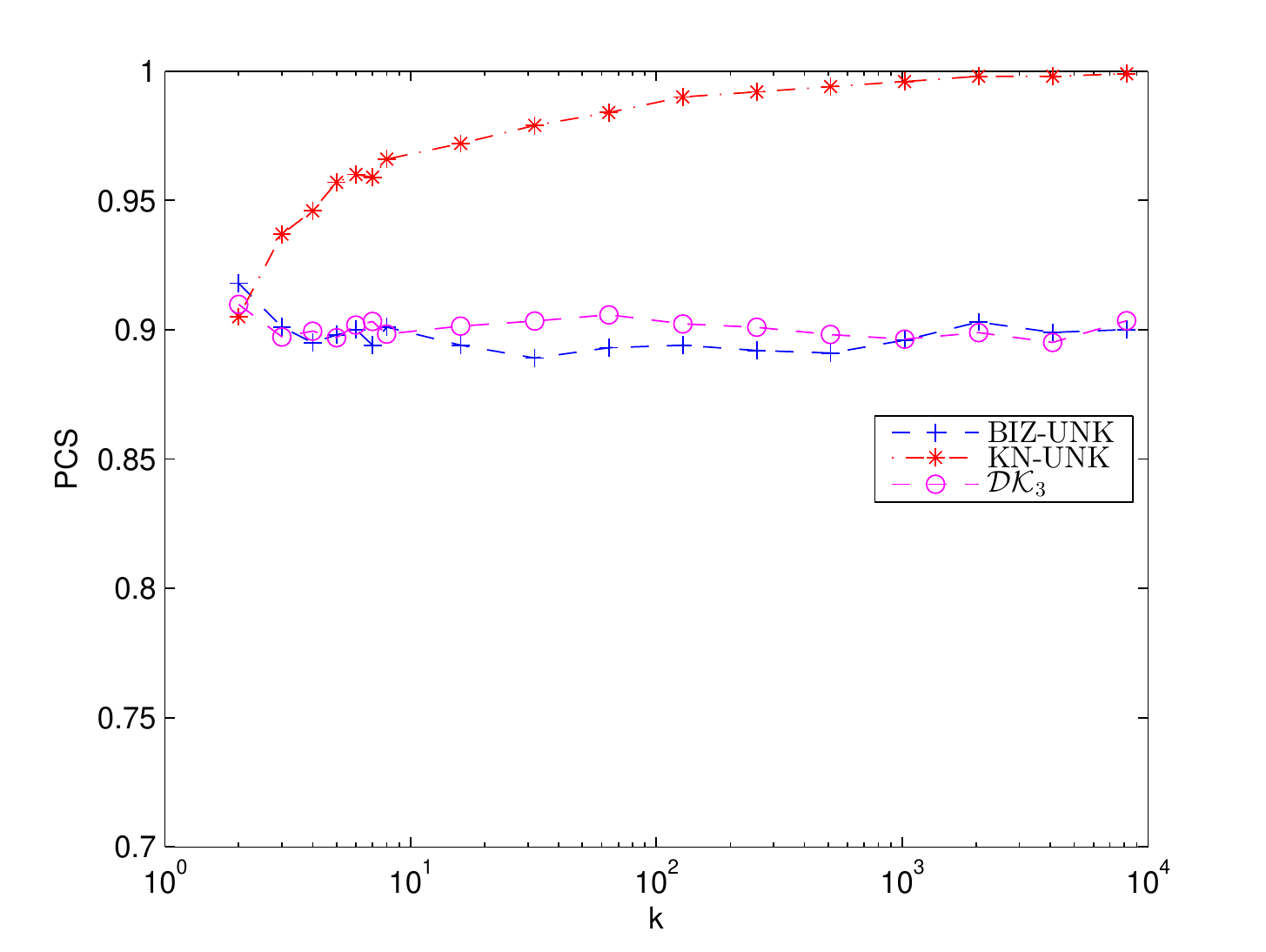}
\caption{SC-DEC-PCS}
\end{subfigure}
\caption{REP/$k$ and PCS for ${\cal DK}_3$ when variances are unknown and unequal with $1-\alpha = 0.9$
}
\label{fig:incdec_SC}
\end{figure}

\begin{figure}[h!]
\begin{subfigure}{.5\textwidth}
  \centering
  \includegraphics[width=\linewidth]{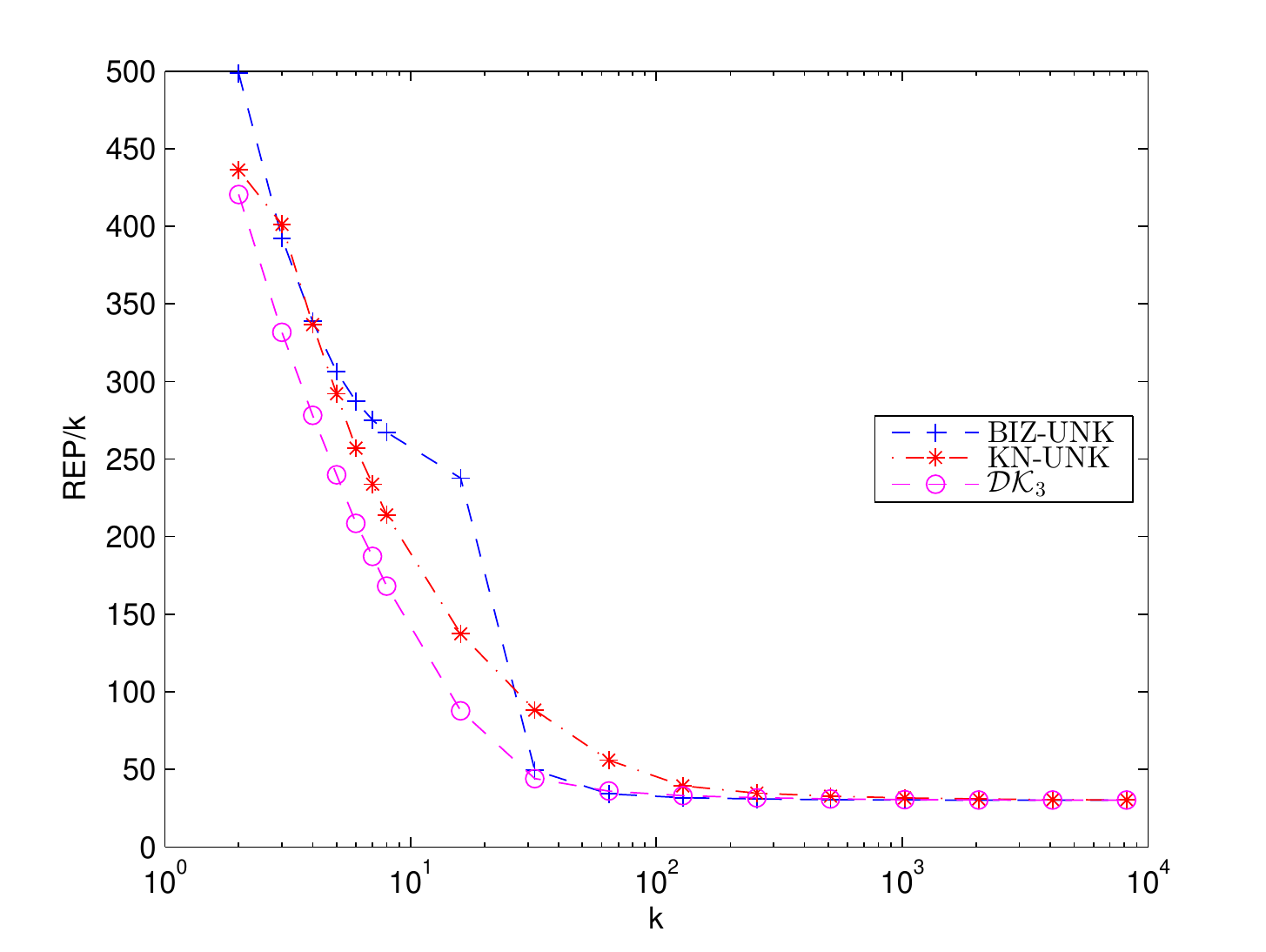}
\caption{ MDM-INC-REP}
\end{subfigure}%
\begin{subfigure}{.5\textwidth}
  \centering
  \includegraphics[width=\linewidth]{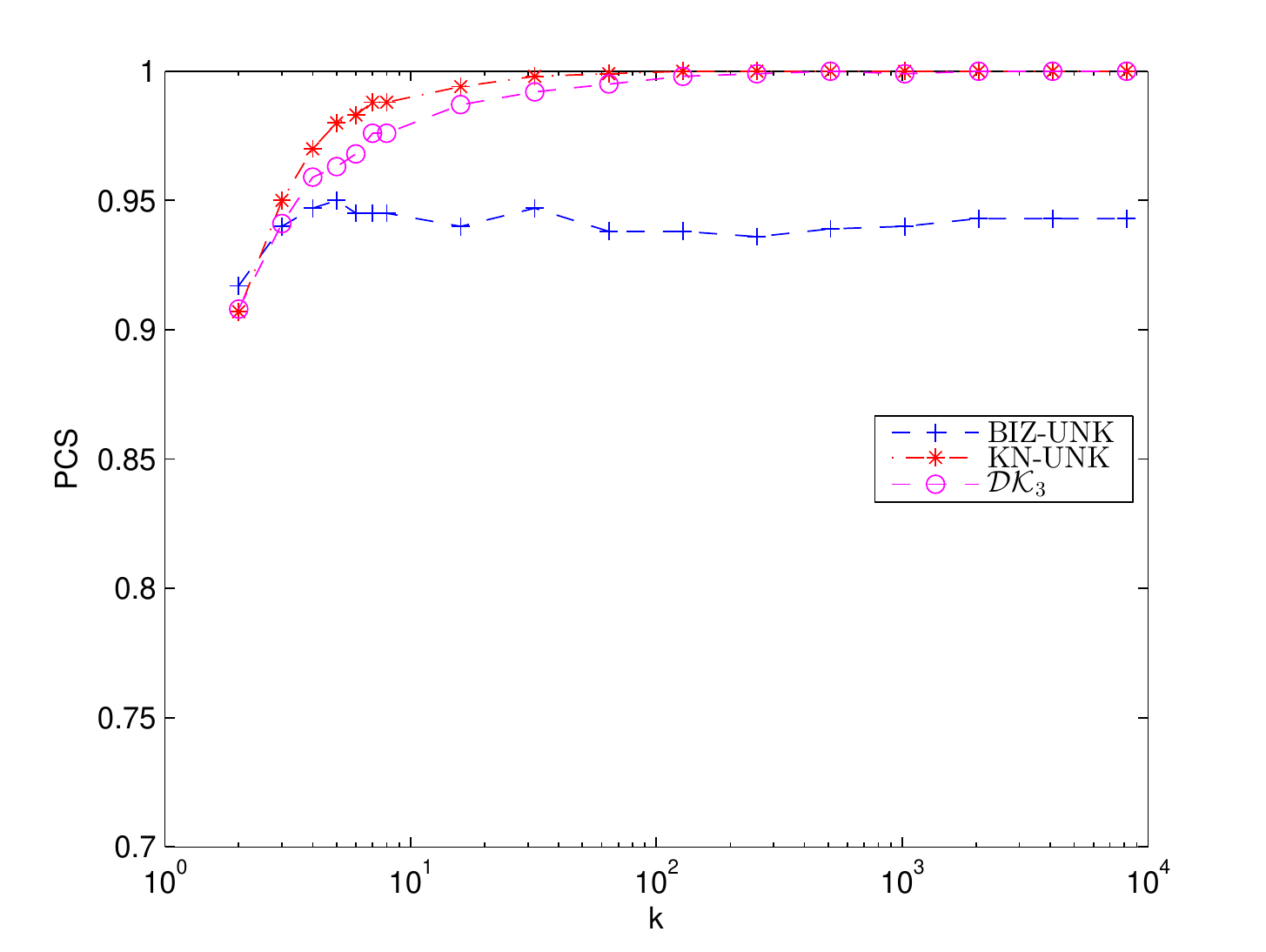}
\caption{ MDM-INC-PCS}
\end{subfigure}
\begin{subfigure}{.5\textwidth}
  \centering
  \includegraphics[width=\linewidth]{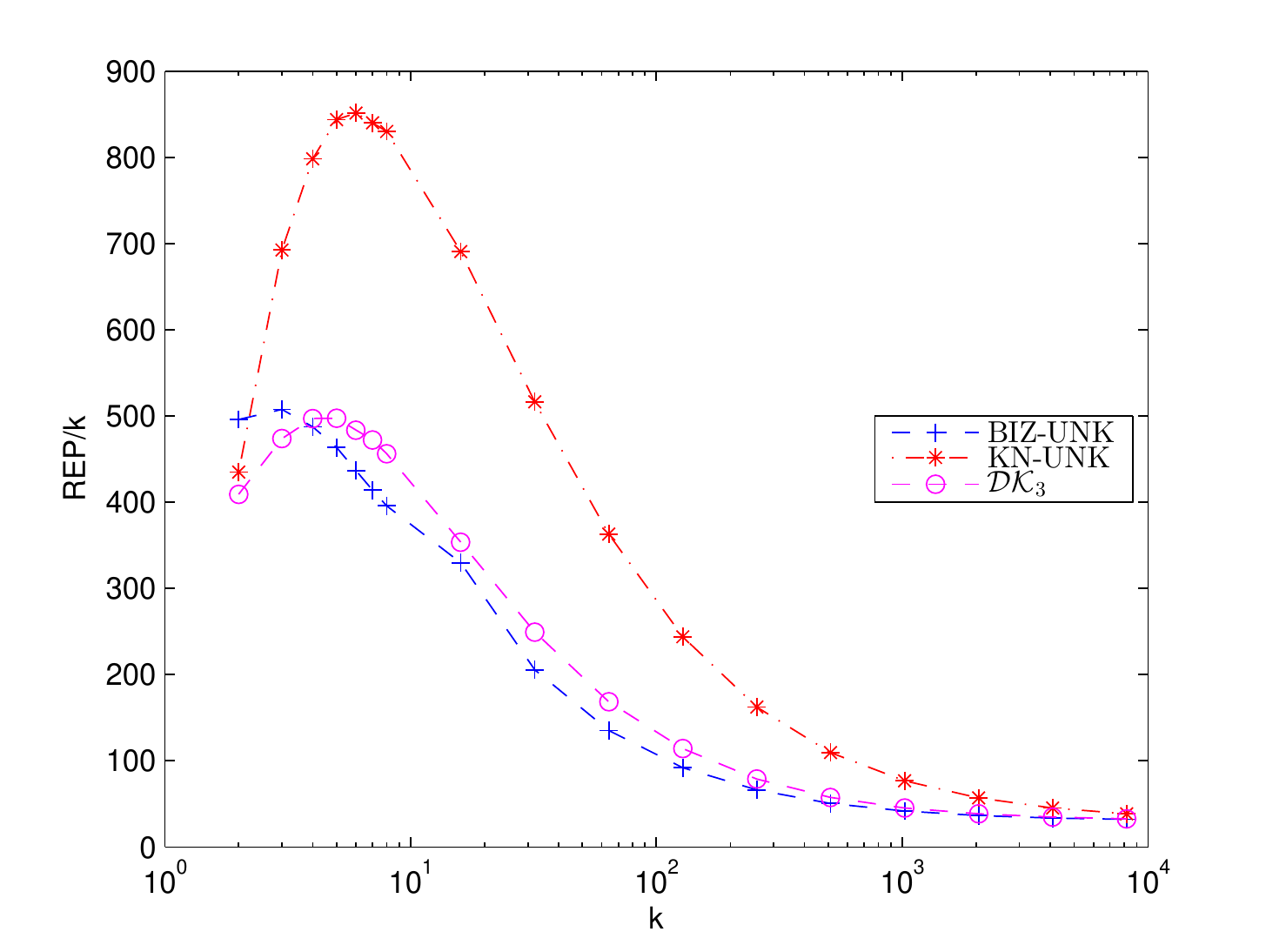}
\caption{MDM-DEC-REP}
\end{subfigure}%
\begin{subfigure}{.5\textwidth}
  \centering
  \includegraphics[width=\linewidth]{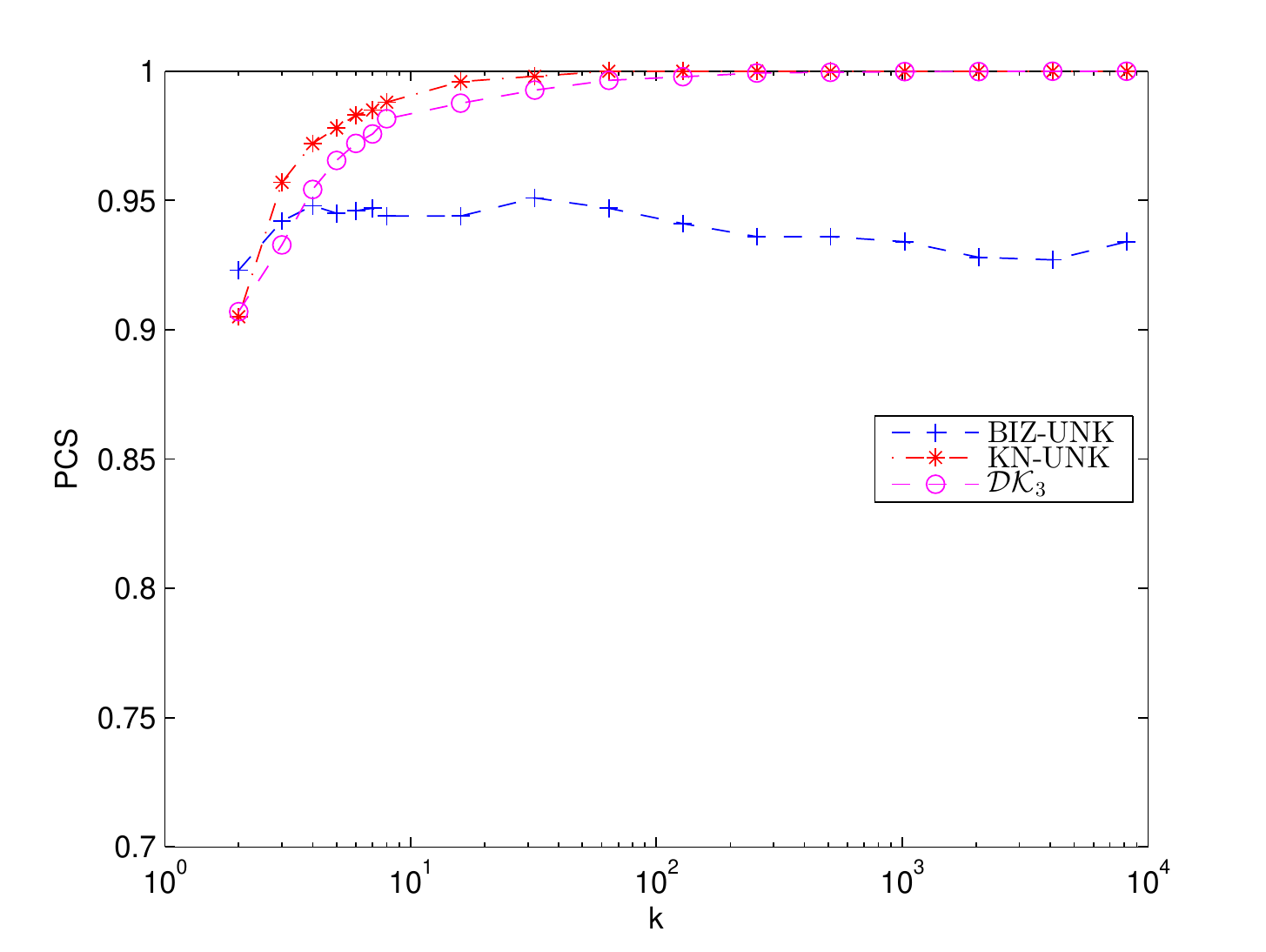}
\caption{MDM-DEC-PCS}
\end{subfigure}

\caption{REP/$k$ and PCS for ${\cal DK}_3$ when variances are unknown and unequal with $1-\alpha = 0.9$
}
\label{fig:incdec_MDM}
\end{figure}

Overall, ${\cal DK}$ procedures achieve PCS close to the nominal value for all settings we tested and they outperform KN significantly while performing similarly to BIZ under easy mean configurations but outperforming it under difficult mean configurations especially with unknown and unequal variances.

\section{Conclusions}
\label{sec:conclusion}

We present new fully-sequential procedures  whose continuation regions are derived exploiting the properties of multidimensional Brownian motions, which is the first work in the literature. Our procedures deliver a probability of correct selection close to the nominal level. Compared to the existing state-of-art fully-sequential IZ procedure KN, the proposed procedures show a tight worst-case probability of incorrect selection under the slippage configuration and significant savings in the number of observations needed until a decision is made. Compared to BIZ, our procedures perform better for a large number of systems under difficult mean configurations, but tend to spend slightly more observations for small $k$ but similar number of observations for large $k$ under easier mean configurations except the increasing-variances case.

\section*{Acknowledgements}

This work is supported by the National Science Foundation under grant CMMI-1131047. The authors would like to thank Seunghan Lee for his
insight for Lemma 1. The authors appreciate Peter Frazier for his code and helpful comments and Barry Nelson for his helpful comments.


\section*{References}
\begin{hangref}

\item Branke, J., Chick, S. E., and Schmidt, C. 2007. ``Selecting a selection procedure''. {\it Management Science} 53(12):1916-1932.

\item Chick, S. E. 2006. “Subjective Probability and Bayesian Methodology”. In {\it Handbooks in Operations
Research and Management Science: Simulation}, edited by S. G. Henderson and B. L. Nelson. Oxford:
Elsevier Science.

\item Chen, C.-H., S. E. Chick, L. H. Lee, N. A. Pujowidianto. 2014. ``Ranking and Selection: Efficient Simulation Budget Allocation''. In {\it Handbook of Simulation Optimization}, edited by M. C. Fu. Springer:NY.

\item Chen, C.-H., and L. H. Lee. 2011. {\it Stochastic Simulation Optimization: An Optimal Computing Budget Allocation (System Engineering and Operations Research)}, Vol 1. Singapore: World Scientific Publishing Company.

\item Chow, T. L., and J. L.Teugels. 1978. “The Sum and the Maximum of I.I.D. Random Variables”. In {\it Proceedings
of the Second Prague Symposium on Asymptotic Statistics}, edited by P. Mandl and M. Huskova, 81-92.
New York: North-Holland.

\item Dieker, A. B., and S.-H. Kim. 2012. “Selecting the Best by Comparing Simulated Systems in a Group
of Three When Variances are Known and Unequal”. In {\it Proceedings of the 2012 Winter Simulation
Conference}, edited by C. Laroque, J. Himmelspach, R. Pasupathy, O. Rose, and A. M. Uhrmacher,
1-7. Piscataway, New Jersey: Institute of Electrical and Electronics Engineers, Inc.

\item Dieker, A. B., and S.-H. Kim 2014. ``A Fully Sequential Procedure for Known and Equal Variances Based on Multivariate Brownian Motion''. In {\it Proceedings of the 2014 Winter Simulation Conference}, edited by A. Tolk, S. D. Diallo, I. O. Ryzhov, L. Yilmaz, S. Buckley, and J. A. Miller,  3749-3760. Piscataway, New Jersey: IEEE.

\item Embrechts, P., C. Kl¨uppelberg, and T. Mikosch. 1997. {\it Modelling Extremal Events for Insurance and
Finance}. New York: Springer.

\item Frazier, P. 2014. “A Fully Sequential Elimination Procedure for Indifference-Zone Ranking and Selection
with Tight Bounds on Probability of Correct Selection”. {\it Operations Research} 62(4):926-942.

\item Hong, L. J., B. L. Nelson, J. Xu. 2014. ``Discrete Optimization via Simulation''. In {\it Handbook of Simulation Optimization}, edited by M. C. Fu. Springer:NY.


\item Kim, S.-H., and A. B. Dieker. 2011. “Selecting the Best by Comparing Simulated Systems in a Group of
Three”. In {\it Proceedings of the 2011 Winter Simulation Conference}, edited by S. Jain, R. R. Creasey, J. Himmelspach, K. P. White, and M. Fu. 4217-4226. Piscataway, New Jersey: IEEE.

\item Kim, S.-H., and B. L. Nelson. 2001. “A Fully Sequential Procedure for Indifference-Zone Selection in
Simulation”. {\it ACM Transactions on Modeling and Computer Simulation} 11(3):251-273.

\item Kim, S.-H., and B. L. Nelson. 2006. ``On the Asymptotic Validity of Fully Sequential Selection Procedures for Steady-State Simulation''. {\it Operations Research} 54:475-488.

\item Nelson, B. L., J. Swann, D. Goldsman, and W. Song. 2001. ``Simple Procedures for Selecting the Best Simulated System when the Number of Alternatives is Large''. {\it Operations Research} 49(6):950-963.

\item Powell, W. B. and Ryzhov, I. O. 2012. ``Ranking and selection''. In Chapter 4 in {\it Optimal
Learning}, pages 71-88. John Wiley and Sons.

\item Rinott, Y. 1978. ``On two-stage selection procedures and related probability inequalities''. {\it Comm.\ Statist.-Theory and Methods} 7(8):799-811.

\item Rogers, L., and J. W. Pitman. 1981. “Markov Functions”. {\it The Annals of Probability} 9:573-582.

\item Waeber, R., P. I. Frazier, and S. G. Henderson. 2011. ``A Bayesian Approach to Stochastic Root Finding''. In {\it Proceedings of the 2011 Winter Simulation Conference}, edited by S. Jain, R. R. Creasey, J. Himmelspach, K. P. White, and M. Fu. 4038-4050. Piscataway, New Jersey: Institute of Electrical and Electronics Engineers, Inc.

\item Waeber, R. 2013. {\it Probabilistic Bisection Search for Stochastic Root-Finding}. PhD Dissertation. Cornell University, Ithaca, NY.

\item Wang, H., and S.-H. Kim. 2011. “Reducing the Conservativeness of Fully Sequential Indifference-Zone
Procedures”. {\it IEEE Transactions on Automatic Control} 58(6):1613-1619

\end{hangref}

\newpage

\setcounter{table}{0}
\renewcommand\thefigure{A.\arabic{figure}}
\renewcommand\thetable{A.\arabic{table}}
\section*{Appendix}

\begin{proof}[Proof of Lemma~\ref{lem:S}]
\begin{eqnarray*}
{\cal S}_I(\Pi x) & = & (V \Pi x)^T (V \Gamma V^T)^{-1} (V \Pi x)\\
                 &=& (V \Gamma V^T (V \Gamma V^T)^{-1}V x)^T (V \Gamma V^T)^{-1} (V \Gamma V^T (V \Gamma V^T)^{-1}V x)\\
                 &=& (Vx)^T (V \Gamma V^T)^{-1} (Vx) = {\cal S}_I(x).
\end{eqnarray*}

\end{proof}

\begin{proof}[Proof of Corollary~\ref{cor:S}]

We first derive an explicit expression for $(V \Gamma V^T)^{-1}$. Without loss of generality, assume that $I=\{1,\ldots,s\}$. Then by noting that $V\Gamma V^T$ is the covariance matrix of $Vx$, we get
\[
Vx = \left[
\begin{array}{c}
x_{1} - x_{s} \\
\vdots\\
x_{s-1} - x_s
\end{array}\right] \quad \quad \mbox{ and } \quad \quad
V\Gamma V^T = \left[ \begin{array}{ccccc}
    {\sigma_1^2 + \sigma_s^2} & \sigma_s^2 & \cdots  & \cdots & \sigma_s^2 \\
    \sigma_s^2 & \sigma_2^2 + \sigma_s^2 & \sigma_s^2 & \cdots & \sigma_s^2\\
    \vdots & & \ddots &  &\vdots \\
    \vdots & &  &  \ddots & \sigma_s^2 \\
   \sigma_s^2 & \cdots & \cdots &\sigma_s^2 & \sigma_{s-1}^2 + \sigma_s^2
\end{array}
\right].
\]

For equal variances,
\[
V\Gamma V^T = \sigma^2 \left[ \begin{array}{ccccc}
    2 & 1 & \cdots  & \cdots & 1 \\
    1 & 2 &1 & \cdots & 1\\
    \vdots & & \ddots &  &\vdots \\
    \vdots & &  &  \ddots & 1 \\
   1 & \cdots & \cdots &1 & 2
\end{array}
\right] = \sigma^2 \left( \id_{s-1} + \boldsymbol{1}_{s-1} \boldsymbol{1}_{s-1}^T\right)
\]
where $\id_{s}$ is the $s \times s$ identity matrix and $\boldsymbol{1}_s$ is the $s \times 1$ vector of ones.

By the Sherman-–Morrison formula,
\begin{eqnarray}
(V\Gamma V^T)^{-1} & = & {1 \over \sigma^2} \left( \id_{s-1}^{-1} - \frac{ \id_{s-1}^{-1} \boldsymbol{1}_{s-1} \boldsymbol{1}_{s-1}^T \id_{s-1}^{-1}}{1 + \boldsymbol{1}^T \id_{s-1}^{-1} \boldsymbol{1}} \right) \nonumber \\
&=& {1 \over \sigma^2} \left( \id_{s-1} - \frac{\boldsymbol{1}_{s-1} \boldsymbol{1}_{s-1}^T}{1 + (s-1)} \right) \nonumber \\
&=& {1 \over \sigma^2} {1 \over s} \left( s \cdot \id_{s-1} - \boldsymbol{1}_{s-1} \boldsymbol{1}_{s-1}^T \right). \label{eqn:inv}
\end{eqnarray}

Then we have
\begin{eqnarray*}
{\cal S}_I(x) & = & {1 \over \sigma^2} {1 \over s} \left[
\begin{array}{c}
x_{1} - x_{s} \\
\vdots\\
x_{s-1} - x_{s}
\end{array}\right]^T
\left( s \cdot \id_{s-1} - \boldsymbol{1}_{s-1} \boldsymbol{1}_{s-1}^T \right)
 \left[
\begin{array}{c}
x_{1} - x_{s} \\
\vdots\\
x_{s-1} - x_{s}
\end{array}\right] \\
&=& {1 \over \sigma^2}{1 \over s} \left\{ (s-1) \sum_{i=1}^{s-1}(x_i-x_s)^2 - 2 \sum_{1 \le i < \ell <s} (x_i-x_s)(x_\ell-x_s)\right\}\\
 & = &  {1 \over \sigma^2} {1 \over s} \sum_{i < \ell \atop i, \ell \in I} (x_i-x_\ell)^2,
\end{eqnarray*}
which shows the first equality in the corollary because $|I|=s$.

Now we show the second equality of the corollary.
From (\ref{eqn:inv}),
\[
V^T(V \Gamma V^T)^{-1} V = {1 \over \sigma^2}{1\over s} ( s \cdot \id_s - \boldsymbol{1}_s \boldsymbol{1}_s^T)
\quad \quad \mbox{and} \quad \quad
\Pi = \Gamma V^T (V \Gamma V^T)^{-1} V = {1 \over s} ( s \cdot  \id_s - \boldsymbol{1}_s \boldsymbol{1}_s^T).
\]

Then
\[
\Pi x = {1 \over s} ( s \cdot \id_s - \boldsymbol{1}_s \boldsymbol{1}_s^T) x = \left[
\begin{array}{c}
x_1 - \bar{x} \\
\vdots\\
x_s - \bar{x}
\end{array}\right].
\]

Finally,
\begin{eqnarray*}
{\cal S}_I(\Pi x)& =& (V \Pi x)^T(V \Gamma V^T)^{-1} (V \Pi x)\\
 &=& (\Pi x)^T [ V^T (V \Gamma V^T)^{-1} V] (\Pi x)\\
 &=& {1 \over \sigma^2} {1 \over s}\left[
\begin{array}{c}
x_1 - \bar{x} \\
\vdots\\
x_s - \bar{x}
\end{array}\right]^T
( s \cdot \id_s - \boldsymbol{1}_s \boldsymbol{1}_s^T)
\left[
\begin{array}{c}
x_1 - \bar{x} \\
\vdots\\
x_s - \bar{x}
\end{array}\right] \\
&=& {1 \over \sigma^2} \left[
\begin{array}{c}
x_1 - \bar{x} \\
\vdots\\
x_s - \bar{x}
\end{array}\right]^T\left[
\begin{array}{c}
x_1 - \bar{x} \\
\vdots\\
x_s - \bar{x}
\end{array}\right] \\
&=& {1 \over \sigma^2} \sum_{i=1}^s (x_i - \bar{x})^2.
\end{eqnarray*}

\end{proof}

\begin{proof}[Proof of Lemma~\ref{lem:Smon}]
It suffices to prove the claim for $|I| = |J| +1$.
By relabeling systems if necessary, it suffices to prove the claim with $J= \{1,\ldots,s\}$ and $I = \{1,\ldots,s+1\}$.
We set
\[
H_{s+1} = \left\{ (x_1, x_2, \ldots, x_{s+1})^T: \sum_{i=1}^{s+1} x_{i}=0 \right\},\quad\quad
Q_{s} = \left\{ (x_1, x_2, \ldots, x_{s+1})^T: \sum_{i=1}^{s} x_{i}=0, x_{s+1}=0 \right\}.
\]
By the second equality of Corollary~\ref{cor:S}, it suffices to show that for $x\in \R^{s+1}$,
\begin{equation}
\label{eq:Sineq1}
{\cal S}_I(x) \ge \frac 1{\sigma^2} \sum_{i=1}^s (x_i-\bar x_s)^2,
\end{equation}
where $\bar x_s = (x_1+\cdots+x_s)/s$.
To see that this holds, we define $\Psi_s$ on $H_{s+1}$ as the matrix that projects orthogonally on $Q_s$, i.e.,
$\Psi_s x = (x_1-\bar x_s,\ldots,x_s-\bar x_s,0)$.
By Lemma~\ref{lem:S} and (\ref{eqn:inv}), we have, for $x\in\R^{s+1}$,
\[
{\cal S}_I (x) =  {1 \over \sigma^2} {1 \over (s+1)} \left[
\begin{array}{c}
x_{1} - x_{s+1} \\
\vdots\\
x_{s} - x_{s+1}
\end{array}\right]^T
\left( (s+1) \cdot \id_{s} - \boldsymbol{1}_{s} \boldsymbol{1}_{s}^T \right)
 \left[
\begin{array}{c}
x_{1} - x_{s+1} \\
\vdots\\
x_{s} - x_{s+1}
\end{array}\right]
\]
This representation immediately yields that
\[
{\cal S}_I(\Psi_s x) = \frac 1{\sigma^2} \sum_{i=1}^s (x_i-\bar x_s)^2.
\]
Since projecting decreases any quadratic form, this establishes (\ref{eq:Sineq1}).
\end{proof}

\begin{table}[h!]
\caption{$\eta_{|I|}$ when $\alpha = 10\%$} \label{tab:eta}
\begin{center}
{\tiny\renewcommand{\arraystretch}{.8}
\resizebox{!}{.35\paperheight}{%
\begin{tabular}{|c|ccccccccc|} \hline
$|I|$	&	$k=64$	&	$k=32$	&	$k=16$	&	$k=8$	&	$k=7$	&	$k=6$	&	$k=5$	&	$k=4$	&	$k=3$	\\ \hline
64	&	6.05042	&		&		&		&		&		&		&		&		\\
63	&	6.79306	&		&		&		&		&		&		&		&		\\
62	&	7.07127	&		&		&		&		&		&		&		&		\\
61	&	7.24002	&		&		&		&		&		&		&		&		\\
60	&	7.35712	&		&		&		&		&		&		&		&		\\
59	&	7.44289	&		&		&		&		&		&		&		&		\\
58	&	7.50694	&		&		&		&		&		&		&		&		\\
57	&	7.55688	&		&		&		&		&		&		&		&		\\
56	&	7.59517	&		&		&		&		&		&		&		&		\\
55	&	7.62437	&		&		&		&		&		&		&		&		\\
54	&	7.64624	&		&		&		&		&		&		&		&		\\
53	&	7.66164	&		&		&		&		&		&		&		&		\\
52	&	7.67146	&		&		&		&		&		&		&		&		\\
51	&	7.67755	&		&		&		&		&		&		&		&		\\
50	&	7.67896	&		&		&		&		&		&		&		&		\\
49	&	7.67697	&		&		&		&		&		&		&		&		\\
48	&	7.67216	&		&		&		&		&		&		&		&		\\
47	&	7.66385	&		&		&		&		&		&		&		&		\\
46	&	7.65204	&		&		&		&		&		&		&		&		\\
45	&	7.63885	&		&		&		&		&		&		&		&		\\
44	&	7.62218	&		&		&		&		&		&		&		&		\\
43	&	7.60345	&		&		&		&		&		&		&		&		\\
42	&	7.58269	&		&		&		&		&		&		&		&		\\
41	&	7.55989	&		&		&		&		&		&		&		&		\\
40	&	7.53440	&		&		&		&		&		&		&		&		\\
39	&	7.50692	&		&		&		&		&		&		&		&		\\
38	&	7.47817	&		&		&		&		&		&		&		&		\\
37	&	7.44679	&		&		&		&		&		&		&		&		\\
36	&	7.41350	&		&		&		&		&		&		&		&		\\
35	&	7.37764	&		&		&		&		&		&		&		&		\\
34	&	7.34060	&		&		&		&		&		&		&		&		\\
33	&	7.30173	&		&		&		&		&		&		&		&		\\
32	&	7.26040	&	4.61250	&		&		&		&		&		&		&		\\
31	&	7.21664	&	5.16401	&		&		&		&		&		&		&		\\
30	&	7.17116	&	5.35848	&		&		&		&		&		&		&		\\
29	&	7.12400	&	5.46804	&		&		&		&		&		&		&		\\
28	&	7.07454	&	5.53579	&		&		&		&		&		&		&		\\
27	&	7.02218	&	5.57870	&		&		&		&		&		&		&		\\
26	&	6.96829	&	5.60356	&		&		&		&		&		&		&		\\
25	&	6.91162	&	5.61553	&		&		&		&		&		&		&		\\
24	&	6.85224	&	5.61724	&		&		&		&		&		&		&		\\
23	&	6.79024	&	5.61064	&		&		&		&		&		&		&		\\
22	&	6.72631	&	5.59634	&		&		&		&		&		&		&		\\
21	&	6.65867	&	5.57542	&		&		&		&		&		&		&		\\
20	&	6.58867	&	5.54845	&		&		&		&		&		&		&		\\
19	&	6.51518	&	5.51603	&		&		&		&		&		&		&		\\
18	&	6.43833	&	5.47773	&		&		&		&		&		&		&		\\
17	&	6.35826	&	5.43468	&		&		&		&		&		&		&		\\
16	&	6.27395	&	5.38576	&	3.55536	&		&		&		&		&		&		\\
15	&	6.18617	&	5.33234	&	3.96549	&		&		&		&		&		&		\\
14	&	6.09394	&	5.27360	&	4.09377	&		&		&		&		&		&		\\
13	&	5.99752	&	5.20970	&	4.15224	&		&		&		&		&		&		\\
12	&	5.89658	&	5.13991	&	4.17580	&		&		&		&		&		&		\\
11	&	5.79086	&	5.06446	&	4.17568	&		&		&		&		&		&		\\
10	&	5.68014	&	4.98455	&	4.15873	&		&		&		&		&		&		\\
9	&	5.46038	&	4.84321	&	4.13041	&		&		&		&		&		&		\\
8	&	5.26699	&	4.68543	&	4.02733	&	2.83446	&		&		&		&		&		\\
7	&	5.05352	&	4.50849	&	3.90131	&	3.05966	&	2.66510	&		&		&		&		\\
6	&	4.80933	&	4.29929	&	3.74380	&	3.04936	&	2.85635	&	2.47348	&		&		&		\\
5	&	4.52132	&	4.04717	&	3.54574	&	2.95465	&	2.81492	&	2.62431	&	2.25053	&		&		\\
4	&	4.16553	&	3.73733	&	3.28628	&	2.78163	&	2.67081	&	2.53302	&	2.34537	&	1.98200	&		\\
3	&	3.67921	&	3.30305	&	2.91240	&	2.49134	&	2.40324	&	2.29859	&	2.16417	&	1.97851	&	1.63182	\\
2	&	2.81738	&	2.52053	&	2.21620	&	1.89611	&	1.83132	&	1.75468	&	1.66093	&	1.54027	&	1.37146	\\ \hline
\end{tabular}}}
\end{center}
\end{table}

\end{document}